\documentclass[a4paper,10pt]{article}
\usepackage{amsmath}
\usepackage{amssymb}
\usepackage{graphicx}
\usepackage{appendix}
\usepackage{color}	
\usepackage{times}
\usepackage{a4wide}
\usepackage{authblk}

\newcommand{\qed}{\hfill$\Box$\smallskip}

\newcommand{\remove}[1]{}

\newtheorem{claim}{Claim}

\newtheorem{proposition}{Proposition}

\newtheorem{corollary}{Corollary}

\newtheorem{definition}{Definition}

\newtheorem{lemma}{Lemma}

\newtheorem{theorem}{Theorem}

\newenvironment{proof}{\noindent{\bf Proof\@:}}{\hfill $\Box$\\}

\newenvironment{theoremproof}[1]{\noindent{\bf Proof of Theorem #1\@:}}{\hfill $\Box$\\}
\newenvironment{lemmaproof}[1]{\noindent{\bf Proof of Lemma #1\@}}{\hfill $\Box$\\}
\newenvironment{claimproof}[1]{\noindent{\bf Proof of Claim #1\@:}}{\hfill $\Box$\\}
\newenvironment{propositionproof}[1]{\noindent{\bf Proof of Proposition #1\@:}}{\hfill $\Box$\\}

\newcommand\dist{\mathrm{dist}}

\newcommand\cA{\mathcal{A}} 
\newcommand\cB{\mathcal{B}} 
\newcommand\cC{\mathcal{C}}

\newcommand\cG{\mathcal{G}} 
\newcommand\cE{\mathcal{E}} 
\newcommand\cU{\mathcal{U}}

\newcommand\cH{\mathcal{H}} 
\newcommand\cS{\mathcal{S}} 
 
\newcommand\cI{\mathcal{I}}

\newcommand\cP{\mathcal{P}}

\newcommand\cZ{\mathcal{Z}}

\def\cC{{\mathcal C}}
\def\cE{{\cal E}}

\newcommand\eul{\mathrm{e}} 
\newcommand\eps{\varepsilon}

\newcommand\Erw{\mathrm{E}} 
\newcommand\pr{\mathrm{P}}

\newcommand{\Po}{{\rm Po}} 
\newcommand{\Bin}{{\rm Bin}}

\newcommand{\bink}[2] {{{#1}\choose {#2}}}

\newcommand\ra{\rightarrow} 

\newcommand\bc[1]{\left({#1}\right)} 
\newcommand\cbc[1]{\left\{{#1}\right\}} 
\newcommand\bcfr[2]{\bc{\frac{#1}{#2}}} 
 
\newcommand\brk[1]{\left\lbrack{#1}\right\rbrack} 
 
\newcommand\norm[1]{\left\|{#1}\right\|} 
\newcommand\abs[1]{\left|{#1}\right|}

\newcommand{\Whp}{W.h.p.} 
\newcommand{\whp}{w.h.p.} 
\newcommand{\stacksign}[2]{{\stackrel{\mbox{\scriptsize #1}}{#2}}}

\newcommand{\Erdos}{Erd\H{o}s}

\newcommand{\Bollobas}{Bollob\'as}

\newcommand{\Kucera}{Ku\v{c}era}

\newcommand\Lem{Lemma}
\newcommand\Prop{Proposition}
\newcommand\Thm{Theorem}
\newcommand\Cor{Corollary}
\newcommand\Sec{Section}

\date{\today}

\title{\bf On independent sets in random 
graphs\thanks{Supported by EPSRC grant EP/G039070/2 and DIMAP.} 
\thanks{A preliminary version of this work appeares in the proceedings of
ACM-SIAM SODA 2011}}

\author[1]{Amin Coja-Oghlan}
\author[2]{Charilaos Efthymiou}
\affil[1]{Goethe University, Mathematics Dept., Frankfurt 60054, Germany\\
\texttt{acoghlan@math.uni-frankfurt.de}}
\affil[2]{University of Warwick, Mathematics and Computer Science, Coventry CV4 7AL, UK\\
\texttt{c.efthymiou@warwick.ac.uk}}

\begin{document}
\maketitle
\thispagestyle{empty}

\begin{abstract}
The independence number of a sparse random graph $G(n,m)$ of average degree $d=2m/n$ is well-known to be
	$(2-\eps_d)n\ln(d)/d\leq\alpha(G(n,m))\leq(2+\eps_d)n\ln(d)/d$ with high probability,
		with $\eps_d\ra0$ in the limit of large $d$.
Moreover, a trivial greedy algorithm \whp\ finds an independent set of size $n\ln(d)/d$, i.e.,
about half the maximum size.
Yet in spite of 30 years of extensive research no efficient algorithm has emerged to
produce an independent set with $(1+\eps)n\ln(d)/d$ for any \emph{fixed} $\eps>0$ (independent of both $d$ and $n$).
In this paper we prove that the combinatorial structure of the independent set problem in random
graphs undergoes a phase transition as the size $k$ of the independent sets passes
the point $k\sim n\ln(d)/d$.
Roughly speaking, we prove that independent sets of size $k>(1+\eps)n\ln(d)/d$ form an 
intricately ragged landscape, in which local search algorithms seem to get stuck.
We illustrate this phenomenon by providing an exponential lower bound for the Metropolis
process, a Markov chain for sampling independents sets.

\end{abstract}

\setcounter{page}{1}
\section{Introduction and Results}
\subsection{Probabilistic analysis  and the independent set problem}

In the early papers on the subject, the motivation for the probabilistic
analysis of algorithms was to alleviate the glum of worst-case
analyses by establishing a brighter `average-case' scenario~\cite{AverAlg1,AverAlg2,AverAlg3}.
This optimism was stirred by early analyses of simple, greedy-type
algorithms, showing that these perform rather well on randomly
generated input instances, at least for certain ranges of the
parameters. Examples of such analyses include Grimmett and
McDiarmid~\cite{grimmett} (independent set problem), Wilf
\cite{Wilf}, Achlioptas and Molloy~\cite{AchListColouring} 
(graph coloring), and Frieze and Suen~\cite{FrSu} ($k$-SAT). Yet, remarkably, in spite of 30
years of research, for many problems no efficient algorithms,
howsoever sophisticated, have been found to outperform those simple
greedy algorithms markedly.

The independent set problem in random graphs $G(n,m)$ is a case
in point. Recall that $G(n,m)$ is a graph on $n$ vertices obtained
by choosing $m$ edges uniformly at random (without replacement).
We say that $G(n,m)$ has a property \emph{with high probability}
if the probability that the property holds tends to 1 as $n\rightarrow\infty$.
One of the earliest results in the theory of random graphs is a
non-constructive argument showing that for $m=\frac12\bink{n}2$
the independence number of $G(n,m)$ is $\alpha(G(n,m))\sim2\log_2(n)$
\whp~\cite{BollobasIS,erdos-1st,matula}. Grimmett and McDiarmid~\cite{grimmett}
analysed a simple algorithm that just constructs an inclusion-maximal
independent set greedily on $G(n,m)$: it yields an independent set
of size $(1+o(1))\log_2n$ \whp, about half the maximum size. But no
algorithm is known to produce an independent set of size $(1+\eps)
\log_2n$ for any fixed $\eps>0$ in polynomial time with a non-vanishing
probability, neither on the basis of a rigorous analysis, nor on
the basis of experiments or other evidence. In fact, devising such
an algorithm is probably the most prominent open problem in the
algorithmic theory of random graphs~\cite{friezeRSA,KarpOLD}. (However,
note that one can find a maximum independent set \whp\ by trying all
$n^{O(\ln n)}$ possible sets of size $2\log_2n$.)

Matters are no better on sparse random graphs. If we let $d=2m/n$
denote the average degree, then non-constructive arguments yield
	$$\alpha(G(n,m))\sim\frac{2\ln(d)}d\cdot n$$
for $1\ll d=o(n)$.
In the case $d\gg\sqrt n$, the proof of this is via a simple
second moment argument~\cite{BollobasIS,matula}. By contrast, for
$1\ll d\ll\sqrt n$, the second moment argument breaks down and
additional methods such as large deviations inequalities are 
needed~\cite{frieze-is}. Yet in either case, no algorithm is known
to find an independent set of size $(1+\eps)\frac{\ln d}d\cdot n$ in 
polynomial time with a non-vanishing probability, while `greedy' yields
an independent set of size $(1+o(1))\frac{\ln d}d\cdot n$ \whp\
In the sparse case, the time needed for exhaustive search scales
as $\exp(\frac{2n}{d}\ln^2(d))$, i.e., the complexity grows as
$d$ decreases.

The aim of this paper is to explore the tenacity
of finding large independent sets in random graphs. The focus is on
the sparse case, both conceptually and computationally  the most
difficult case. We exhibit a phase transition in the structure of the
problem that occurs as the size of the independent sets passes the
point $\frac{\ln d}d\cdot n$ up to which efficient algorithms are
known to succeed. Roughly speaking, we show that independent sets of
sizes bigger than $(1+\eps)\frac{\ln d}d\cdot n$ form an intricately
ragged landscape, which plausibly explains why local-search algorithms get
stuck. Thus, ironically, instead of showing that the `average case'
scenario is brighter, we end up suggesting that random graphs provide
an excellent source of difficult examples. Taking into account the (substantially)
different nature of the independent set problem, our work complements
the results obtained in~\cite{AchCoOg} for random constraint satisfaction
problem such as $k$-SAT or graph coloring.

\subsection{Results}\label{sec:Results}

Throughout the paper we will be dealing with sparse random graphs 
where the average degree $d=2m/n$ is `large' but remains bounded
as $n\rightarrow\infty$. To formalise this sometimes we work 
with functions $\eps_d$ that tend to zero as $d$ gets large.%
	\footnote{The reason why we need to speak about $d$ `large' is
	that 	the sparse random graph $G(n,m)$ is not connected. This
	implies, for instance, that algorithms can find independent sets
	of size $(1+\eps_d)n\ln(d)/d$ for some $\eps_d\ra0$ by optimizing
	carefully over the small tree components of $G(n,m)$. Our
	results/proofs actually carry over to the case that $d=d(n)$ tends
	to infinity as $n$ grows, but to keep matters as simple as possible,
	we will	confine ourselves to fixed $d$.}
Thus $\alpha(G(n,m))=(2-\eps_d)\frac{\ln d}d\cdot n$ and the greedy
algorithm finds independent sets of size $(1+\eps'_d)\frac{\ln d}d
\cdot n$ \whp, where $\eps_d,\eps_d'\ra 0$. However, no efficient
algorithm is known to find independent sets of size $(1+\eps'')
\frac{\ln d}d\cdot n$ for any \emph{fixed} $\eps''>0$.

For a graph $G$ and an integer $k$ we let $\cS_k(G)$ denote the set
of all independent sets in $G$ that have size exactly $k$. What we
will show is that in $G(n,m)$ the set $\cS_k(G(n,m))$ undergoes a
phase transition as $k\sim\frac{\ln d}dn$. For two sets $S,T\subset
V$ we let $S\triangle T$ denote the symmetric difference of $S,T$.
Moreover, $\dist(S,T)=|S\triangle T|$ is the Hamming distance of
$S,T$ viewed as vectors in $\cbc{0,1}^V$.

To state the result for $k$ smaller than $\frac{\ln d}dn$, we need
the following concept. Let $\cS$ be a set of subsets of $V$, and
let $\gamma>0$ be an integer. We say that $\cS$ is {\em $\gamma$-connected}
if for any two sets $\sigma,\tau\in\cS$ there exist $\sigma_1,\ldots,\sigma_N\in\cS$
such that $\sigma_1=\sigma$, $\sigma_N=\tau$, and $\dist(\sigma_t,\sigma_{t+1})\leq \gamma$
for all $1\leq t<N$. If $\cS_k(G(n,m))$ is $\gamma$-connected for
some $\gamma=O(1)$, one can easily define various simple Markov
chains on $\cS_k(G)$ that are ergodic.

\begin{theorem}\label{thrm:Connectivity}
There exists $\eps_d\ra0$ and for any $d$ a number $C_d>0$ (independent
of $n$) such that $\cS_k(G(n,m))$ is $C_d$-connected \whp\
for any
	$$k\leq(1-\eps_d)\frac{\ln d}d\cdot n.$$
\end{theorem}

\noindent
The proof of \Thm~\ref{thrm:Connectivity} is `constructive' in the
following sense. Suppose given $G=G(n,m)$ we set up an auxiliary
graph whose vertices are the independent sets $\cS_k(G)$ with
$k\leq (1-\eps_d)\frac{\ln d}d\cdot n$. In the auxiliary graph
two independent sets $\sigma,\tau\in\cS_k(G)$ are adjacent if
$\dist(\sigma,\tau)\leq C_d$. Then the proof of \Thm~\ref{thrm:Connectivity}
yields an algorithm  for finding paths of length $O(n)$ between
any two elements of $\cS_k(G)$ \whp\ Thus, intuitively \Thm~\ref{thrm:Connectivity}
shows that for $k\leq(1-\eps_d)\frac{\ln d}d\cdot n$ the set
$\cS_k(G(n,m))$ is easy to `navigate' \whp

By contrast, our next result shows that for $k>(1+\eps_d)\frac{\ln d}d\cdot n$
the set $\cS_k(G(n,m))$ is not just disconnected \whp, but that
it shatters into exponentially many, exponentially tiny pieces.

\begin{definition}\label{Def_shattering}
We say that $\cS_k(G(n,m))$ {\bf shatters} if there exist constants
$ \gamma , \zeta>0$ such that \whp\ the set ${\cal S}_k(G(n,m))$ admits a
partition into subsets such that
\begin{enumerate}
\item Each subset contains at most $\exp\bc{-\gamma n}\abs{{\cal S}_k(G(n,m))}$
		independent sets.
\item If $\sigma,\tau$ belong to different subsets, then
	$\dist(\sigma,\tau)\geq\zeta n$.
\end{enumerate}
\end{definition}

\begin{theorem}\label{theorem:shattering}
There is $\eps_d\ra0$ so that $\cS_k(G(n,m))$ shatters for all
$k$ with
\begin{displaymath}
(1+\eps_d)\frac{\ln d}d\cdot n\leq k\leq(2-\eps_d)\frac{\ln d}d\cdot n.
\end{displaymath}
\end{theorem}

\noindent
\Thm s~\ref{thrm:Connectivity} and~\ref{theorem:shattering} deal
with the geometry of a single `layer' $\cS_k(G(n,m))$ of independents
of a specific size. The following two results explore if/how a 
`typical' independent set in $\cS_k(G(n,m))$ can be extended 
to a larger one. To formalize the notion of `typical', we let
$\Lambda_k(n,m)$ signify the set of all pairs $(G,\sigma)$, where
$G$ is a graph on $V=\cbc{1,\ldots,n}$ with $m$ edges and
$\sigma\in {\cal S}_k(G)$. Let $\cU_k(n,m)$ be the probability
distribution on $\Lambda_k(n,m)$  induced  by the following 
experiment.
\begin{quote}
Choose a graph $G=G(n,m)$ at random.\\
If $\alpha(G)\geq k$,
choose an independent set $\sigma\in\cS_k(G)$ uniformly at random
and output $(G,\sigma)$.
\end{quote}
We say a pair $(G,\sigma)$ chosen from the distribution $\cU_k(n,m)$
has a property $\cP$ \emph{with high probability} if the probability
of the event $\cbc{(G,\sigma)\in\cP}$ tends to one as $n\ra\infty$.

\begin{definition}
Let $\gamma,\delta\geq0$, let $G$ be a graph, and let $\sigma$ be
an independent set of $G$. We say that $(G,\sigma)$ is {\bf $(\gamma,
\delta)$-expandable} if $G$ has an independent set $\tau$ such that 
$|\tau|\geq (1+\gamma)|\sigma|$ and
$|\tau \cap \sigma|\geq(1-\delta)|\sigma|$.
\end{definition}

\begin{theorem}\label{theorem:non-maximal}
There are $\eps_d, \delta_d\ra0$ such that for any $\eps_d\leq \eps\leq1-\eps_d$
the following is true. For 	$k=(1-\eps)\frac{\ln d}d\cdot n$ a
pair $(G,\sigma)$ chosen from the distribution $\cU_k(n,m)$ is
$((2-\delta_d)\eps/(1-\eps),0)$-expandable \whp
\end{theorem}

\noindent
\Thm~\ref{theorem:non-maximal} shows that  \whp\ in a random graph
$G(n,m)$ almost all independent sets of size $k=(1-\eps)\frac{\ln d}
d\cdot n$ are contained in \emph{some} bigger independent set of size
$(1+\eps)\frac{\ln d}d\cdot n$. That is, they can be expanded beyond
the critical size $\frac{\ln d} d\cdot n$ where shattering occurs.
However, as $k$ approaches the critical size $\frac{\ln d}d\cdot n$,
i.e., as $\eps\ra0$, the typical potential for expansion diminishes.

\begin{theorem}\label{theorem:maximal}
There is $\eps_d\ra0$ such that for any $\eps$ satisfying $\eps_d\leq \eps\leq1-\eps_d$
and 	$k=(1+\eps)\frac{\ln d}d\cdot n$ \whp\ a pair $(G,\sigma)$
chosen from the distribution $\cU_k(n,m)$ is not $(\gamma,\delta)$-expandable
for any $\gamma>\eps_d$ and 
$$\delta<\gamma+\frac{2(\eps-\eps_d)}{1+\eps}.$$
\end{theorem}

\noindent
In other words, \Thm~\ref{theorem:maximal} shows that for $k=
(1+\eps)\frac{\ln d}d\cdot n$, a typical $\sigma\in\cS_k(G(n,m))$
cannot be expanded to an independent set of size $(1+\gamma)k$,
$\gamma>\eps_d$ without first \emph{reducing} its size below 
$$(1-\delta)k=(1-\eps-\gamma(1+\eps)+2\eps_d)\frac{\ln d}d\cdot
n<\frac{\ln d}d\cdot n.$$ (However,  a random independent set of
size $k\leq(2-\eps_d)\ln(d)n/d$ is typically not inclusion-maximal
because, for instance, it is unlikely to contain \emph{all} isolated
vertices of the random graph $G(n,m)$.)

Metaphorically, the above results show that 
\whp\ the independent sets of $G(n,m)$ form 
a ragged mountain range.
 Beyond the `plateau level'
$k\sim\frac{\ln d}d\cdot n$ there is an abundance of smaller `peaks',
i.e., independent sets of sizes $(1+\eps)k$ for any $\eps_d<\eps<1-\eps_d$,
almost all of which are
not expandable (by much).

The algorithmic equivalent of a mountaineer aiming to ascend
to the highest summit is a Markov chain
called the \emph{Metropolis process}, \cite{KirkpatricMetropolis,Metropolis1st}.
For a given graph $G$ its state space is the set of all independent
sets of $G$. Let $I_t$ be the state at time $t$. In step $t+1$, the
chain chooses a vertex $v$ of $G$ uniformly at random. If $v\in I_t$,
then with probability $1/\lambda$ the next state is $I_{t+1}=I_t
\setminus\{v\}$, and with probability $1-1/\lambda$ we let $I_{t+1}
=I_t$, where $\lambda\geq1$ is a `temperature' parameter. If $v\not
\in I_t\cup N(I_t)$ (with $N(I_t)$ the neighbourhood of $I_t$), then
$I_{t+1}=I_t\cup\cbc v$. Finally, if $v\in N(I_t)$, then $I_{t+1}=
I_t$.
It is well know that the probability of an independent set $S$ of $G$ in the stationary
distribution equals $\lambda^{|S|}/Z(G,\lambda)$, where
$$
Z(G,\lambda)=\sum_{k=0}^n \lambda^k\cdot\abs{\cS_k(G)}
$$
is the partition function.
Hence, the larger $\lambda$, the higher the mass of large independent
sets. Let
$$\mu(G,\lambda)=\sum_{k=0}^n k{\lambda}^k\cdot\abs{\cS_k(G)}/Z(G,\lambda)$$
denote the average size of an independent set of $G$ under the
stationary distribution.

It is easy to see that every state in $\Omega=\bigcup_kS_k(G(n,m))$ 
communicates with every other in the Metropolis process.
Thus the process is ergodic and possesses a unique stationary 
distribution. Let $\pi:\Omega\to [0,1]$ denote the stationary 
distribution of the Metropolis process with parameter $\lambda$, 
for some $\lambda>0$. It is well known 
that $\pi(\sigma)={\lambda}^{|\sigma|}/Z$  
where $Z=\sum_{\sigma\in \Omega}\lambda^{|\sigma|}$ (e.g. \cite{jerrum-planted}).

Here, we are interested in finding the rate at which the Metropolis
process converges to its equilibrium. There are a number of ways
of quantifying the closeness to stationarity. Let $P^t(\sigma,\cdot)
:\Omega\to [0,1]$ denote the distribution of the state at time $t$
given that $\sigma$ was the initial state. The {\em total variation
distance} at time $t$ with respect to the initial state $\sigma$
is 
\begin{displaymath}
\Delta_{\sigma}(t)=\max_{S \subset \Omega}|P^t(\sigma, S)-\pi(S)|=\frac{1}{2}
\sum_{\tau \in \Omega}|P^t(\sigma, \tau)-\pi(\tau)|.
\end{displaymath}
Starting from $\sigma$, the rate of convergence 
to stationarity may then be measured by the function
\begin{displaymath}
\tau_{\sigma}=\min_t\{\Delta_{\sigma}(t')<e^{-1} \textrm{for all $t'>t$}\}.
\end{displaymath}

\noindent
The {\bf mixing time} of the Metropolis process  is defined as
\begin{displaymath}
T=\max_{\sigma\in \Omega}\tau_{\sigma}.
\end{displaymath}

Our above results on the structure of the sets $\cS_k(G(n,m))$
imply that \whp\ the mixing time of the Metropolis process is
exponential if the parameter $\lambda$ is tuned so that the
Metropolis process tries to ascend to independent sets bigger
than $(1+\epsilon_d)\frac{\ln d}d\cdot n$.

\begin{corollary}\label{cor:MixingTimeBound}
There is $\eps_d\ra0$ such that for $\lambda>1$ with
\begin{equation}\label{eq:themuglambda}
	(1+\eps_d)\frac{\ln d}d\cdot n\leq\Erw\brk{\mu(G(n,m),\lambda)}\leq(2-\eps_d)\frac{\ln d}d\cdot n.
\end{equation}
the mixing time of the Metropolis process on $G(n,m)$ is
$\exp(\Omega(n))$ \whp
\end{corollary}

\subsection{Related work}
To our knowledge, the connection between transitions in the geometry
of the `solution space' (in our case, the set of all independent sets of a given size)
and the  apparent failure of {\em local algorithms}
in finding a solution has been pointed first out in the statistical mechanics 
literature~\cite{FuAnderson,MPZ-Science,1RSBPaper}.
In that work, which mostly deals with CSPs such as $k$-SAT,
the shattering phenomenon goes by the name of `dynamic replica symmetry breaking.'
Our present work is clearly inspired by the statistical mechanics ideas,
although we are unaware of explicit contributions from that line of work addressing the
independent set problem in the case of random graphs with average degree $d\gg1$.
Generally, the statistical mechanics work is based on deep, insightful, but, alas, 
mathematically non-rigorous techniques.

In the case that the average degree $d$ satisfies $d\gg\sqrt n$,
the independent set problem in random graphs is conceptually
somewhat simpler than in the case of $d=o(\sqrt n)$.
The reason for this is that for $d\gg\sqrt n$
the second moment method can be used to show that the \emph{number}
of independent sets is concentrated about its mean. As we will see
in \Cor~\ref{cor:transfer-theorem} below, this is actually untrue for
sparse random graphs.

The results of the present paper extend the main results from Achlioptas
and Coja-Oghlan~\cite{AchCoOg}, which dealt with constraint satisfaction
problems such as $k$-SAT or graph coloring, to the independent set
problem. This requires new ideas, because the natural questions are
somewhat different (for instance, the concept of `expandability'
has no counterpiece in CSPs). Furthermore, in~\cite{AchCoOg} we
conjectured but did not manage to prove the counterpiece of
\Thm~\ref{thrm:Connectivity} on the connectivity of $\cS_k(G(n,m))$.
On a technical level, we owe to~\cite{AchCoOg} the idea of analysing
the distribution $\cU_k(n,m)$ via a different distribution $\cP_k(n,m)$,
the so-called `planted model' (see \Sec~\ref{sec:PlantedModel}
for details). However, the proof that this approximation is indeed
valid (\Thm~\ref{thrm:transfer-theorem} below) requires a rather
different approach. In~\cite{AchCoOg} we derived the corresponding
result from the second moment method in combination with sharp
threshold results. By contrast, here we use an indirect approach
that reduces the problem of estimating the number $|\cS_k(G(n,m))|$
of independent sets of a given size to the problem of (very accurately)
estimating the independence number $\alpha(G(n,m))$. Indeed, the
argument used here carries over to other problems, particularly
random $k$-SAT, for which it yields a conceptually simpler proof
than given in~\cite{AchCoOg} (details omitted).

Subsequently to~\cite{AchCoOg}, it was shown in~\cite{MRT} that
in many random CSPs the threshold for the shattering of the solution
space into exponentially small components coincides asymptotically
with the \emph{reconstruction threshold}. Roughly speaking, the
reconstruction threshold marks the onset of long-range correlations
in the Gibbs measure. More precisely, it is shown in~\cite{MRT} that
for a class of `symmetric' random CSPs the reconstruction threshold
derives from the corresponding threshold on random trees, and that
it happens to coincide with the shattering threshold. Our
\Thm~\ref{theorem:shattering} determines the threshold for shattering
in the independent set problem in random graphs. Furthermore,
Bhatnagar, Sly, and Tetali~\cite{BST10} recently studied the
reconstruction problem for the independent set problem on $k$-regular
trees. It would be most interesting to obtain a result
similar to~\cite{MRT}, namely that the reconstruction threshold
on the $G(n,m)$ random graph is given by the reconstruction threshold
on trees and that it coincides with the shattering threshold from
\Thm~\ref{theorem:shattering}.

The work that is perhaps most closely related to ours is  a remarkable
paper of Jerrum~\cite{jerrum-planted}, who studied the Metropolis
process on random graphs $G(n,m)$ with average degree $d=2m/n>n^{2/3}$.
The main result is that \whp\  \emph{there exists} an initial
state from which the expected time for the Metropolis process to
find an independent set of size $(1+\eps)\frac{\ln d}d\cdot n$ is
superpolynomial.
This is quite a non-trivial achievement, as it is a result about the
\emph{initial} steps of the process where the states might potentially
follow a very different distribution than the stationary distribution.
The proof of this fact is via a concept called `gateways', which is
somewhat reminiscent of the expandability property in the present work.
However, Jerrum's proof hinges upon the fact that the number of
independent sets of size $k\sim(1+\eps)\frac{\ln d}d\cdot n$ is
concentrated about its mean. The techniques from the present work
(particularly \Thm~\ref{thrm:transfer-theorem} below) can be used
to extend Jerrum's result to the sparse case quite easily, showing
that the expected time until a large independent set is found is fully
exponential in $n$ \whp\ Yet as also pointed out in~\cite{jerrum-planted},
an unsatisfactory aspect of this type of result is that it only shows that
\emph{there exists} a `bad' initial state, 	while it seems natural to
conjecture that indeed most specific initial states (such as the empty set)
are `bad'. Since we are currently unable to establish such a
stronger statement, we will confine ourselves to proving an exponential
lower bound on the mixing time (\Cor~\ref{cor:MixingTimeBound}).

For \emph{extremely} sparse random graphs, namely $d<\eul\approx2.718$,
finding a maximum independent set in $G(n,m)$ is easy. More
specifically, the greedy matching algorithm of Karp and
Sipser~\cite{KarpSipser} can easily be adapted so that it
yields a maximum independent set \whp\ But this approach does
not generalize to average degrees $d>\eul$ (see, however,
\cite{Gamarnik} for a particular type of weighted independent
sets).

Recently Rossman~\cite{monotone-circuit} obtained a monotone circuit
lower bound for the clique problem on random graphs that is exponential
in the size of the clique. The setup of~\cite{monotone-circuit} is
somewhat orthogonal to our contribution, as we are concerned with the
case that the size of the desired object (i.e., the independent set)
is linear in the number of vertices, while~\cite{monotone-circuit}
deals with the case that the size of the clique is $O(1)$ in terms
of the order of the graph. Nevertheless, the punchline of viewing
random graphs as a potential source of hard problem is similar.

In the course of the analysis in this paper we need a lower bound
on $\alpha(G(n,m))$ which is bigger what is calculated in \cite{frieze-is}.
For this reason, in \cite{arxivTR}, a previous version of this
work, we improved slightly on the value of $\alpha(G(n,m))$. The
analysis is similar to that in \cite{frieze-is}, i.e. combine
vanilla second moment with Talagrand's inequality. 
A bit later our result was improved even more by Dani and Moore
\cite{W2ndM}. Raughly speaking, the authors show that a $G(n,m)$
of expected degree $d\leq 2(n/k) \ln(n/k) + 2(n/k) -O(\sqrt{n/k})$
has an independent set of size $k$ w.h.p.
In comparison to \cite{W2ndM}, our bound on $d$ in \cite{arxivTR}
is $d<2(n/k)(\ln (n/k)+1)-O(\sqrt{\ln(n/k)\cdot (n/k)})$. To absolve
our work from the tendious second moment calculations we make direct
use of the result \cite{W2ndM}.

\subsection{Organisation of the paper}

The remaining material of this work is organised as follows:
For completeness, in Section \ref{sec:Prelim} we provide some 
very elementary results, which are either known  or easy to
derive. 
In Section \ref{sec:PlantedModel} we analyse the so-called `planted
model' to approximate the distribution $\cU_k(n,m)$. 
Then in Section \ref{section:thrm:Connectivity} we show Theorem
\ref{thrm:Connectivity}. In Section \ref{sec:theorem:shattering}
we show Theorem \ref{theorem:shattering}. In Section \ref{section:theorem:non-maximal}
we show Theorem \ref{theorem:non-maximal}. In Section
\ref{section:theorem:maximal} we show  Theorem \ref{theorem:maximal}.
In Section \ref{sec:cor:MixingTimeBound} we show Corollary \ref{cor:MixingTimeBound}.

\section{Preliminaries and notation}\label{sec:Prelim}

\noindent
We will need the following Chernoff bounds on the tails of a sum
of independent Bernoulli variables~\cite{MotwRandAlgBook}.

\begin{theorem}\label{chernoffbounds}
Let $I_1, I_2\ldots, I_n$ be independent Bernoulli variables.
Let $X=\sum_{i=1}^{n}I_i$ and $\mu=E[X]$.
Then
\begin{eqnarray}
Pr[X<(1-\delta)\mu]&\leq&\exp\left( -\mu\delta^2/2 \right)\qquad\mbox{ for any $0<\delta\leq 1$, and}\label{eq:ChernoffLB}\\
Pr[X>(1+\delta)\mu]&\leq&\exp\left( -\mu\delta^2/4 \right)\qquad\mbox{ for any $0<\delta<2e-1$}\label{eq:ChernoffUB}.
\end{eqnarray}
\end{theorem}

Let $G^*(n,m)$ be random graph on $n$ vertices obtained as follows:
choose $m$ pairs of vertices independently out of all $n^2$ possible
pairs; insert the $\leq m$ edges induced by these pairs, omitting
self-loops and replacing multiple edges by single edges. For technical
reasons it will sometimes be easier to first work with $G^*(n,m)$
and  then transfer the results to $G(n,m)$. The two distributions
are related as follows.

\begin{lemma}\label{lemma:model-equivalence}
For any fixed $c>0$ and $m=cn$ we have
\begin{displaymath}
Pr[G(n,m)\in\cA]\leq (1+o(1))\exp(c+c^2)\cdot 
		Pr[G^*(n,m)\in\cA]\qquad\mbox{for any event $\cA$}.
\end{displaymath}
\end{lemma}
\begin{proof}
This is a standard counting argument. The random graph $G^*(n,m)$ is
obtained by choosing one of the $n^{2m}$ possible sequences of vertex
pairs uniformly at random. Out of these $n^{2m}$ sequences, precisely
$2^m\bink n2_m$ sequences induce simple graphs with $m$ edges (where
$\bc\cdot_m$ denotes the falling factorial). Indeed, each of the
$\bink{\bink n2}m$ simple graph with $m$ edges can be turned into
a sequence of pairs by ordering the edges arbitrarily (a factor $m!$),
and then choosing for each edge in which order its vertices appear in
the sequence (a factor $2^m$). Hence, letting $\Sigma$ denote the event
that $G^*(n,m)$ is a simple graph with $m$ edges, we see that
	\begin{eqnarray}\nonumber
	Pr[G^*(n,m)\in\Sigma]&=&\frac{2^m\bink n2_m}{n^{2m}}
		=\bcfr 2{n^2}^m\cdot\prod_{j=0}^{m-1}\bink n2-j
			=\prod_{j=0}^{m-1}1-\frac1n-\frac{2j}{n^2}\\
		&=&\exp\brk{\sum_{j=0}^m\ln\bc{1-\frac1n-\frac{2j}{n^2}}}\nonumber\\
		&\sim&\exp\brk{-\sum_{j=0}^m\frac1n+\frac{2j}{n^2}}\qquad\mbox{[using $\ln(1-x)=-x+O(x^2)$ as $x\ra0$]}\nonumber\\
		&\sim&\exp\brk{-c-c^2}.
			\label{eqGnm*}
	\end{eqnarray}
Furthermore, given that the event $\Sigma$ occurs, $G^*(n,m)$ is
just a uniformly distributed (simple) graph with $m$ edges. 
Therefore,  (\ref{eqGnm*}) yields
	\begin{eqnarray*}
	\pr\brk{G(n,m)\in\cA}&=&\pr\brk{G^*(n,m)\in\cA|\Sigma}
		\leq\frac{\pr\brk{G^*(n,m)\in\cA}}{\pr\brk{G^*(n,m)\in\Sigma}}\\
		&\sim&\exp\brk{c+c^2}\pr\brk{G^*(n,m)\in\cA},
	\end{eqnarray*}
as claimed.
\end{proof}

\begin{corollary}\label{corollary:expectation-equivalence}
Suppose that $m=cn$ for a fixed $c>0$.
For a graph $G$ let $Z_k(G)=|{\cal S}_k(G)|$.
Then for any $1\leq k\leq 0.99n$ we have
\begin{displaymath}
\ln E[Z_k(G^*(n,m))] =\ln  E[Z_k(G(n,m))]+O(1).
\end{displaymath}
\end{corollary}
\begin{proof}
Let $Q\subset V$ be a set of size $k$, and let $Z_Q(G)=1$ if $Q$ is
independent in $G$, and set $Z_Q(G)=0$ otherwise. The total number
of sequences of $m$ vertex pairs such that $Q$ is an independent set
in the corresponding graph $G^*(n,m)$ equals
	$(n^2-k^2)^m$ (just avoid the $k^2$ pairs of vertices in $Q$).
Hence,
	\begin{eqnarray}\label{eqexpequi1}
	E\brk{Z_Q(G^*(n,m))}&=&\frac{(n^2-k^2)^m}{n^{2m}},\qquad\mbox{and similarly}\\
	E\brk{Z_Q(G(n,m))}&=&\bink{\bink{n}2-\bink k2}m/\bink{\bink n2}m=\frac{(\bink{n}2-\bink k2)_m}{\bink n2_m}.	\label{eqexpequi2}
	\end{eqnarray}
Combining~(\ref{eqexpequi1}) with (\ref{eqexpequi2}) and using
$\ln(1-x)=-x+O(x^2)$ as $x\ra0$, we obtain
	\begin{eqnarray*}
	\frac{E\brk{Z_Q(G^*(n,m))}}{E\brk{Z_Q(G(n,m))}}
		&=&\frac{2^m\bink n2_m}{n^{2m}}\cdot\frac{(n^2-k^2)^m}{2^m(\bink n2-\bink k2)_m}
			\;\stacksign{(\ref{eqGnm*})}{\sim}\;
				\exp(-c-c^2)\frac{(n^2-k^2)^m}{2^m(\bink n2-\bink k2)_m}\\
		&=&\exp\brk{-c-c^2-\sum_{j=0}^{m-1}\ln\bc{1-\frac{n-k}{n^2-k^2}-\frac{2j}{n^2-k^2}}}\\
		&\sim&\exp\brk{-c-c^2+\frac{m(n-k)}{n^2-k^2}+\frac{m^2}{n^2-k^2}}\\
		&=&\exp\brk{-c-c^2+\frac{c}{1+k/n}+\frac{c^2}{1-(k/n)^2}}
				=\exp\brk{-\frac{ck}{n+k}+\frac{c^2k^2}{n^2-k^2}}.
	\end{eqnarray*}
Hence, by the linearity of expectation,
	\begin{eqnarray*}
	E[Z_k(G^*(n,m))]&=&\bink nk\cdot E\brk{Z_Q(G^*(n,m))}
			=\exp\brk{-\frac{ck}{n+k}+\frac{c^2k^2}{n^2-k^2}}\cdot \bink nk E\brk{Z_Q(G(n,m))}\\
		&=&\exp\brk{-\frac{ck}{n+k}+\frac{c^2k^2}{n^2-k^2}}E[Z_k(G(n,m))].
	\end{eqnarray*}
Taking logarithms and recalling that $k\leq0.99n$ completes the proof.
\end{proof}

\noindent
Finally we present a lemma that it will be very useful 
in the course of this paper.

\begin{lemma}[Expectation.]\label{lemma:expectation}
Let $m=dn/2$ for a real $d>0$.
Let $0<\beta<\ln d-\ln\ln d+1-\ln 2$ and set
	\begin{displaymath}
	k=\frac{2n}{d}\left(\ln d-\ln\ln d+1-\ln 2-\beta \right)>0.
	\end{displaymath}
If $Z_k(G)$ is the number of independent sets of size  $k$ in $G$, then
	\begin{displaymath}
		\ln E[Z_k(G^*(n,m))]=k\left[\beta
			-\ln\left(1-\frac{\ln \ln d-1+\ln 2+\beta}{\ln d}\right)
		-\frac{1-\epsilon_d}{2}\frac kn
		\right].
	\end{displaymath}
for $\epsilon_d\to 0$ as $d\to \infty$.
\end{lemma}
\begin{proof}
Since $G^*(n,m)$ is obtained by choosing $m$ independent pairs of
vertices, we have
\begin{equation}\label{eqexpgeneral}
E[Z_k(G^*(n,m))]={n \choose k}(1-(k/n)^2)^m.
\end{equation}
Let $s=\frac{k}{n}$.
By Stirling's formula and the fact that for $x>0$ it holds that
$\ln(1-x)=-x-\frac{x^2}{2(1-\xi)^2}$ for  some $0<\xi<x$, we get
that
\begin{eqnarray}
\ln {n \choose k}&=& -n(s\ln s+(1-s)\ln(1-s)) +o(n) \nonumber \\ 
&=& ns(-\ln s+1 - s/2-s^2/(2(1-\xi_1)^2)+o(n) \hspace{3.5cm}\mbox{[where $0<\xi_1<s$]} \nonumber\\ 
&=& k\left[\ln d-\ln\ln d-\ln 2 +1 - \ln (1-q_d)-k/(2n) +(k/n)^2/(2(1-\xi_1)^2)\right]+o(n),
\label{eqexpcomput1}
\end{eqnarray}
where $q_d=\frac{\ln \ln d-1+\ln 2+\beta}{\ln d}$.
As $m=\frac{d}{2}n$, we obtain
\begin{eqnarray}%
\ln (1-s^2)^m&=& -dn/2\left(s^2+s^4/(2(1-\xi_2)^2)\right) \nonumber \\
&=&-ns[ds/2+ds^3/(2(1-\xi_2))^2 ] \hspace{4.5cm}\mbox{[where $0<\xi_2<s^2$]} \nonumber \\ 
&=&-k\left(\ln d-\ln\ln d-\ln 2 +1-\beta +d(k/n)^3/(2(1-\xi_2)^2) \right ).
\label{eqexpcomput2}
\end{eqnarray}
Note that both $\xi_1,\xi_2$ tend to zero with $d$.
Combining~(\ref{eqexpcomput1}) and~(\ref{eqexpcomput2}) 
yields the assertion.
\end{proof}

\noindent
We also need the following theorem from Dani and Moore~\cite{W2ndM} on
the independence number of $G^*(n,m)$.
\begin{theorem}\label{lemma:SMBound}
There is a constant $\alpha_0>0$ such that
for any $x>4/e$ and any $k\leq \alpha_0n$ the following is true.
Suppose that
\begin{displaymath}
d \leq 2 (n/k)(\ln(n/k)+1)-{x}{\sqrt{n/k}}
\end{displaymath}
and let $m=dn/2$.
Then
$\alpha(G^*(n,m))\geq k$
\whp
\end{theorem}

\noindent
{\bf Remark.}
In a previous version of this work \cite{arxivTR} we derived a slightly 
weaker bound on $d$, i.e. $d<2(n/k)(\ln (n/k)+1)-O(\sqrt{\ln(n/k)\cdot (n/k)})$. 
As opposed to the weighted second moment in \cite{W2ndM}, 
our approach is based on  ``vanilla'' second moment calculations and the
use of a Talagrand type inequality, i.e. similar to that in \cite{frieze-is}.
\\ \vspace{-.3cm}

\noindent
From \cite{W2ndM} we, also, have the following corollary.

\begin{corollary}\label{cor:SMBoundRev}
Let $W(z)$ denote the largest positive root $y$ of the equation
$ye^y=z$. W.h.p. it holds that
\begin{displaymath}
0\leq \frac{2}{d}W\left(\frac{ed}{2}\right)-\alpha(G^*(n,m))  \leq
y\sqrt{\frac{\ln d}{d^3}},
\end{displaymath}
for any constant $y>4\sqrt{2}/e$. Expanding $W(ed/2)$ asymptotically in $d$ 
we have that
\begin{eqnarray}
W\left(\frac{ed}{2}\right) &=& 
\ln d-\ln\ln d+1-\ln 2+\frac{\ln\ln d}{\ln d}-\frac{1-\ln 2}{\ln d}
\nonumber \\
&&+\frac{1}{2}\left(\frac{\ln\ln d}{\ln d}\right)-(2-\ln 2)\frac{\ln\ln d}{\ln^2 d}+
\frac{3+\ln^22-4\ln 2}{2\ln^2d}+O\left(\left(\frac{\ln\ln d}{\ln d}\right)^3\right).
\nonumber
\end{eqnarray}
\end{corollary}

\noindent
It is well known that the independence number $\alpha(G^*(n,m))$ of
the random graph is tightly concentrated. More precisely, the following 
lower tail  bound follows from a standard application of Talagrand's
large deviations inequality~\cite{TalagrandIneq}, similar to the
one used in~\cite[\Sec~7.1]{janson} to establish concentration for
$\alpha(G(n,p))$.

\begin{theorem}\label{thrm:TalagrandTailBound}
Suppose that $d,k$ are as in Theorem \ref{lemma:SMBound}. 
Then for $m=\frac{dn}{2}$ and for any positive integer $t<k$ 
it holds that
	$$Pr[\alpha(G^*(n,m))<t]\leq12\exp\left(-\frac{(k-t+1)^2}{4k}\right).$$
\end{theorem}
\begin{proof}
Consider the graph $G(n,p)$ where $p=d/n$ and let $E(G(n,p))$
denote the number of its edges. It holds that
\begin{eqnarray*}
Pr[\alpha(G(n,p))\geq k]&=&
\sum_{M=0}^{{n\choose 2}} Pr[\alpha(G^*(n,M))\geq k]Pr[E(G(n,p))=M] \\
&\geq & \sum_{M\leq dn/2}Pr[\alpha(G^*(n,m))\geq k]Pr[E(G(n,p))=M]
\qquad \mbox{[where $m=dn/2$]}\\ 
&\geq & Pr[\alpha(G^*(n,m))\geq k]Pr\left[E(G(n,p))\leq \frac{dn}2\right].
\end{eqnarray*}
From the above derivations and Theorem \ref{lemma:SMBound}, 
it is direct that
\begin{equation}\label{eq:GnpLB}
Pr[\alpha(G(n,p))\geq k]\geq \frac{1}{3}Pr[\alpha(G^*(n,m))\geq k]\geq 1/4.
\end{equation}
A vertex exposure argument allows to apply Talagrand's large 
deviation inequality for the independence number of $G(n,p)$
(in the form that appears in \cite{janson}, page 41 (2.39)).
The following holds:
\begin{displaymath}
Pr[\alpha(G(n,p))<t]Pr[\alpha(G(n,p))\geq k]
\leq \exp\left( -{(k-t+1)^2}/{4k} \right).
\end{displaymath}
Using (\ref{eq:GnpLB}) we get
\begin{displaymath}
Pr(\alpha(G(n,p))<t)\leq 4 \exp\left( -{(k-t+1)^2}/{4k} \right).
\end{displaymath}
Working as in (\ref{eq:GnpLB}) we get that
$
\frac{1}{3}Pr[\alpha(G^*(n,m))<t]\leq Pr[\alpha(G(n,p))<t].
$
The theorem follows.
\end{proof}

\begin{corollary}\label{theorem:Reverse-Frieze}
For the integer $k>0$ let
	\begin{displaymath}
	\delta_k = 2(n/k)\ln (n/k)+2(n/k)-8\sqrt{(n/k)}.
	\end{displaymath}
There is a constant $\alpha_0>0$ such that for $k<\alpha_0n$ and $G^*(n,m)$
of expected degree $d\leq \delta_k$ it holds that
\begin{eqnarray}\label{eq:tail4aG}
Pr[\alpha(G^*(n,m))< k]&\leq&12 \exp\left(- n/(d^{2}\ln^5 d)\right).
\end{eqnarray}
Also, for $d=\delta_k$ it holds that
$E|{\cal S}_k(G^*(n,m))|\leq \exp\left( 14 n\sqrt{{\ln^5 d}/{d^3}} \right)$.
\end{corollary}
\begin{proof}
Let $G^*(n,m)$ be of expected degree $d = 2 (n/k)(\ln(n/k)+1)-{8}{\sqrt{n/k}}$,
where $k$ is as in the statement. Also, let $k'$ be such that
$d=2 (n/k')(\ln(n/k')+1)-2{\sqrt{n/k'}}$. By Theorem \ref{thrm:TalagrandTailBound}
we have that
\begin{eqnarray}
Pr[\alpha(G^*(n,m))<k]&\leq& 12\exp\left(-\frac{(k'-k+1)^2}{4k'}\right)
\leq 12\exp\left(-\frac{(k'-k+1)^2}{8k}\right), \label{eq:talagrand1821}
\end{eqnarray}
where the last inequality follows from the fact that $k'<2k$. The
tail bound in (\ref{eq:tail4aG}) will follow by bounding appropriately
$t=k'-k>0$. We bound $t$ by using the fact that
\begin{displaymath}
2 (n/k)(\ln(n/k)+1)-{8}{\sqrt{n/k}}=2 (n/k')(\ln(n/k')+1)-2{\sqrt{n/k'}}.
\end{displaymath}
Set $s=k/n$ and $q=t/k$. Let $h(s,q)$ be the difference of the l.h.s.
minus r.h.s. in the above equality, written in terms of $s,t$.
Clearly, it holds that that $h(s,q)=0$. That is
\begin{eqnarray}
h(s,q)
&=&\frac{2\ln(1+q)}{s}+\frac{q}{1+q}\left(-\ln s-\ln(1+q)+1\right)-
\frac{2}{\sqrt{s}}\left(4-\frac{1}{\sqrt{1+q}}\right)=0.
\nonumber
\end{eqnarray}
For $1.5 n\ln d/d < k,k'< 2n\ln d/d$, it is direct to verify that for 
$q=10/\sqrt{d\ln^5 d}$ and sufficiently small $s$ 
it holds that $h(s,q)<0$. Furthermore, it is easy to see that
\begin{displaymath}
\frac{\partial}{\partial q} h(s,q)=\frac{2}{s(1+q)}+
\frac{1}{(1+q)^2}\left(-\ln s-\ln(1+q)+1-q\right)-\frac{1}{\sqrt{s(1+q)^3}}.
\end{displaymath}
For any $q\in [0,1]$ and sufficiently small $s$ we have that
$\frac{\partial}{\partial q} h(s,q)>0$. This yields to the fact
that for any $q\leq 10/\sqrt{d\ln^5}$ and sufficiently small
$s$ we have $h(s,q)<0$. Thus, we get that $k'-k\geq 10k/\sqrt{d\ln^5 d}$.
Plugging this  into (\ref{eq:talagrand1821}) we get that
\begin{eqnarray}
Pr[\alpha(G^*(n,m))<k]&\leq& 12\exp\left(-\frac{100}{8}\frac{k}{d\ln^5 d}\right)
\nonumber \\
&\leq& 12\exp\left(-\frac{300}{16}\frac{n}{d^{2}\ln^4d}\right), \qquad \mbox{[as $k\geq 1.5n \ln d/d$]}
\nonumber
\end{eqnarray}
which implies (\ref{eq:tail4aG}).

For the rest of the proof, consider $G^*(n,m)$ with expected degree
$d=\delta_k$.  Assume that we add to $G^*(n,m)$ edges at random so
as to increase the expected degree to $d^+= 2\frac{s\ln s+(1-s)\ln(1-s))}{\ln(1-s^2)}$ 
and get the graph $G^*(n,m')$.
That is, we need to insert into $G^*(n,m)$ as many as $(d^+-d)n/2$
random edges. Therefore, each independent set of size $k$ in
$G^*(n,m)$ is also an independent set of $G^*(n,m')$ with probability
$\left(1-(k/n)^2 \right)^{(d^+-d)n/2}$. Let $s=(k/n)$. It is direct
that 
\begin{equation}\label{eq:Relation1896A}
E|{\cal S}_k(G(n,m'))|=\left(1-s^2 \right)^{(d^+-d)n/2}E[|{\cal S}_k(G(n,m))|].
\end{equation}
Using Corollary \ref{corollary:expectation-equivalence} we get that
\begin{eqnarray}
\frac{1}{n}\ln E|{\cal S}_k(G(n,m'))|&=&\frac{1}{n}\ln\left({n \choose k}(1-(k/n)^2)^{d^+n/2}\right)+O\left(\frac{1}{n}\right)
\nonumber \\
&\sim&-[s\ln s+(1-s)\ln(1-s)]+d^+\ln(1-s^2)/2-\frac{\ln n}{2n} 
\nonumber \\
&\sim &-\frac{\ln n}{2n}. \label{eq:Relation1896B}
\end{eqnarray}
Furthermore, using the fact that $-\frac{x}{1-x}\leq \ln(1-x)\leq -x$,
for $0<x<1$, it is direct that
\begin{equation}\label{eq:Relation1896C}
2\frac{-\ln s+1}{s}\leq d^+\leq 2\frac{-\ln s+1}{s}+2.
\end{equation}
Combining (\ref{eq:Relation1896A}), (\ref{eq:Relation1896B}) and
(\ref{eq:Relation1896C}), we get that 
\begin{eqnarray*}
\frac{1}{n}\ln E|{\cal S}_k(G(n,m))|&\leq& - \ln(1-s^2)(d^+-d)/2-o(1)
\hspace{2.4cm}\mbox{[by (\ref{eq:Relation1896A}) and (\ref{eq:Relation1896B})]}\\
&\leq &4\frac{s^{3/2}}{1-s^2} \hspace{1.4cm}  \mbox{[by (\ref{eq:Relation1896C})
and $1-x>e^{-x/(1-x)}$ for $0<x<1$]}.
\end{eqnarray*}
The upper bound for $E|{\cal S}_k(G(n,m))|$ follows by using the
above inequality and noting that $k \leq 2n\ln d/d$, i.e.
$s \leq 2\ln d/d$.
\end{proof}

\begin{corollary}\label{cor:ExistenceGnm}
For the graph $G(n,m)$ of expected degree $d$ it holds that
\begin{displaymath}
Pr\left[\alpha(G(n,m))\geq 2n(1-\epsilon_d){\ln d}/{d} \right]
\geq 1-\exp \left[-{8n}/(d\ln^3 d) \right]. 
\end{displaymath}
where $\epsilon_d\to0$ as $d$ increases.
\end{corollary}
\begin{proof}
Consider $G^*(n,m)$ of expected degree $d$ and let $k$ be such that
$k/n=\frac{2}{d}\left(W(ed/2)-10\sqrt{\ln d/d^3}-2\frac{\ln\ln d}{\ln d}\right)$, 
where $W(z)$ is defined in the statement of Corollary \ref{cor:SMBoundRev}. 
Using Corollary \ref{cor:SMBoundRev} and Theorem \ref{thrm:TalagrandTailBound}, 
we get that
\begin{displaymath}
Pr\left[\alpha(G^*(n,m))\leq k \right] \leq \exp \left(-\frac{8(\ln\ln d)^2}{d\ln^3 d}n \right).
\end{displaymath}
The corollary follows by using Lemma \ref{lemma:model-equivalence}.
\end{proof}

\noindent
The following is taken from \cite[p.~156]{janson}.

\begin{lemma}\label{Lemma_isoGnm}
Let $d>0$ be fixed and $m=dn/2$.
Let $Y$ be the number of isolated vertices in $G(n,m)$.
Then $Y=(1+o(1))n\exp(-d)$ \whp
\end{lemma}

\section{Approaching the distribution $\cU_k(n,m)$}\label{sec:PlantedModel}

\subsection{The planted model}

The main results of this paper deal with properties of `typical'
independents sets of a given size in a random graph, i.e., the
probability distribution $\cU_k(n,m)$. In the theory of random
discrete structures often the conceptual difficulty of analysing
a probability distribution is closely linked to the computational
difficulty of sampling from that distribution (e.g., \cite[Chapter~9]{janson}).
This could suggest that analysing $\cU_k(n,m)$ is a formidable task,
because for $k>(1+\eps)n\ln(d)/d$ there is no efficient procedure known
for finding an independent set of size $k$ in a random graph $G(n,m)$,
let alone for sampling one at random. In effect, we do not know of
an efficient method for sampling from $\cU_k(n,m)$.

To get around this problem, we are going to `approximate' the
distribution $\cU_k(n,m)$ by another distribution $\cP_k(n,m)$
on the set $\Lambda_k(n,m)$ of graph/independent set pairs, the
so-called planted model, which is easy to sample from. This
distribution is induced by the following experiment:
\begin{quote}
Choose a subset $\sigma\subset\brk n$ of size $k$ uniformly at random.\\
Choose a graph $G$ with $m$ edges \emph{in which $\sigma$ is an
 independent set} uniformly at random.\\
Output the pair $(G, \sigma)$.
\end{quote}
In other words, the probability assigned to a given pair $(G_0,\sigma_0)\in\Lambda_k(n,m)$ is
	\begin{eqnarray}\label{eqProbPlanted}
	\pr_{\cP_k(n,m)}\brk{(G_0,\sigma_0)}&=&\brk{\bink{n}k\cdot\bink{\bink{n}2-\bink{k}2}m}^{-1},
	\end{eqnarray}
i.e., $\cP_k(n,m)$ is nothing but the uniform distribution on $\Lambda_k(n,m)$.
The key result that allows us to study the distribution $\cU_k(n,m)$ is the following.

\begin{theorem}\label{thrm:transfer-theorem}
There is $\eps_d\ra0$ such that  for $k<(2-\eps_d)n\ln(d)/d$ the following is true.
If $\cB$ is an event such that
	\begin{equation}\label{eqtransfer-theoremassumption}
	\pr_{\cP_k(n,m)}\brk\cB=o\bc{\exp\bc{-14 n\sqrt{{\ln^5d}/{d^3}}}},
	\end{equation}
then $\pr_{\cU_k(n,m)}\brk\cB=o(1)$.
\end{theorem}
Hence, \Thm~\ref{thrm:transfer-theorem} allows us to bound the 
probability of some `bad' event $\cB$ in the distribution $\cU_k(n,m)$
by bounding its probability in the distribution $\cP_k(n,m)$.

To establish \Thm~\ref{thrm:transfer-theorem}, we need to find
a way to compare $\cP_k(n,m)$ and $\cU_k(n,m)$. Suppose that
$k<(2-\eps_d)n\ln(d)/d$ is such that $\alpha(G(n,m))\geq k$ \whp\
Then the probability of a pair $(G_0,\sigma_0)\in\Lambda_k(n,m)$
under the distribution $\cU_k(n,m)$ is
\begin{eqnarray}\label{eqProbUniform}
	\pr_{\cU_k(n,m)}\brk{(G_0,\sigma_0)}&\sim&\brk{\bink{\bink{n}2}m\cdot|\cS_k(G_0)|}^{-1}
\end{eqnarray}
(because we first choose a graph uniformly, and then an independent
set of that graph). Hence, the probabilities assigned to $(G_0,\sigma_0)$
under~(\ref{eqProbUniform}) and~(\ref{eqProbPlanted}) coincide
(asymptotically) iff
\begin{eqnarray}\label{eqProbCoincide}
	|\cS_k(G_0)|&\sim&\bink{n}k\bink{\bink{n}2-\bink{k}2}m/\bink{\bink{n}2}m.
\end{eqnarray}
A moment's reflection shows that the expression on the r.h.s.
of~(\ref{eqProbCoincide}) is precisely the \emph{expected}
number $\Erw|\cS_k(G(n,m))|$ of independent sets of size $k$.
Thus, $\cP_k(n,m)$ and $\cU_k(n,m)$ coincide asymptotically iff
the number $|\cS_k(G(n,m))|$ of independents sets of size $k$ is
concentrated about its expectation.

This is indeed the case in `dense' random graphs with $m\gg n^{3/2}$.
For this regime one can perform a `second moment' computation to show
that $|\cS_k(G(n,m))|\sim\Erw|\cS_k(G(n,m))|$ \whp, (e.g. see \cite[Chapter~7]{janson}) 
whence the measures $\cP_k(n,m)$ and $\cU_k(n,m)$ are interchangeable.
This fact forms (somewhat implicitly) the foundation of the proofs
in~\cite{jerrum-planted}.

By contrast, in the sparse case $m\ll n^{3/2}$ a straight second
moment argument fails utterly. As it turns out, this is because
the quantity $|\cS_k(G(n,m))|$ simply it not concentrated about
its expectation anymore. In fact, maybe somewhat surprisingly
\Thm~\ref{thrm:transfer-theorem} can be used to infer the following
corollary, which shows that in sparse random graphs the expectation
$\Erw|\cS_k(G(n,m))|$ `overestimates' the typical number of independent
sets by an exponential factor \whp\

\begin{corollary}\label{cor:transfer-theorem}
There exist functions $\eps_d\ra 0$ and $g(d)>0$ such that for
$10n/d < k<(2-\eps_d)n\ln(d)/d$ we have
$$|\cS_k(G(n,m))|\leq\Erw|\cS_k(G(n,m))|\cdot\exp(-g(d)n)\qquad\mbox\whp$$
\end{corollary}
The proof of Corollary \ref{cor:transfer-theorem} appears in Section
\ref{sec:cor:transfer-theorem}.

Conversely, in order to prove \Thm~\ref{thrm:transfer-theorem} we
need to bound the `gap' between the typical value of $|\cS_k(G(n,m))|$
and its expectation from above. This estimate can be summarized as
follows.

\begin{proposition}\label{Lemma_gap}
There is $\eps_d\ra0$ such that for $k<(2-\eps_d)n\ln(d)/d$ we have
$$|\cS_k(G(n,m))|\geq\Erw|\cS_k(G(n,m))|\cdot\exp\bc{-14n\sqrt{\ln^5d/d^3}}$$
with probability at least $1-\exp\left[-n/(2d^2\ln^4d) \right ]$.
\end{proposition}
Before we prove \Prop~\ref{Lemma_gap} in \Sec~\ref{section:concentration},
let us indicate how it implies \Thm~\ref{thrm:transfer-theorem}.

\begin{corollary}\label{Cor_quantexchange}
There is $\eps_d\ra0$ such that for $k<(2-\eps_d)n\ln(d)/d$ the following is true.
Let
\begin{equation}\label{eqtransfer-theoremproof0}
\cZ=\cbc{(G,\sigma)\in\Lambda_k(n,m):|\cS_k(G)|\geq\Erw|\cS_k(G(n,m))|\cdot\exp\bc{-14n\sqrt{\ln^5d/d^3}}}.
	\end{equation}
Then $\pr_{\cU_k(n,m)}\brk{\cZ}=1-o(1)$, and for any event $\cB\subset\Lambda_k(n,m)$ we have
	$$\pr_{\cU_k(n,m)}\brk{\cB|\cZ}\leq(1+o(1))\exp\bc{-14n\sqrt{\ln^5d/d^3}}\cdot\pr_{\cP_k(n,m)}\brk{\cB}.$$
\end{corollary}
\begin{proof}
\Prop~\ref{Lemma_gap} directly implies that
	\begin{equation}\label{eqtransfer-theoremproof1}
	\pr_{\cU_k(n,m)}\brk\cZ=1-o(1).
	\end{equation}
Furthermore, by the definition~(\ref{eqProbUniform}) of the uniform distribution, 
	\begin{eqnarray}\nonumber
	\pr_{\cU_k(n,m)}\brk{\cB\cap\cZ}&=&
		\sum_{(G,\sigma)\in\cB\cap\cZ}\brk{\bink{\bink n2}m|\cS_k(G)|}^{-1}\\
	&\leq&\exp\brk{14n\sqrt{\ln^5d/d^3}}\sum_{(G,\sigma)\in\cB\cap\cZ}\brk{\bink{\bink n2}m\Erw|\cS_k(G(n,m))|}^{-1}
		\quad\mbox{[by~(\ref{eqtransfer-theoremproof0})]}\nonumber\\
	&=&\exp\brk{14n\sqrt{\ln^5d/d^3}}\pr_{\cP_k(n,m)}\brk{\cB\cap\cZ}
		\qquad\qquad\qquad\qquad\qquad\mbox{[by~(\ref{eqProbPlanted})]}\nonumber\\
	&\leq&\exp\brk{14n\sqrt{\ln^5d/d^3}}\pr_{\cP_k(n,m)}\brk{\cB}.
		\label{eqtransfer-theoremproof2}
	\end{eqnarray}
The assertion is immediate from~(\ref{eqtransfer-theoremproof1}) and~(\ref{eqtransfer-theoremproof2}).
\end{proof}

\smallskip
\begin{theoremproof}{\ref{thrm:transfer-theorem}}
The theorem follows directly from \Cor~\ref{Cor_quantexchange}.
\end{theoremproof}

\subsection{Proof of Proposition \ref{Lemma_gap}}\label{section:concentration}

Since the second moment method fails to yield a lower bound on
the typical number  of independent sets $|\cS_k(G(n,m))|$, we
need to invent a less direct approach to prove Proposition \ref{Lemma_gap}.
Of course, the demise of the second moment argument also presented
an obstacle to Frieze~\cite{frieze-is} in his proof that
		\begin{equation}\label{eqAlan}
		\alpha(G(n,m))\geq(2-\eps_d)n\ln(d)/d\qquad\mbox\whp
		\end{equation}
However, unlike the \emph{number} $|\cS_k(G(n,m))|$ of independent
sets $\alpha(G(n,m))$, the \emph{size}  of the largest one actually
is concentrated about its expectation. In fact, an arsenal of large
deviations inequalities applies (e.g., Azuma's and Talagrand's
inequality), and~\cite{frieze-is} uses these to bridge the gap
left by the second moment argument.Unfortunately, these large
deviations inequalities draw a blank on $|\cS_k(G(n,m))|$. Therefore,
we are going to derive the desired lower bound on $|\cS_k(G(n,m))|$
directly from~(\ref{eqAlan}).

To simplify our derivations we consider the model of random graphs
$G^*(n,m)$ and we show the following proposition.
\begin{proposition}\label{proposition:star-conctration}
There is $\eps_d\ra0$ such that for $k<(2-\eps_d)n\ln(d)/d$ we 
have

\begin{equation}\label{eq:prop:star-conct}
|\cS_k(G^*(n,m))|\geq\Erw|\cS_k(G^*(n,m))|\cdot\exp\bc{-14n 
\sqrt{{\ln^{5}d}/{d^3}}}
\end{equation}
with probability at least $1-\exp\left[-n/(d\ln^2d)^2\right]$.
\end{proposition}
Then, Proposition \ref{Lemma_gap} follows by Lemmas~\ref{lemma:model-equivalence}
and~\ref{corollary:expectation-equivalence}.

Given some integer $k>0$ and $q\in [0,1]$, let 	$Z_k(G^*(n,m))=|{\cal S}_k(G^*(n,m))|$
and let 		$$M^q_k=\max\{m\in \mathbb{N}:Pr[Z_k(G^*(n,m))>0]\geq 1-q\}.$$
In words, $M^q_k$ is the largest number of edges that we can squeeze
in while keeping the probability that $G^*(n,m)$  has an independent
set of size $k$ above $1-q$. The following lemma summarizes the key
step of our proof of \Prop~\ref{proposition:star-conctration}. The
idea is that \Lem~\ref{lemma:concentration-1} gives a tradeoff
between the \emph{likely} number of independent set of size $k$
in the random graph with $m<M^q_k$ edges and the \emph{expected}
number of such independent sets in the random graph with $M^q_k$
edges.

\begin{lemma}\label{lemma:concentration-1}
Suppose that $k,m>0,q\in\brk{0,1}$ are such that $m<M^q_k$. Then
\begin{displaymath}
Pr\left[Z_k(G^*(n,m))<\frac{ E[Z_k(G^*(n,m))}{2E[Z_k(G^*(n,M^q_k))]} \right]\leq 2 q.
\end{displaymath}
\end{lemma}
\begin{proof}
Let $M=M^q_k$. The random graph $G^*(n,M)$ is obtained by choosing
$M$ pairs of vertices independently	and inserting the corresponding
edges (while omitting loops and reducing multiple edges to single edges).
Let us think of the $M$ pairs as being generated in two rounds.
In the first round, we generate $m$ pairs, which induce the
random graph $G_1=G^*(n,m)$. In the second round, we choose a
further $M-m$ pairs independently and add the corresponding edges
to $G_1$ (again, omitting self-loops and reducing multiple edges
to single edges) to obtain $G_2=G^*(n,M)$.

By the linearity of the expectation and because the $m$ (resp.\ $M$)
pairs that the random graph $G_1$ (resp.\ $G_2$) consists of are
chosen independently, we have  (cf.~(\ref{eqexpgeneral}))
	\begin{eqnarray}\nonumber
	E\brk{Z_k(G_1)}&=&\bink nk(1-(k/n)^2)^m,\qquad\mbox{and}\\
	E\brk{Z_k(G_2)}&=&\bink nk(1-(k/n)^2)^M=E\brk{Z_k(G_1)}\cdot(1-(k/n)^2)^{M-m}.
		\label{eqtrajectory3}
	\end{eqnarray}
Furthermore, with respect to the 	number of independent sets of
size $k$ in $G_2$ \emph{given} their number in the outcome $G_1$
of the `first round', we have
	\begin{equation}\label{eqtrajectory1}
	E[Z_k(G_2)|Z_k(G_1)]=Z_k(G_1)(1-(k/n)^2)^{M-m}.
	\end{equation}
Indeed, for each independent set $Q$ of size $k$ in $G_1$ each of
the $M-m$ additional random pairs has its two vertices in $Q$ with
probability $(k/n)^2$. Hence, (\ref{eqtrajectory1}) follows because
these $M-m$ pairs are independent and by the linearity of the expectation.

\noindent
Now, let $\cE_1$ be the event that
	$$Z_k(G_1)<\frac{E[Z_k(G_1)]}{2E[Z_k(G_2)]}.$$
Then by and Markov's inequality and~(\ref{eqtrajectory1}),
	\begin{eqnarray*}
	\frac12\leq Pr\left[Z_k(G_2)<2\Erw\brk{Z_k(G_2)|\cE_1} \big| {\cal E}_1\right]
	\leq Pr\left[Z_k(G_2)<\frac{E[Z_k(G_1)]\cdot(1-(k/n)^2)^{M-m}}{E[Z_k(G_2)]} \bigg| {\cal E}_1\right],
	\end{eqnarray*}
whence
	\begin{eqnarray}\label{eqtrajectory2}
	Pr\left[Z_k(G_2)<\frac{E[Z_k(G_1)]\cdot(1-(k/n)^2)^{M-m}}{E[Z_k(G_2)]}\right]\geq Pr\brk{\cE_1}/2.
	\end{eqnarray}
Combining~(\ref{eqtrajectory2}) and~(\ref{eqtrajectory3}), we see that
	$Pr\brk{\cE_1}\leq2\,Pr\left[Z_k(G_2)<1\right]\leq 2q,$
as claimed.
\end{proof}

\noindent{\bf Proof of \Prop~\ref{proposition:star-conctration}.}
Consider $G^*(n,m)$ of expected degree $d$ and let $k=\frac{2}{d}
\left(\ln d-\ln\ln d+1-\ln 2 \right)$. We are going to 
show that (\ref{eq:prop:star-conct}) holds for $G^*(n,m)$ and
$k$ with probability at least $1-\exp\left[-n/(d\ln^2d)^2\right]$.

Consider, now, the graph $G(n,M)$ of expected degree 
$d^+=2\frac{-\ln s+1}{s}+\frac{8}{\sqrt{s}}$, where $s=k/n$.
According to \ref{theorem:Reverse-Frieze} it holds that 
$Pr[|S_k(G(n,M))|>0]\geq 1-12\exp\left(-n/(d^2\ln^5d)\right)$
and
$E|S_k(G(n,M))|\leq \exp\left(14\sqrt{\frac{\ln^5d}{d^3}}\right)$.

The proposition will follow by just showing that $m<M$,
i.e. $d^+>d$, and using Lemma \ref{lemma:concentration-1}.
Note, first, that
\begin{eqnarray}
-\ln s+1&=&\ln d-\ln\ln d+1-\ln 2-\ln \left(1-\frac{\ln\ln d-1+\ln 2}{\ln d}\right)\nonumber\\
&\geq & \ln d-\ln\ln d+1-\ln 2+ \frac{\ln\ln d-1+\ln 2}{\ln d}. \hspace{3cm} \mbox{[as $1-x\leq e^{-x}$]}\nonumber 
\end{eqnarray}
Using the above, it is elementary to derive that $2\frac{-\ln s+1}{s}\geq d$.
Then, it follows that $d^+>d$ as promised.
\qed

\subsection{Proof of Corollary \ref{cor:transfer-theorem}}\label{sec:cor:transfer-theorem}

In this section we keep the assumptions of \Cor~\ref{cor:transfer-theorem},
i.e., we let $k,d$ be such that $10 n/d< k<(2-\eps_d)n\ln(d)/d$, with
$\eps_d\ra0$ sufficiently slowly in the limit of large $d$.

\begin{lemma}\label{Lemma_isoPnm}
There is a number $\xi>0$ such that the following is true. Let
$(G,\sigma)$ be a pair chosen from the distribution $\cP_k(n,m)$.
Let $X$ be the number of isolated vertices in $G$. Then
\begin{equation}\label{eqiso1}
	\pr\brk{X\leq2n\exp(-d)}\leq\exp(-3\xi n).
\end{equation}
\end{lemma}
\begin{proof}
Let $\alpha=k/n$. It is convenient to first consider the following
variant of the planted distribution: given a set $\sigma\subset V$
of size $k$, let $G'$ be the random graph obtained by including each
of the $\bink n2-\bink k2$ possible edges that do not link two
vertices in $\sigma$ with probability
$$q=\frac{m}{\bink{n}2-\bink{k}2}\sim\frac{m}{\bink{n}2(1-\alpha^2)}\sim\frac{d}{n(1-\alpha^2)}$$
independently. Hence, the total number of edges in $G'$ is binomially
distributed with mean $m$. By Stirling's formula, the event $\cE$
that $G'$ has precisely $m$ edges has probability $\Theta(m^{-1/2})$,
and given that $\cE$ occurs, the pair $(G',\sigma)$ has the same
distribution as the pair $(G,\sigma)$ chosen from the distribution
$\cP_k(n,m)$. Therefore, for any event $\cA$ we have 
\begin{equation}\label{eqG'sigma}
	\pr\brk{(G,\sigma)\in\cA}=\pr\brk{(G',\sigma)\in\cA|\cE}
		\leq O(\sqrt m)\cdot\pr\brk{(G',\sigma)\in\cA}.
\end{equation}

\noindent
Now, consider the number $X'$ of vertices in $\sigma$ that are
isolated in $G'$. Since each possible edge is present in $G'$
with probability $q$ independently, the degree of each vertex
$v\in\sigma$ has a binomial distribution $\Bin(n-k,q)$ with mean
	$$q(1-\alpha)n=d\cdot\frac{1-\alpha}{1-\alpha^2}=\frac d{1+\alpha}.$$
In particular, for each $v\in\sigma$ we have
$$\pr\brk{v\mbox{ is isolated in }G'}\sim\exp(-(1+\alpha)d).$$
Furthermore, because $\sigma$ is an independent set, the degrees
of the vertices in $\sigma$ are mutually independent. Hence, $X'$
has a binomial distribution $\Bin(k,(1+o(1))\exp(-(1+\alpha)d))$
with mean
\begin{eqnarray*}
	\Erw\brk{X'}&\sim&\alpha n\exp(-d/(1+\alpha))\geq
			n\exp\brk{-d\bc{1-\alpha+O_d(\alpha^2)}}\\
			&\geq&n\exp\brk{-d-10+o_d(1)}\geq100 n\exp(-d)\qquad\qquad\mbox{[as we assume $\alpha\geq10/d$]},
\end{eqnarray*}
provided that $d$ is sufficiently large. Since $X'$ is binomially 
distributed, Chernoff bounds yield a number $\xi=\xi(d)>0$ such
that 
\begin{equation}\label{eqiso2}
	\pr\brk{X'\leq2n\exp(-d)}\leq\exp(-4\xi n).
\end{equation}
Finally, combining~(\ref{eqiso2}) and~(\ref{eqG'sigma}), we obtain
$$\pr_{\cP_k(n,m)}\brk{X\leq2n\exp(-d)}\leq O(\sqrt m)\pr\brk{X'\leq2n\exp(-d)}\leq\exp(-3\xi n),$$
as claimed.
\end{proof}

\noindent{\bf Proof of Corollary \ref{cor:transfer-theorem}.}
Let $\cB\subset\Lambda_k(n,m)$ be the set of all pairs $(G,\sigma)$
such that $G$ has fewer than $2n\exp(-d)$ isolated vertices. \Lem s~\ref{Lemma_isoGnm} and~\ref{Lemma_isoPnm} entail that 
\begin{eqnarray}\label{eqcor:transfer-theorem1}
	\pr_{\cU_k(n,m)}\brk\cB=1-o(1)&\mbox{ while }&\pr_{\cP_k(n,m)}\brk\cB\leq\exp(-\xi n).
\end{eqnarray}
Since $\cP_k(n,m)$ is the uniform distribution over $\Lambda_k(n,m)$,
(\ref{eqcor:transfer-theorem1}) implies that 
\begin{eqnarray}\label{eqcor:transfer-theorem2}
	\abs\cB&\leq&\abs{\Lambda_k(n,m)}\cdot\exp(-\xi n)=\bink{\bink{n}2}m\Erw|\cS_k(G(n,m))|\cdot\exp(-\xi n).
\end{eqnarray}
Now, let $\cA\subset\Lambda_k(n,m)$ be the set of all pairs
$(G,\sigma)$ such that 	$|\cS_k(G)|\geq\exp(-\xi n/3)\Erw|\cS_k(G(n,m))|$,
and assume for contradiction that there is a fixed $\eps>0$ such
that $\pr_{\cU_k(n,m)}\brk{\cA}\geq\eps$. Then~(\ref{eqcor:transfer-theorem1})
implies that
\begin{eqnarray*}\label{eqcor:transfer-theorem3}
	\pr_{\cU_k(n,m)}\brk{\cA\cap\cB}&\geq&\eps-o(1) 
\end{eqnarray*}
Therefore,
\begin{eqnarray*}
	\abs\cB&\geq&\abs{\cA\cap\cB}\geq \bink{\bink n2}{m}\pr_{\cU_k(n,m)}\brk{\cA\cap\cB}\cdot\exp(-\xi n/3)\Erw|\cS_k(G(n,m))|\\
		&\geq&(\eps-o(1))\bink{\bink n2}{m}\exp(-\xi n/3)\Erw|\cS_k(G(n,m))|\geq(\eps-o(1))\exp(-\xi n/3)\cdot\abs{\Lambda_k(n,m)},
\end{eqnarray*}
which contradicts~(\ref{eqcor:transfer-theorem2}). Hence, $\pr_{\cU_k(n,m)}\brk{\cA}=o(1)$,
as claimed.
\qed

\section{Proof of Theorem \ref{thrm:Connectivity}}\label{section:thrm:Connectivity}

Instead of the random graph model $G(n,m)$  we consider the model $G(n,p)$, 
where $p=d/n$ for fixed real $d$ and we prove the following theorem.

\begin{theorem}\label{thrm:ConnectivityGnp}
There is $\eps_d\ra0$ such that $\cS_k(G(n,d/n))$ is $O(1)$-connected 
for any $k\leq(1-\eps_d)\frac{\ln d}d\cdot n$, with probability at least 
$1-\exp\left(-\frac{\ln^{40}d}{d}n\right)$.
\end{theorem}

\noindent
Theorem \ref{thrm:Connectivity} follows 
by using standard arguments, i.e. the following corollary.
\begin{corollary}
For any fixed $d>0$, $m=dn/2$ and any graph property $A$ it holds that
$Pr[G(n,m)\in A]\leq \Theta(\sqrt{n})Pr[G(n,d/n)\in A]$.
\end{corollary}  
\begin{proof}
Let $E_d$ be the number of edges in $G(n,d/n)$. It holds that
\begin{eqnarray}
Pr[G(n,m)\in A]&=&Pr[G(n,d/n)\in A|E_d=dn/2]\leq \frac{Pr[G(n,d/n)\in A]}{Pr[E_d=dn/2]}.
\nonumber
\end{eqnarray}
$E_d$ is binomially distributed with parameters ${n \choose 2}$ and $d/n$.
Straightforward calculations yield to that $Pr[E_d=dn/2]=\Theta(1/\sqrt{n})$.
The corollary follows.
\end{proof}

\begin{figure}
\begin{minipage}{0.5\textwidth}
	\centering
		\includegraphics[width=0.4\textwidth]{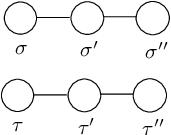}
	\caption{ The short chains}
	\label{fig:Chains}
\end{minipage}
\begin{minipage}{0.5\textwidth}
	\centering
		\includegraphics[width=0.75\textwidth]{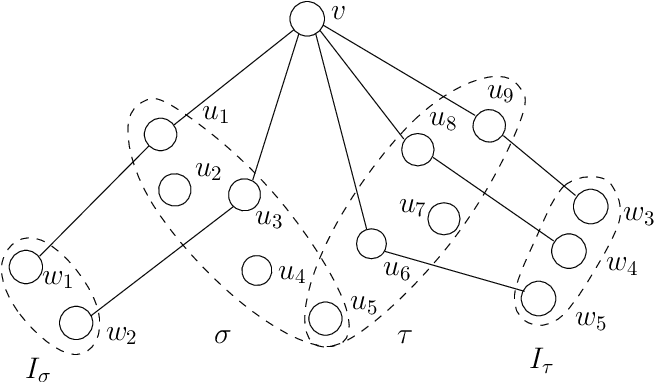}
	\caption{ $\sigma, \tau$ with Property $\Gamma$ }
	\label{fig:G-Property}
\end{minipage}
\end{figure}

\noindent
{\bf Remark.} 
We show Theorem \ref{thrm:ConnectivityGnp} by just considering 
the adjacent independent sets with Hamming distance at most $20d$.
\\ \vspace{-.3cm}

\noindent
For every vertex $u$ in $G(n,d/n)$ we let $N(u)$  (or $N_v$)
denote the  set vertices which are  adjacent to $u$.
A sufficient condition for establishing the connectivity of
${\cal S}_k(G(n,d/n))$ is requiring this space to have what
we call  Property $\Gamma$:\\ \vspace{-.3cm}

\noindent
{\bf Property $\boldsymbol{\Gamma}$.}
For any two  $\sigma, \tau \in {\cal S}_k(G(n,d/n))$ there 
exist chains $\sigma, \sigma', \sigma''$ and $\tau,\tau',\tau''$ of 
independent sets in ${\cal S}_k(G(n,d/n))\bigcup {\cal S}_{k+1}
(G(n,d/n))$ connected as in Figure \ref{fig:Chains}.
Furthermore, we have that $\sigma'', \tau''\in{\cal S}_k(G(n,d/n))$
and $dist(\sigma'', \tau'')<dist(\sigma, \tau)$. In particular it holds
that $|\sigma''\cap \tau''| = |\sigma\cap \tau|+1$.\\ \vspace{-.3cm}

\noindent
The following result is straightforward.

\begin{corollary}\label{cor:GPropertyToConnect}
If ${\cal S}_k(G(n,d/n))$ has Property $\Gamma$, then it is connected.
\end{corollary}

\noindent
Using Corollary \ref{cor:GPropertyToConnect}, Theorem \ref{thrm:ConnectivityGnp} 
will follow by showing that with probability $1-o(1)$ the set ${\cal S}_k(G(n,d/n))$ 
has Property $\Gamma$, for $k<(1-\epsilon_d)\ln d/d$ .
For this, we need to introduce the notion of ``augmenting vertex''.

\begin{definition}[Augmenting vertex]
For the pair $\sigma, \tau \in {\cal S}_k(G(n,d/n))$ the vertex
$v \in V\backslash (\sigma \cup \tau)$ is {\em augmenting} if
one of the following $A$, $B$ holds.
\begin{description}
	\item[A.] $N_v \cap (\sigma \cup \tau)=\emptyset$
	\item[B.] $N_v \cap (\sigma \cap \tau)=\emptyset$ and there are
	 {\em terminal sets} $I_v(\sigma)$ and $I_v(\tau)$ of size at most $7d$ such that
	\begin{itemize}
		\item $I_v(\sigma)\cup\{v\}$ is an independent set of $G(n,d/n)$
		\item $|I_v(\sigma)|=|N_v\cap \sigma|$  
		\item $\forall w \in I_v(\sigma)$ it holds that $|N_w\cap \sigma|=1$ and $|N_w\cap N_u \cap \sigma|=1$
	\end{itemize}
	The corresponding conditions should hold for $I_{v}(\tau)$, as well.
\end{description}
\end{definition} 

\noindent
Figure \ref{fig:G-Property} shows an example of a pair of independent
sets where the vertex $v$ is an {\em augmenting vertex}.

We will show that for a pair $\sigma,\tau \in {\cal S}_k(G(n,d/n))$
that has an {\em augmenting vertex $v$} we can find short chains
$\sigma, \sigma', \sigma ''$  and  $\tau, \tau', \tau''$. That is,
if we can find an augmenting vertex for any two members of ${\cal S}_k(G(n,d/n))$,
then ${\cal S}_k(G(n,d/n))$ has Property $\Gamma$.

First, let us show how  we can create  short chains as in
Figure \ref{fig:Chains} for two independent sets $\sigma,\tau$
with augmenting vertex $v$. For this, we introduce a process
called {\em Collider}.  This process takes as an input $\sigma$, 
$\tau$  and the augmenting vertex  $v$ and returns the independent
sets $\sigma''$ and $\tau''$ of the chains. \\ \vspace{-.3cm}

\noindent
{\bf Collider $(\sigma, \tau, v)$:} 
\\ \vspace{-.3cm}\\
{\em  Phase 1.} \hspace*{0.5cm} /*Creation of $\sigma'$ and $\tau'$.*/ \vspace{-.1cm} 
\begin{enumerate}
\item Derive $\sigma'$ from $\sigma$ by removing the all its vertices in $N_{v}\cap \sigma$
and by inserting $\{v\} \cup I_{v}(\sigma)$.
	\item Do the same for $\tau'$. 
\end{enumerate}
{\em  Phase 2.} \hspace*{0.5cm} /* Creation of $\sigma''$ and $\tau''$*/. \vspace{-.1cm}
\begin{enumerate}
	\item $\sigma''$ is derived from $\sigma'$ by deleting one  (any) vertex
	from $\sigma'\backslash \tau'$.
	\item $\tau''$ is derived from $\tau'$ by deleting one  (any) vertex 
	from $\tau'\backslash \sigma'$.
\end{enumerate}
 {\em Return}  $\sigma''$ and $\tau''$.\\ \vspace{-.3cm}\\
{\bf End}\\ \vspace{-.3cm}

\noindent
Figure \ref{fig:Phase1} shows the changes that have taken place 
to the independent sets in Figure \ref{fig:G-Property} at the 
end of ``Phase 1''. 
Note that after Phase 1 both $\sigma', \tau'$ contain the 
augmenting vertex $v$, i.e. the overlap has increased as 
$\sigma'\cap \tau'=(\sigma\cap\tau)\cup\{v\}$.
\begin{figure}
\begin{minipage}{0.5\textwidth}
	\centering
		\includegraphics[width=0.6\textwidth]{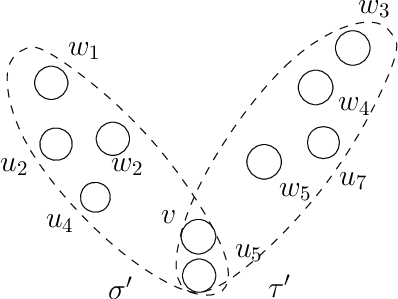}
	\caption{The independent sets $\sigma' $, $\tau'$. }
	\label{fig:Phase1}
\end{minipage}
\begin{minipage}{0.5\textwidth}
	\centering
		\includegraphics[width=0.6\textwidth]{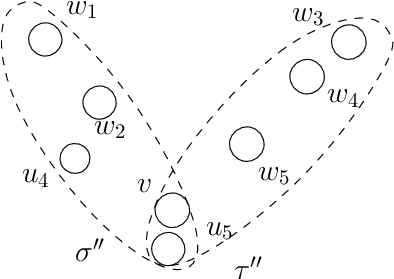}
	\caption{Final sets }
	\label{fig:finalstep}
\end{minipage}
\end{figure}
After ``Phase 2'', the independent sets in  Figure \ref{fig:Phase1}
are transformed to those in Figure \ref{fig:finalstep}. There the
vertices $u_2$ and $u_7$ are removed from $\sigma'$ and $\tau'$,
correspondingly.

In the following lemma we show that {\em Collider} has all the
desired properties we promise above.

\begin{lemma}\label{lemma:TrnsfrmProperties}
Let $\sigma, \tau \in {\cal S}_k(G)$ with augmenting vertex $v$.
Let $\sigma''$ and $\tau''$ be the two sets of vertices that are 
returned from {\em Collider($\sigma, \tau, v$)} .
The two sets have the following properties:
\begin{enumerate}
	\item $\sigma'', \tau'' \in {\cal S}_k(G)$,
	\item $|\sigma''\cap \tau''|=|\sigma \cap \tau|+1$,
	\item There are  $\sigma', \tau'\in {\cal S}_{k+1}(G)$
	such that $\sigma'$ (resp. $\tau'$) is adjacent to both 
	$\sigma$ and $\sigma''$ (resp. $\tau$ and $\tau'$).
\end{enumerate}
\end{lemma}
\begin{proof}
First we show that $\sigma''$ and $\tau''$, as returned by {\em Collider
$(\sigma, \tau, v)$}, are independent sets. The same arguments
apply to both $\sigma''$ and $\tau''$. For this reason we only consider 
the case of $\sigma''$, the other case would then be obvious.

Let $v$ be an augmenting vertex for the pair $\sigma$, $\tau$.
Assume that $\sigma''$, at the end of the process, is not 
an independent set, i.e. there is an edge between
two vertices in $\sigma''$. Clearly, this edge must be either 
between two new vertices, i.e. $\{v\}\cup I_v(\sigma)$, or
between some  newly inserted vertex  and an old one.

The first case cannot be true since the assumption that $v$
is an augmenting vertex implies $\{v\}\cup I_{\sigma}(v)$
is an independent set. As far as the second case is considered note that both $v$
and $I_v(\sigma)$ have the same neighbours in $\sigma$. During the
process {\em Collider}$(\sigma, \tau, v)$ all the vertices in 
$\sigma$ that are adjacent to $v$ and $I_v(\sigma)$ are removed
(Phase 1, step 1).  The second case cannot occur either. 
Thus $\sigma''$ and $\tau''$ are independent sets.

For showing Property 1 it suffices to show that $|\sigma''|=|\tau''|=k$.
This is straightforward by just counting how many vertices are inserted
into $\sigma$ (resp. $\tau$) and how many are removed. Property 2
follows by noting that $\sigma''\cap \tau''=(\sigma\cap \tau)\cup\{v\}$.
Property 3 follows directly by noting that $|I_v(\sigma)|$ and
$|I_v(\tau)|$ are at most $7d$.
\end{proof}

\noindent
Since for every pair $\sigma,\tau\in {\cal S}_k(G(n,d/n))$ with
augmenting vertex we can construct short chains as in Figure
\ref{fig:Chains} by using Collider, we have the following corollary:

\begin{corollary}\label{cor:Augment2GProperty}
If for any two $\sigma,\tau \in {\cal S}_k(G)$ there is an
{\em augmenting} vertex $v$, then ${\cal S}_k(G)$ has Property
$\Gamma$.
\end{corollary}

We are going to use the first moment method to show that with
probability $1-o(1)$, the graph $G(n,d/n)$ has no pair of
independent sets in ${\cal S}_{k}(G(n,d/n))$ with no augmenting
vertex. According to Corollary \ref{cor:Augment2GProperty}, this
implies that with probability $1-o(1)$ the set ${\cal S}_k
(G(n,d/n))$ has Property $\Gamma$. Then, Theorem
\ref{thrm:ConnectivityGnp} follows from Corollary \ref{cor:GPropertyToConnect}.

We compute, first,  the probability for a pair in ${\cal S}_k(G(n,d/n))$
to have an augmenting vertex.

\begin{proposition}\label{lemma:PropertyGVertex}
For some integers $i,k$, consider $\sigma,\tau$, two sets of
vertices each of size $k$ such that $|\sigma \cap \tau|=i$. Let
$G_{\sigma,\tau}$ denote $G(n,d/n)$ conditional that each of
$\sigma,\tau$ is an independent set. Also, let  $p_{k,i}$ be
the probability that the pair $\sigma, \tau$ has an {\em augmenting
vertex} in $G_{\sigma,\tau}$. Then, there exists $\epsilon_d\to 0$
such that for any $\epsilon_d\leq\epsilon\leq 1-\epsilon_d$ and
$k=(1-\epsilon)\frac{\ln d}{d}n$ the following is true
\begin{displaymath}
p_{k,i}\geq 1-\exp\left(-\frac{\ln^{90}d}{d}n \right).
\end{displaymath}
\end{proposition}
The proof of Proposition \ref{lemma:PropertyGVertex} appears in 
Section \ref{sec:PropoGVertex}.\\

\begin{theoremproof}{\ref{thrm:ConnectivityGnp}}
Let $Z_k$ be the number of pairs of independent sets of size $k$
in $G(n,d/n)$  that do not have an augmenting vertex. From Corollary
\ref{cor:Augment2GProperty} and Corollary \ref{cor:GPropertyToConnect},
it suffice to show that $Pr\left[\sum_{k\leq K}Z_{k}>0\right]=o(1)$,
where $K=(1-\epsilon_d)n\ln d/d$ and $\epsilon_d\to 0$ with $d$.
For this, we are going to use Markov's inequality, i.e.
$Pr\left[\sum_{k\leq K}Z_{k}>0\right]\leq E\left[\sum_{k\leq K} Z_{k}\right]$
and we are going to show that $E\left[\sum_{k\leq K} Z_{k}\right]=o(1)$.

First consider the case where $\frac{1}{10}\frac{\ln d}{d}n\leq k\leq (1-\epsilon_d)\frac{\ln d}{d}n$
and $\epsilon_d$ is as defined in the statement of Proposition \ref{lemma:PropertyGVertex}.
 Using Proposition \ref{lemma:PropertyGVertex} we get that 
\begin{equation}\label{eq:ExpctPairsNonAugm}
E[Z_{k}]\leq  \displaystyle {n \choose k}^2\exp \left(-\frac{\ln^{90}d}{d}n\right).
\end{equation}
It follows easily that
${n \choose k}^2\leq {n \choose \frac{\ln d}{d}n}^2\leq  \left(\frac{de}{\log d}\right)^{2\frac{\ln d}{d}n}=\exp\left( 3n\ln^2 d/d\right).$
Thus, from (\ref{eq:ExpctPairsNonAugm}) we get that there is 
$\epsilon_d\to 0$ with $d$ such that 
$$
E[Z_k]\leq \exp\left(-0.5\frac{\ln^{90}d}{d}n \right),
$$
for any $k=(1-\epsilon)\frac{\ln d}{d}n$, where $\epsilon_d<\epsilon<1-\epsilon_d$.

Consider now the case where $k<n\ln d/(10d)$. For a pair of independent
sets any vertex that is not adjacent to the vertices of the pair is an
augmenting vertex. Let $\sigma$, $\tau$ be a pair of independent sets
each of size $k\leq(1-\epsilon) n \ln d/d$, for $\epsilon\geq 0.9$. 
Let $R_{\sigma,\tau}$ be the vertices not in $\sigma\cup \tau$ but 
not adjacent to any vertex in $\sigma \cup \tau$, as well.
Every $w\notin \sigma\cup\tau$  belongs to $R_{\sigma,\tau}$ independently of the
other vertices with probability at least  $(1-p)^{2k}=\left(d^{\epsilon}/d \right)^2$.
Thus, $E|R_{\sigma,\tau}|\geq (n-2k)(d^{\epsilon}/d)^2$. Using Chernoff bounds 
we get
\begin{displaymath}
Pr[|R_{\sigma,\tau}|=0]\leq \exp\left( -\frac{d^{2\epsilon}}{10d^2}n \right)\leq
\exp\left( -\frac{d^{0.8}}{10d}n \right) \qquad \mbox{[since $\epsilon>0.9$]}.
\end{displaymath}
Since $R_{\sigma,\tau}$ consists of {\em augmenting vertices} for
the pair $\sigma,\tau$, the probability for $\sigma,\tau $ not to 
have any  augmenting vertex is upper bounded by $Pr[|R_{\sigma,\tau}|=0]$.
For  $k<n\ln d/(10d)$ it holds that
\begin{displaymath}
E[Z_k]\leq {n \choose k}^2\exp\left( -\frac{d^{0.8}}{10d}n \right)\leq
\exp\left(3\frac{\ln^2 d}{d}n \right)\cdot \exp\left( -\frac{d^{0.8}}{10d}n \right)
\leq \exp\left( -\frac{d^{0.8}}{15d}n \right).
\end{displaymath}
The theorem follows.
\end{theoremproof}


\subsection{Proof of Proposition \ref{lemma:PropertyGVertex}}\label{sec:PropoGVertex}

Consider  an arbitrary pair $\sigma, \tau \in {\cal S}_k(G(n,d/n))$
where $k=(1-\epsilon)n\ln d/d$ and $100\frac{\ln\ln d}{\ln d}\leq \epsilon
\leq 1-100\frac{\ln\ln d}{\ln d}$.
For the rest of the proof assume that $|\sigma \cap \tau|=ak$ where
$a\in [0,1]$. Also, let $\epsilon'$ be such that $1-\epsilon'
=(1-a)(1-\epsilon)$. Clearly, it holds that $\epsilon' \in
\left[100\frac{\ln\ln d}{\ln d},1\right]$. 
For proving the proposition, we consider two cases.
In the first one we take $100\frac{\ln\ln d}{\ln d}\leq \epsilon'\leq 1-100
\frac{\ln\ln d}{\ln d}$. In the second  we take $ 1-100\frac{\ln\ln d}{\ln d}<\epsilon'\leq 1$.

Take $100\frac{\ln\ln d}{\ln d}\leq \epsilon'\leq 1-100
\frac{\ln\ln d}{\ln d}$. We will show that with sufficiently large
probability  there exists a non-empty set $Q_0$ of augmenting
vertices for the pair $\sigma$, $\tau$. The set $Q_0$ contains a
specific kind of augmenting vertices. That is, the cardinality of
$Q_0$ will be a lower bound on the actual number of augmenting
vertices. So as to specify $Q_0$, we need  the following definitions:

\begin{description}
\item[$\mathbf{Q_1(\sigma)}$:]  $Q_1(\sigma)\subseteq V \backslash (\sigma
\cup \tau)$ contains those vertices that have exactly one neighbour
in $\sigma\backslash \tau$.

\item [$\mathbf{Q_2(\sigma)}$:]  $Q_2(\sigma)\subseteq \sigma\backslash \tau$
is the  set of vertices that have at least one neighbour in $Q_1(\sigma)$.

\item [$\mathbf{Q_3(\sigma)}$:]  Every $w \in Q_3(\sigma)\subseteq V \backslash (\sigma 
\cup \tau \cup Q_1(\sigma))$ has the following properties: 
\begin{description}
	\item[S$_1$-] $N_w\cap (\sigma\backslash \tau)\subseteq Q_2(\sigma)$ and 
	$|N_w \cap (\sigma\backslash \tau) |\leq 7d$. 
	\item[S$_2$-] There exists $R\subseteq Q_1(\sigma)$ that contains exactly
	one neighbour of each $v \in N_w\cap (\sigma\backslash \tau)$ in $Q_1(\sigma)$
	and no other vertex. Furthermore, $R\cup \{w\}$ is an independent set.
\end{description}
\end{description}
In an analogous manner we define $Q_1(\tau),Q_2(\tau)$ and $Q_3(\tau)$.

For each augmenting vertex  $u\in Q_0$ the following should hold:
{\bf (A)} $u\in Q_3(\sigma)\cap Q_3(\tau)$, 
{\bf (B)} $N_u\cap (\sigma\backslash \tau)\subseteq Q_2(\sigma)$ and
$N_u\cap (\tau\backslash \sigma)\subseteq Q_2(\tau)$,
{\bf (C)} $I_v(\sigma)\subseteq Q_1(\sigma)\backslash Q_1(\tau)$  and  $I_v(\tau)\subseteq 
Q_1(\tau)\backslash Q_1(\sigma)$.
\\ \vspace{-.1cm}

\noindent
{\bf Remark.}
Observe that each $u\in Q_3(\sigma)\cap Q_3(\tau)$ is not necessarily
augmenting. However, if, additionally, $u$ it does not have any neighbours 
in  $\sigma \cap \tau$, then it is augmenting.
\\ \vspace{-.2cm}

\noindent
Consider a process where we reveal all the sets $Q_i(\sigma),Q_i(\tau)$,
for $i=1,2,3$ in steps. In each step we reveal a certain amount of
information regarding these six sets. Since $Q_i(\sigma)$ is symmetric 
to $Q_i(\tau)$ for every $i=1,2,3$ we just presents results related to 
$Q_i(\sigma)$ while those for $Q_i(\tau)$ follow immediately. The results 
appear as a series of claims whose proofs appear after the proof of this 
proposition.

In  {\tt Step 1}, we reveal the sets $Q_1(\sigma)$, $Q_1(\tau)$.  
There we have  the following result.

\begin{claim}\label{claim:Q1Claim}
Let $X_1=|Q_1(\sigma)\backslash Q_1(\tau)|$.  It holds that 
$E[X_1]=\frac{(1-\epsilon')\ln d}{d^{1-\epsilon'}}n(1-\epsilon_d)-O(1)$,
where $\epsilon_d\to 0$ as $d$ grows. Furthermore, it holds that 
\begin{displaymath}
Pr \left[\left|X_1-E[X_1]\right|\geq 0.5 E[X_1] \right]\leq 2\exp
\left(-n{d^{\epsilon'}}/{d} \right).
\end{displaymath}
\end{claim}

\noindent
{\bf Remark.}
After {\tt Step 1}, for each $v\in V\backslash  \{Q_1(\sigma)\cup 
Q_1(\tau)\cup\sigma\cup \tau\}$ we have the information that both
the number of edges that connect $v$ with $\sigma\backslash \tau$ 
and the number edges that connect $v$ with $\tau\backslash \sigma$
are different than 1. 
\\ \vspace{-.3cm}

\noindent
Then, we proceed with {\tt Step 2} where we reveal $Q_2(\sigma)$ and 
$Q_2(\tau)$. Also reveal the edges between $Q_2(\sigma)$ and $Q_1(\sigma)$
as well as the edges between $Q_2(\tau)$ and $Q_1(\tau)$. There we have 
the following result.

\begin{claim}\label{claim:Q2Claim}
Let $X_2=|Q_2(\sigma)|$. For $\gamma=1-\ln^{-5}d$, it holds that
\begin{displaymath}
Pr[X_2\leq \gamma\cdot |\sigma\backslash \tau | |{\cal F}_1]\leq \exp\left(-n{d^{\epsilon'}}/(4d\ln^5 d) \right),
\end{displaymath}
where ${\cal F}_1=\{|X_1-E[X_1]|< 0.5 E[X_1]\}$.
\end{claim}
%
%

\noindent
Revealing the sets $Q_3(\sigma)$ and $Q_3(\tau)$ is, technically,
a more complex task. Let us make some observations regarding these
sets. Assume that some vertex $u \in V\backslash (\sigma\cup \tau \cup Q_{1}(\sigma))$
satisfies condition\footnote{In the definition of set $Q_3(\sigma)$.}
$\mathbf{S_1}$. So as $u$ to belong to $Q_3(\sigma)$ there should exist a set $R\subseteq Q_1(\sigma)$
as specified in the condition $\mathbf{S_2}$. However, the possibility 
of edges between vertices in $Q_1(\sigma)$ leaves open whether we 
can have such a set for $u$. To this end consider the following.

\begin{definition}
For every $i=1\ldots 7d$, let ${\cal A}_i$ be the family of subsets
$B\subseteq Q_2(\sigma)$ of cardinality  $i$ which have the following 
property:  There exists independent set $R\subseteq Q_1(\sigma)$
that contains exactly one  neighbour of each $v \in B$ in $Q_1(\sigma)$
and no other vertex.
\end{definition}

\noindent
That is, a vertex $u$ which satisfies $\mathbf{S_1}$ satisfies  also $\mathbf{S_2}$ (
i.e.  belongs to $Q_3(\sigma)$) only if $N_u\cap(\sigma\backslash \tau)\in {\cal A}_i$, 
for some appropriate $i>0$ or $N_u\cap(\sigma\backslash \tau)=\emptyset$.
Observe that the families ${\cal A}_i$ are uniquely determined by the
edges whose both ends are in $Q_1(\sigma)$. In  {\tt Step 3} we reveal
exactly these edges, i.e. with both ends either in $Q_1(\sigma)$ or in $Q_1(\tau)$.
This results  to the following.

\begin{claim}\label{claim:GoodFraction}
Let ${\cal F}_2=\{{\cal F}_1$ and $X_2>\gamma\cdot |\sigma\backslash \tau| \}$.
For every $2\leq i\leq 7d$ it holds that
\begin{displaymath}
Pr\left[|{\cal A}_i| \leq (1-2d^5/n){|Q_2(\sigma)|\choose i} |{\cal F}_2\right]
\leq 2\exp\left( -nd^{2\epsilon'}/d \right).
\end{displaymath}
\end{claim}
It is direct to see that it always holds that ${\cal A}_1=Q_2(\sigma)$.

Let the set $V'=V\backslash (\sigma \cup \tau \cup Q_{1}(\sigma)\cup Q_{1}(\tau))$.
In {\tt Step 4}, we reveal the vertices that belong in $Q_3(\sigma)\cap Q_3(\tau)$.
This step amounts to revealing the edges between each vertex  $v\in V'$ 
and the sets $Q_i(\sigma)$ and $Q_i(\tau)$, for $i=1,2$.
In particular, revealing the edges between $v$ and the set $Q_1(\sigma)\cup Q_2(\sigma)$
(resp. $Q_1(\tau)\cup Q_2(\tau)$) specifies whether $v\in Q_3(\sigma)$
(resp. $v\in Q_3(\tau)$), or not.

Despite the information we have for $v\in V'$, from {\tt Step 1}, the edge
events between $v$ and the vertices in $Q_1(\sigma)\cup Q_2(\sigma)$ are 
independent of the edge events between $v$ and the vertices in
$Q_1(\tau)\cup Q_2(\tau)$. That is 
$Pr[v\in Q_3(\sigma)\cap Q_3(\tau)]=(Pr[v\in Q_3(\sigma)])^2$.
Also, it is easy to observe that $v\in Q_3(\sigma)\cap Q_3(\tau)$
independently of the other vertices in $V'$.

For every $v \in V'$ let $J_v$ be an indicator random variable such that
$J_v=1$ if $v\in Q_3(\sigma)\cap Q_3(\tau)$ and $J_v=0$ otherwise.
The observations in the previous paragraph suggest that $J_v$s  
are independent with each other and $E[J_v]=Pr[v \in Q_3(\sigma)]^2$.

\begin{claim}\label{Lemma:Q3}
Let the event  ${\cal F}_3=\left \{ {\cal F}_2 \textrm{ and }
 |{\cal A}_i| \geq (1-2d^5/n){|Q_2(\sigma)|\choose i}\right\}$.
For every $u \in V'$, it holds that $$Pr[u \in Q_3(\sigma)|{\cal F}_3] \geq 9/10.$$
\end{claim}
%
%

\noindent
Let $X_3=\sum_{v}J_v$ where $v$ varies over all vertices in $V'$.  
Using Claim \ref{claim:Q1Claim}  and Claim \ref{Lemma:Q3} we get that
\begin{eqnarray*}
E[X_3|{\cal F}_3]&\geq& \left(1-{10d^{\epsilon'}\ln d}/{d}\right)n
\cdot Pr^2[u \in Q_3(\sigma)|{\cal F}_3]
\geq 8n/10,
\end{eqnarray*}
Applying Chernoff bounds and get that 
\begin{equation}\label{eq:X3Chernoff}
Pr[X_3 < 0.7 n|{\cal F}_3 ]\leq \exp \left(-n/{350} \right).
\end{equation}

\noindent
Finally, in {\tt Step 5} we reveal which vertices in $Q_3(\sigma)\cap Q_3(\tau)$
are augmenting, i.e. those which are adjacent to $\sigma\cap \tau$.
Only these vertices will belong the set $Q_0$.

Due to edge independence in $G(n,d/n)$, every $u\in Q_3(\sigma)\cap Q_3(\tau)$
is augmenting independently of all the rest vertices  with probability 
$d^{-a(1-\epsilon)}+O(n^{-1})$. 
Let the event ${\cal F}_4=\{ {\cal F}_3$ and $X_3 \geq 0.7n \}$. It is direct
that $E[|Q_0||{\cal F}_4]\geq 0.7 nd^{-a(1-\epsilon)}-O(1)$.
Since $a \in [0,1]$, there exists $\delta=\delta(\epsilon,a)> \epsilon$ 
such that $a(1-\epsilon)=1-\delta$. Applying Chernoff bounds we get that
\begin{equation}\label{eq:ConditionalF3Q0}
Pr[|Q_0 |= 0 |{\cal F}_4]\leq \exp\left( -0.2 {d^{\delta}n}/{d} \right)\leq
\exp\left( -0.2 {d^{\epsilon}n}/{d} \right) \qquad \mbox{[as $\delta>\epsilon$].}
\end{equation}
Using Claim \ref{claim:Q1Claim}, Claim \ref{claim:Q2Claim}, Claim \ref{claim:GoodFraction} 
and (\ref{eq:X3Chernoff}) we get that $Pr[{\cal F}_4]\geq 1-20\exp\left(-n{d^{\epsilon'}}/(4d\ln^5 d)\right)$.
Combining the probability bound for $Pr[{\cal F}_4]$ with (\ref{eq:ConditionalF3Q0})
we get that
\begin{equation}\label{eq:Q0BoundForSmallA}
Pr[|Q_0| = 0]\leq 30\exp\left(-n{d^{\epsilon'}}/(3d\ln^5 d)\right)
\leq \exp\left(-n\ln^{90}d/d\right)
\end{equation}
as $100\frac{\ln\ln d}{\ln d}\leq \epsilon'\leq 1-100\frac{\ln\ln d}{\ln d}$.

It remains to study the case where $1-100\frac{\ln\ln d}{\ln d}<\epsilon'\leq 1$.
There, it holds that $|\sigma\cup \tau|=k_0\leq (1-\epsilon)
\frac{\ln d}{d}n+100\frac{\ln\ln d}{d}n$. Let $R_{\sigma,\tau}$ be the 
set of vertices, outside $\sigma, \tau$, that are not adjacent to any
vertex in $\sigma\cup \tau$. 
Every $w\notin \sigma\cup\tau$  belongs to $R_{\sigma,\tau}$ independently of the
other vertices with probability $(1-p)^{k_0}\leq \left(d^{\epsilon/2}/d \right)$.
Thus, $E|R_{\sigma,\tau}|\geq (n-k_0)d^{\epsilon/2}/d$. Using Chernoff bounds 
we get
\begin{equation}\label{eq:Q0BoundForBigA}
Pr[|R_{\sigma,\tau}|=0]\leq \exp\left( -\frac{d^{\epsilon/2}}{2d}n \right).
\end{equation}
Since $R_{\sigma,\tau}$ consists of {\em augmenting vertices} for
the pair $\sigma,\tau$, the probability that there is no augmenting
vertex is upper bounded by $Pr[|R_{\sigma,\tau}|=0]$. The proposition
follows from (\ref{eq:Q0BoundForSmallA}) and (\ref{eq:Q0BoundForBigA}).
 \qed \\

\begin{claimproof}{\ref{claim:Q1Claim}}
Let $r$ be the probability for a vertex $v$ outside $\sigma,\tau$,
to have exactly one neighbour in $\sigma\backslash \tau$. It holds
that 
$$r=(1-a)kp(1-p)^{(1-a)k-1}=(1-\epsilon')\ln d/d^{1-\epsilon'}-O(n^{-1}).$$
Of course, with the same probability  $v$ has exactly one neighbour
in $\tau\backslash \sigma$. Then, the probability for $v$ to be in 
$Q_1(\sigma)\backslash Q_1(\tau)$ is $p_1=r(1-r)$. Observe that $v$ belongs
to $Q_1(\sigma)\backslash Q_1(\tau)$ independently of the other vertices.
It is direct that there exists
$\epsilon_d\to 0$ such that 
\begin{displaymath}
E[X_1]=(n-2k)p_1=\frac{(1-\epsilon')\ln d }{d^{1-\epsilon'}}n\left(1-\epsilon_d\right) -O(1).
\end{displaymath}
The claim follows by applying the Chernoff bounds.
\end{claimproof}

\begin{claimproof}{\ref{claim:Q2Claim}}
Due to symmetry each vertex $u\in Q_1(\sigma)$ is adjacent to
exactly one random vertex in $\sigma\backslash \tau$, independently
of the other vertices in $Q_1(\sigma)$. An equivalent way of
looking adjacencies between vertices in $Q_1(\sigma)$ and
$\sigma\backslash \tau$ is by assuming that the vertices in
$Q_1(\sigma)$ are balls and each vertex in $\sigma\backslash
\tau$ is a bin and each ball is thrown into a {\em random} bin.
The non-empty bins correspond to vertices in $Q_2(\sigma)$. The
claim will follow by deriving an appropriate tail bound on the
number of occupied bins.

Let $N$ denote the number or balls and $m$ denote the number of
bins, it holds that $N\geq \frac{d^{\epsilon'}}{d}n$ and $m=(1-\epsilon')\frac{\ln d}{d}n$.
For $c\in (0,1)$, let $P_c$ be the probability that there is a
subset of bins of size $cm$ that contains all the balls. For $B_c$
a fixed subset of bins of size $cm$ and for a fixed ball $r$, it
holds that
\begin{eqnarray}
P_c&\leq& {m \choose cm}\left (Pr[r \textrm{ is placed into some bin in } B_c]\right)^N
\leq \left(\frac{me}{cm}\right)^{cm}c^N \nonumber \\
&\leq & \exp\left(cm(1-\ln c)+ N\ln c\right).\nonumber 
\end{eqnarray}
For $c_0=(1-\ln^{-5} d)$ we have that
\begin{eqnarray}
P_{c_0}&\leq& \exp\left( 2\frac{\ln d}{d}n -\frac{d^{\epsilon'}}{2d\ln^5 d}n \right)
\qquad \mbox{[as $1-x\geq \exp(-x/(1-x)$ for $0<x<0.1$]}\nonumber\\
&\leq & \exp\left(-nd^{\epsilon'}/(3d\ln^5 d)\right) \hspace{6.35cm} \mbox{[for large $d$]}.
\nonumber
\end{eqnarray}
It is easy to check that for any $0\leq c\leq c_0$ we have $P_c\leq
P_{c_0}$. Hence, letting $E_{c_0}$ be the event that ``there is a 
subset of at most $c_0\cdot m$ bins that has all the balls'', it
holds that
$$P[E_{c_0}]\leq \exp\left( -nd^{\epsilon'}/(4d\ln^5 d) \right).$$
The claim follows.
\end{claimproof}

\begin{claimproof}{\ref{claim:GoodFraction}}
The cardinality of each family ${\cal A}_i$, for $2\leq i\leq 7d$,
depends on the edges whose both ends are in $Q_1(\sigma)$. 
As a  first step we estimate how many are these vertices 
conditional on the event ${\cal F}_2$.

Let $R_1$ be the set of edges whose both ends are in $Q_1(\sigma)$.
The bound on $X_1$ and the cardinality of $Q_1(\sigma)$ that
${\cal F}_2$ specifies as well as the fact that each edge appears
independently with probability $d/n$ yields to the following relation.
\begin{displaymath}
E[|R_1||{\cal F}_2]=C\frac{d^{2\epsilon'}n}{d}  (1-\epsilon')^2 \ln^2 d,
\end{displaymath}
where $1/8<C<9/8$. Chernoff bounds yield to the following inequality.
\begin{equation}\label{AlmostNoOfEdgesInQ1}
Pr\left[|R_1|\geq n/d^{1-3\epsilon'} |{\cal F}_2\right]\leq \exp\left( -nd^{2\epsilon'}/d \right).
\end{equation}

\noindent
Let the event $H=\{ {\cal F}_2$  and $|R_1|< n/d^{1-3\epsilon'}\}$. 

Next, we compute $E[|{\cal A}_i||H]$.  Note that the event $H$
specifies, only, an upper bound on $|R_1|$ and it does not tell
where the edges are placed. That is, all subsets of $Q_2(\sigma)$
of cardinality $i$ are symmetric thus they belong to ${\cal A}_i$
with the same probability. By the linearity  of expect we get that
\begin{displaymath}
E[|{\cal A}_i|H]={|Q_2(\sigma)| \choose i}Pr[L\notin {\cal A}_i|H]
\qquad\qquad \mbox{[for a fixed  $L\subseteq Q_2(\sigma)$ and $|L|=i$]}
\end{displaymath}

\noindent
Let $M_{L}$ be the family of subsets of $Q_1(\sigma)$, each of 
cardinality $i$, such that for each ${\cal W}\in M_L$ the following
is true:  The set ${\cal W}$  contains  exactly one neighbour of
each vertex $q\in L$ and no other vertex. By definition the 
family $M_L$ must have at least one member. Moreover, if there 
exists one set in $M_L$ which is independent, then $L\in {\cal A}_i$.

When we reveal the edges between the vertices in $Q_1(\sigma)$
it is easy to see that the probability that $M_L$ contains no 
independent set is maximized when $M_L$ is a singleton.
Given $|R_1|$ and $X_1$, observe that each pair of vertices in 
$Q_1(\sigma)$ is adjacent with probability  at most 
$|R_1|/{X_1\choose 2}$.  Each subset of $Q_1(\sigma)$ of cardinality $i$
has expected number of adjacent vertices ${i \choose 2}|R_1|/{X_1\choose 2} \leq d^4/n$,
for large $d$.  That is, the probability that $M_L$ does not contain
an independent set is  at most $d^4/n$. Thus, 
\begin{equation}\label{eq:ExpectCalAi}
E[|{\cal A}_i|H]\geq \left(1-\frac{d^4}{n}\right){|Q_2(\sigma)| \choose i}.
\end{equation}

\noindent
Having calculated a lower bound for $E[|{\cal A}_i||H]$ we will show that
given the event $H$,  $|{\cal A}_i|$ is tightly concentrated about its 
expectation. Then, claim will be immediate. So as to show the concentration
result,  we use an edge exposure martingale argument for the edges in $R_1$ 
and then we apply Azuma's inequality (see e.g. \cite{janson} Theorem 2.25).

Observe that the revelation of each edge in $R_1$ cannot reduce the
cardinality of ${\cal A}_i$ by more than $c={X_2-2 \choose i-2}\leq (X_2)^{i-2}/(i-2)!$
sets. Standard arguments with Azuma's inequality  yield to that for any 
$\lambda>0$ it holds that
\begin{eqnarray}
Pr\left[|{\cal A}_i|\leq E[{\cal A}_i|H]-\lambda|H\right ]
&\leq& \exp\left(-\frac{\lambda^2}{2|R_1|c^2}\right). \nonumber 
\end{eqnarray}
Setting $\lambda=d^4X_2^{i-1}/i!$ we get that
\begin{eqnarray}
Pr\left[|{\cal A}_i|\leq \left(1-2\frac{d^5}{n}\right){Q_2(\sigma) \choose i}|H\right ]
&\leq&
\exp\left(-\frac{d^8X_2^{2}}{2|R_1|i^2}\right) 
\leq \exp\left(-dn\right), 
\nonumber 
\end{eqnarray}
where the last derivation follows by using the fact that $1\leq i\leq 7 d$,
$|R_1|\leq n/d^{1-3\epsilon'}$ and $100\ln\ln d/\ln d<1-\epsilon'<1-100\ln\ln d/\ln d$. 
The claim follows by just using the law of total probability and get that
\begin{eqnarray}
Pr\left[|{\cal A}_i|\leq \left(1-2\frac{d^5}{n}\right){Q_2(\sigma) \choose i}|{\cal F}_2\right ]
&\leq& Pr\left[|{\cal A}_i|\leq \left(1-2\frac{d^5}{n}\right){Q_2(\sigma) \choose i}|H\right ]+
Pr\left[|R_1|\geq n/d^{1-3\epsilon'} |{\cal F}_2\right]    \nonumber \\
&\leq & 2\exp\left( -nd^{2\epsilon'}/d \right). \nonumber 
\end{eqnarray}
\end{claimproof}

\begin{claimproof}{\ref{Lemma:Q3}}
Consider some $u\in V \backslash (\sigma \cup \tau \cup Q_1(\sigma))$.
Let $d_{\sigma,\tau}(u)$ be the number of vertices in $\sigma
\backslash \tau$ which are adjacent to $u$. Also, let 
the event $E_i=\{N_u\cap (\sigma\backslash\tau)\in {\cal A}_i\}$ for $i>0$
and $E_0=\{N_u\cap (\sigma\backslash\tau)=\emptyset\}$.
By the law of total probability we get that
\begin{eqnarray}
Pr[u \in Q_3(\sigma)|{\cal F}_3] 
&\geq& \sum_{i=0}^{7d}Pr[u\in Q_3(\sigma)|d_{\sigma,\tau}=i, E_i, {\cal F}_3]
\cdot Pr[E_i|d_{\sigma,\tau}=i,{\cal F}_3]  \cdot Pr[d_{\sigma,\tau}=i|{\cal F}_3].
\qquad \label{eq:PruQ3FirstBound}
\end{eqnarray}
We impose the bound $i\leq 7d$  since no vertex in $Q_3(\sigma)$
can have more than $7d$ neighbours in $Q_2(\sigma)$.
Conditional on $d_{\sigma, \tau}(u)=i$,  all the
subsets of size $i$ in $\sigma \backslash \tau$ are equiprobably adjacent to $u$.
Thus, we get that
\begin{eqnarray}
Pr[E_i|d_{\sigma,\tau}=i,{\cal F}_3]&=&\frac{|{\cal A}_i|}{{|\sigma\backslash \tau|\choose i}}
\geq (1-2d^5/n)\frac{{X_2\choose i}}{{|\sigma\backslash \tau|\choose i}} 
\qquad \qquad \mbox{[by Claim \ref{claim:GoodFraction}]} 
\nonumber \\
&\geq& \left( \frac{X_2}
{|\sigma\backslash \tau|}\right)^i(1-o(1))\geq  \gamma^i(1-o(1)),
\label{eq:4TargetOneH1}
\end{eqnarray}
where $\gamma=1-\ln^{-5} d$. Also, it is easy to see that 
\begin{equation}
Pr[u\in Q_3(\sigma)|d_{\sigma,\tau}=i, E_i, {\cal F}_3]\geq (1-d/n)^i\geq 1-7d^2/n.
\qquad \qquad \mbox{[as $0\leq i\leq 7d$]}
\label{eq:Target1stBound}
\end{equation}

\noindent
Let the event $C=$``$d_{\sigma, \tau}(u) \neq 1$ and $d_{\sigma, \tau}(u)\leq 7d$''.
Observe that the variable $d_{\sigma, \tau}(u)$  is distributed as in ${\cal B}((1-a)k,d/n)$
conditional on the event $C$. Using this  along with 
(\ref{eq:Target1stBound}) and (\ref{eq:4TargetOneH1})
we can rewrite (\ref{eq:PruQ3FirstBound}) as follows:
\begin{eqnarray}
Pr[u \in Q_3|{\cal F}_3] &\geq & \frac{1-o(1)}{Pr[C|{\cal F}_3]}
\left [\sum_{i=0}^{7d}{(1-a)k \choose i}p^{i}(1-p)^{(1-a)k-i}\gamma^i
-\gamma{(1-a)k \choose 1}p(1-p)^{(1-a)k-1} \right]
\nonumber \\
&\geq & (1-o(1))
\left [\sum_{i=0}^{7d}
{(1-a)k \choose i}p^{i}(1-p)^{(1-a)k-i}\gamma^{i} -d^{-(1-\epsilon')} \ln d \right],
 \label{eq:PQ3OneBeforeEnd}
\end{eqnarray}
where the last inequality follows from the fact that $\gamma, Pr[C|{\cal F}_3]\leq 1$
and a simple derivation which implies that 
${(1-a)k \choose 1}p(1-p)^{(1-a)k-1}\leq d^{-(1-\epsilon')} \ln d$.
Also, note that
\begin{eqnarray}
\sum_{i=7d+1}^{(1-a)k}{(1-a)k \choose i}p^{i}(1-p)^{(1-a)k-i}\gamma^i&\leq &
\sum_{i=7d+1}^{(1-a)k}{(1-a)k \choose i}p^{i}(1-p)^{(1-a)k-i} \qquad \mbox{[as $0\leq \gamma <1$]}
\nonumber \\
&\leq& \exp\left(-7d\right).\label{eq:ComplementBound}
\end{eqnarray}
The last inequality follows by noting that the summation on the l.h.s. of the first line
is equal to the probability $Pr[{\cal B}((1-a)k,d/n)>7d]$ and  bounding it by using
Chernoff bound  (as it appears in Theorem 2.1 in \cite{janson}).
Using (\ref{eq:ComplementBound}), we get that
\begin{eqnarray}
\sum_{i=0}^{7d}{(1-a)k \choose i}p^{i}(1-p)^{(1-a)k-i}\gamma^{i}
&\geq&(1-p \ln^{-5}d )^{(1-a)k}-\exp(-7d)
\nonumber \\ 
&\hspace{-3.2cm}\geq &\hspace{-1.9cm}\exp\left[ -(1-\epsilon')\ln^{-4} d - O(n^{-1}) \right] -\exp(-7d)
 \qquad \mbox{[as $\ln(1-x)=-x-O(x^2)$]}
\nonumber \\ 
&\hspace{-3.2cm}\geq &\hspace{-1.9cm} 1-\frac{1-\epsilon'}{\ln^4 d}-\exp(-7d)-O(n^{-1})
\hspace{3cm} \mbox{[as $1+x\leq e^x$]} \nonumber
\nonumber \\
&\hspace{-3.2cm}\geq &\hspace{-1.9cm} 95/100. \label{eq:sumupto7dBin}
\end{eqnarray}
The claim follows by plugging (\ref{eq:sumupto7dBin}) into (\ref{eq:PQ3OneBeforeEnd}) and
get that $Pr[u \in Q_3|{\cal F}_3] \geq 9/10$.
\end{claimproof}

\section{Proof of Theorem \ref{theorem:shattering}}\label{sec:theorem:shattering}

The following proposition reduces the problem of establishing shattering to an exercise in calculus.

\begin{proposition}\label{Prop_calcShattering}
There exist a constant $d_0>0$ and $\eps_d\ra0$ such that for all
$d>d_0$ the following is true. Suppose that $s=(1+q)\ln d/d$  for
$\epsilon_d\leq q\leq (1-\epsilon_d)$ and let
	$$\psi(x)=\psi_{d,s}(x)=xs(2-2\ln x-\ln s)+\frac d2\ln\bc{1-\frac{s^2(1-(1-x)^2)}{1-s^2}}.$$
If there is a real $0<b<1$ such that
	\begin{eqnarray}\label{eqcalcShattering1}
	\psi(b)&<&-18qs
	\qquad\mbox{and}\\
	\sup_{x<b}\psi(x)&<&-s\ln(s)-(1-s)\ln(1-s)+\frac{d}2\ln(1-s^2)-20s
		\label{eqcalcShattering2}
	\end{eqnarray}
then $\cS_k(G(n,m))$ shatters, with $m=dn/2$ and $k=sn$.
\end{proposition}

\noindent
{\bf Proof of \Thm~\ref{theorem:shattering} (assuming \Prop~\ref{Prop_calcShattering}):}
Let $\eps_d$ be as in \Prop~\ref{Prop_calcShattering}, assume that $d>d_0$ is sufficiently large,
let $\delta=5\ln\ln d/\ln d$ and set
	$$k=sn\qquad\mbox{ with }\qquad\frac{(1+\delta)\ln d}d\leq s\leq\frac{(2-\eps_d)\ln d}d.$$
Moreover, let $b=20\ln^{-1}d$.
We are going to verify~(\ref{eqcalcShattering1}) and~(\ref{eqcalcShattering2}).
Then \Thm~\ref{theorem:shattering} will follow from \Prop~\ref{Prop_calcShattering}.
Indeed, using the elementary inequality $\ln(1-x)\leq-x$, we find
	\begin{eqnarray}
	\psi(x)&\leq&sx(2-2\ln x-\ln s)-\frac {ds^2}2(1-(1-x)^2)\nonumber\\
		&=&sx\bc{2-2\ln x-\ln s-ds+dsx/2}\nonumber\\
		&\leq&sx\bc{2-2\ln x-\ln d-ds+dsx/2}\hspace*{2.65cm}\mbox{[as $s\geq\ln d/d$]}\nonumber\\
		&\leq&sx\bc{2-2\ln x-\delta\ln d+dsx/2}\hspace*{2.1cm}\mbox{[as $s\geq(1+\delta)\ln d/d$]}.
			\label{eqpftheorem:shattering1}
	\end{eqnarray}
Hence, for $d\geq d_0$ sufficiently large our choice of $\delta,b$ ensures that
	\begin{eqnarray*}
\psi(b)&\leq & b s\left(22+2\ln\ln d-\ln20-q\ln d\right)\leq -\frac{9}{10}bsq\ln d\leq -18qs. 
	\end{eqnarray*}
Thus, we have verified~(\ref{eqcalcShattering1}).

Starting from~(\ref{eqpftheorem:shattering1}), we see that for any $\beta<b$ and $d>d_0$ large,
	\begin{eqnarray}\nonumber
	\psi(\beta)&\leq&\beta s(22-2\ln\beta-100\ln\ln d)\qquad\qquad\mbox{[as $\beta ds<40$ and by the choice of $\delta$]}\\
		&\leq&-2\beta s\ln\beta<s, 			\label{eqpftheorem:shattering2}
	\end{eqnarray}
because $-x\ln x<1/2$ for all $x>0$.
By comparison, for $s\leq(2-\delta)\ln d/d$ we have
	\begin{eqnarray}
	-s\ln(s)-(1-s)\ln(1-s)+\frac{d}2\ln(1-s^2)&\geq&
		-s\ln s+s-\frac{ds^2}2-\frac{ds^4}2
			\quad\mbox{[using $\ln(1-x)\geq-x-x^2$]}\nonumber\\
		&\geq&s\bc{-\ln s-ds/2+1}\nonumber\\
		&\geq&s\bc{\frac{1-q}{2}\ln d-\ln\ln d+1}\geq40s\ln\ln d.
			 			\label{eqpftheorem:shattering3}
	\end{eqnarray}
Combining~(\ref{eqpftheorem:shattering2}) and~(\ref{eqpftheorem:shattering3}), we obtain
	$$\psi(\beta)<-s\ln(s)-(1-s)\ln(1-s)+\frac{d}2\ln(1-s^2)-s<-s\ln(s)-(1-s)\ln(1-s)+\frac{d}2\ln(1-s^2)-20s
	$$
as $s\geq\ln d/d$.
Thus, we have got~(\ref{eqcalcShattering2}).

Lemma \ref{lemma:SameAsShaterringThrm} (in a following section)
states explicitly what is implied in this proof.  That is there
exists $0<b<1$ such that (\ref{eqcalcShattering1}) and
(\ref{eqcalcShattering2}) hold.
Thus, we are going to use the proof here for Lemma 
\ref{lemma:SameAsShaterringThrm}.
\qed

\subsection{Proof of \Prop~\ref{Prop_calcShattering}}

Let $(G,\sigma)$ be a pair chosen from the planted model $\cP_k(n,m)$.
To prove the proposition, we are going to show that under the
assumptions~(\ref{eqcalcShattering1}) and~(\ref{eqcalcShattering2})
the independent set $\sigma$ belongs to a small `cluster' of
independent sets that is separated from the others by a linear
Hamming distance with a probability very close to one.
We will then use \Thm~\ref{thrm:transfer-theorem} to transfer
this result to the distribution $\cU_k(n,m)$, which will imply
that $\cS_k(G(n,m))$ shatters \whp

Let $Z_{k,\beta}$ be the number of independent sets $\tau\in\cS_k(G)$
such that
	$|\sigma\cap\tau|=(1-\beta)k$.

\begin{lemma}\label{Lemma_condexp}
We have
	$\frac1n\ln\Erw_{\cP_k(n,m)}\brk{Z_{k,\beta}}\leq\psi(\beta)+o(1).$
\end{lemma}
\begin{proof}
Let $\tau\subset V$ be such that $|\sigma\cap\tau|=(1-\beta)k$. 
The total number of graphs with $m$ in which both $\sigma,\tau$
are independent sets equals
	$$\bink{\bink n2-2\bink k2+\bink{(1-\beta)k}2}{m}.$$
For we can choose any $m$ edges out of those potential edges that
do not join two vertices of either $\sigma$ or $\tau$.
Since both $\sigma,\tau$ have size $k$ and $|\sigma\cap\tau|=(1-\beta)k$,
the number of such `bad' potential edges is $2\bink k2-\bink{(1-\beta)k}2$
by inclusion/exclusion.
Since $G$ is chosen uniformly among all $\bink{\bink n2-\bink k2}{m}$
graphs in which $\sigma$ is independent, we thus get
	\begin{eqnarray}
	\pr\brk{\tau\mbox{ is independent}}&=&\bink{\bink n2-2\bink k2+\bink{(1-\beta)k}2}{m}/\bink{\bink n2-\bink k2}{m}\nonumber\\
		&=&\prod_{j=0}^{m-1}\frac{\bink n2-2\bink k2+\bink{(1-\beta)k}2-j}{\bink n2-\bink k2-j}
			\leq\bcfr{\bink n2-2\bink k2+\bink{(1-\beta)k}2}{\bink n2-\bink k2}^m\nonumber\\
		&=&\bc{1-\frac{k^2-((1-\beta)k)^2}{n^2-k^2}+O(1/n)}^m\nonumber\\
		&\leq&O\bc{1}\cdot\bc{1-\frac{s^2(1-(1-\beta)^2)}{1-s^2}}^m\qquad\mbox{[as $k=sn$]}.
			\label{eqLemmacondexp1}
	\end{eqnarray}
Furthermore, the total number of ways to choose a set $\tau$ with
$|\sigma\cap\tau|=(1-\beta)k$ equals 	$\bink{k}{(1-\beta)k}\cdot\bink{n-k}{\beta k}$
(choose the $(1-\beta)k$ vertices in the intersection $\sigma\cap\tau$
and then choose the remaining $\beta k$ vertices).
By the linearity of the expectation, we get from~(\ref{eqLemmacondexp1})
	\begin{eqnarray*}
	\Erw\brk{Z_{k,\beta}}&=&O(1)\cdot\bink{k}{(1-\beta)k}\cdot\bink{n-k}{\beta k}\cdot\bc{1-\frac{s^2(1-(1-\beta)^2)}{1-s^2}}^m\\
		&=&O(1)\cdot\bink{k}{\beta k}\cdot\bink{n-k}{\beta k}\cdot\bc{1-\frac{s^2(1-(1-\beta)^2)}{1-s^2}}^m\\
		&\leq&O(1)\cdot\bcfr{\eul}{\beta}^{\beta k}\bcfr{\eul(n-k)}{\beta k}^{\beta k}\cdot\bc{1-\frac{s^2(1-(1-\beta)^2)}{1-s^2}}^m\\
		&=&O(1)\cdot\bcfr{\eul^2(1-s)}{s\beta^2}^{\beta sn}\cdot\bc{1-\frac{s^2(1-(1-\beta)^2)}{1-s^2}}^{dn/2}
			\qquad\mbox{[as $k=sn$ and $m=dn/2$]}.
	\end{eqnarray*}
Taking logarithms and dividing by $n$ completes the proof.
\end{proof}

Let us call an independent set $\sigma$ of size $k$ of a graph $G$
\emph{$(b_1,b_2,\gamma)$-good} if $G$ has no independent set $\tau$
such that $(1-b_1)k\leq|\sigma\cap \tau|\leq(1-b_2)k$ and if
		$\abs{\cbc{\tau\in\cS_k(G):|\sigma\cap\tau|>(1-b_2)k}}\leq\exp(-\gamma n)|\cS_k(G)|$.
Moreover, let
	\begin{equation}\label{eqeventZdk}
	\cZ_{d,k}=\cbc{(G,\sigma)\in\Lambda_k(n,m):|\cS_k(G)|\geq\Erw|\cS_k(G(n,m))|\cdot\exp\bc{-14n\sqrt{\ln^5d/d^3}}}.
	\end{equation}

\begin{corollary}\label{Cor_condexp}
Suppose that $b>0$ is such that~(\ref{eqcalcShattering1}) and~(\ref{eqcalcShattering2}) hold.
Then there exist $b_1,b_2,\gamma>0$ such that
	$$\pr_{\cU_k(n,m)}\brk{(G,\sigma)\mbox{ is $(b_1,b_2,\gamma)$-good}|\cZ_{d,k}}\geq1-\exp(-\gamma n).$$
\end{corollary}
\begin{proof}
The function $\psi$ is continuous.
Therefore, if (\ref{eqcalcShattering1}) and~(\ref{eqcalcShattering2}) 
are satisfied for some $b<0$ then there exist $b_1>b_2$ and $\zeta>0$ such that
	\begin{eqnarray}\label{eqcalcShattering1a}
	\sup_{b_2\leq\beta\leq b_1}\psi(\beta)&<&-18qs
	-\zeta\qquad\mbox{and}\\
	\sup_{x<b_2}\psi(x)&<&-s\ln(s)-(1-s)\ln(1-s)+\frac{d}2\ln(1-s^2)-d^{-1.49}-\zeta.
		\label{eqcalcShattering2a}
	\end{eqnarray}
Let $Z_{k,b_1,b_2}(G,\sigma)$ be the number of $\tau\in\cS_k(G)$ such that $(1-b_1)k\leq|\sigma\cap\tau|\leq(1-b_2)k$.
Then \Lem~\ref{Lemma_condexp}, (\ref{eqcalcShattering1a}), and Markov's
inequality yield
	\begin{eqnarray}\nonumber
	\pr_{\cP_k(n,m)}\brk{Z_{k,b_1,b_2}>0}&\leq&
		\Erw_{\cP_k(n,m)}\brk{Z_{k,b_1,b_2}}
			\leq\sum_{b_2k\leq j\leq b_1k}\Erw_{\cP_k(n,m)}\brk{Z_{k,j/k}}\\
		&\leq&\exp\brk{n\bc{\sup_{b_2\leq\beta\leq b_1}\psi(\beta)+o(1)}}\leq
			\exp\brk{-n\ln\ln d/d}.
				\label{eqCorcondexp1}
	\end{eqnarray}
The last inequality follows by taking $q>100\ln\ln d/\ln d$
and then $18qs\geq \ln\ln d/d$
Similarly, let $Z_{k,<b_2}(G,\sigma)$ be the number of $\tau\in|\cS_k(G)|$
such that $|\sigma\cap\tau|>(1-b_2)k$. Moreover, let $s=k/n$ and let
	\begin{eqnarray*}
	\mu&=&\Erw|\cS_k(G(n,m))|\cdot\exp\bc{-14n\sqrt{\ln^5d/d^3}}\\
		&=&O(1)\bink nk(1-(k/n)^2)^m\cdot\exp\bc{-14n\sqrt{\ln^5d/d^3}+o(n)} 
		\qquad\mbox{[by Corollary \ref{corollary:expectation-equivalence}]}\\
		&=&\exp\brk{n\bc{-s\ln(s)-(1-s)\ln(1-s)-\frac d2\ln(1-s^2)-14\sqrt{\ln^5d/d^3}+o(1)}},
	\end{eqnarray*}
where in the last step we used Stirling's formula.
Using~(\ref{eqcalcShattering2a}) and Markov's inequality, we find that
	\begin{eqnarray}\nonumber
	\pr_{\cP_k(n,m)}\brk{Z_{k,<b_2}>\mu}&\leq&
		\frac{\Erw_{\cP_k(n,m)}\brk{Z_{k,<b_2}}}\mu
			\leq\sum_{ j<b_2k}\frac{\Erw_{\cP_k(n,m)}\brk{Z_{k,j/k}}}\mu\\
		&\leq&\frac1\mu\exp\brk{n\bc{\sup_{\beta< b_2}\psi(\beta)+o(1)}}
				\leq\exp\brk{-n\ln d/d}.
					\label{eqCorcondexp2}
	\end{eqnarray}
Combining~(\ref{eqCorcondexp1}) and~(\ref{eqCorcondexp2}) with 
\Cor~\ref{Cor_quantexchange}, and letting, say, $\gamma=d^{-2}$, we see that
	\begin{eqnarray*}
	\pr_{\cU_k(n,m)}\brk{(G,\sigma)\mbox{ is not $(b_1,b_2,\gamma)$-good}|\cZ_{d,k}}
		&\leq&\pr_{\cU_k(n,m)}\brk{Z_{k,<b_2}>\mu\mbox{ or }Z_{k,b_1,b_2}>0}\\
		&\hspace{-8cm}\leq&\hspace{-4cm}
			(1+o(1))\pr_{\cP_k(n,m)}\brk{Z_{k,>b_2}>\mu\mbox{ or }Z_{k,b_1,b_2}>0}\cdot\exp\brk{14n\sqrt{\ln^5d/d^3}}\\
		&\hspace{-8cm}\leq&\hspace{-4cm}\exp(-\gamma n),
	\end{eqnarray*}
as claimed.
\end{proof}

\noindent{\bf Proof of \Prop~\ref{Prop_calcShattering}:}
Let $\cZ$ be the event that
	$$|\cS_k(G(n,m))|\geq\Erw|\cS_k(G(n,m))|\cdot\exp\bc{-14n\sqrt{\ln^5d/d^3}}.$$
\Cor~\ref{Cor_condexp} implies that there exists $b_1,b_2,\gamma$
such that given $\cZ$, 	\whp\ $G=G(n,m)$ has the property that all
but $\exp(-\gamma n)|\cS_k(G(n,m))|$ independent sets $\sigma\in\cS_k(G)$
are $(b_1,b_2,\gamma)$-good. Let $\cG$ denote this event.
As \Lem~\ref{Lemma_gap} ensures that $G(n,m)\in\cZ$ \whp, 
we have
		$$\pr\brk{\cG}\geq\pr\brk{\cG\cap\cZ}=\pr\brk{\cG|\cZ}\cdot\pr\brk{\cZ}=1-o(1).$$
As a consequence, we just need to show that
	the two conditions in Definition~\ref{Def_shattering} are satisfied if $\cG$ occurs.

Thus, let $G\in\cG$. We construct a decomposition of $\cS_k(G)$ 
into pairwise disjoint subsets $S_1,\ldots,S_{N}$ inductively
as follows. Suppose $i\geq1$.
If the set $\cS_k(G)\setminus\bigcup_{j=1}^{i-1}S_j$ does not contain 
a $(b_1,b_2,\gamma)$-good anymore, let $N=i$, set
	$$S_{N}=\cS_k(G)\setminus\bigcup_{j=1}^{N-1}S_j$$
and stop.
Otherwise, choose some $\sigma_i\in\cS_k(G)\setminus\bigcup_{j=1}^{i-1}S_j$ that is $(b_1,b_2,\gamma)$-good, let
	$$S_i=\cbc{\tau\in\cS_k(G):|\sigma\cap\tau|>b_2k}\setminus\bigcup_{j=1}^{i-1}S_j$$
and proceed to $i+1$.

Let $\zeta=k(b_1-b_2)/n$. We claim that this construction satisfies the two
conditions in Definition~\ref{Def_shattering}. Indeed, each $\sigma_i$ is
$(b_1,b_2,\gamma)$-good for all, we have 	$|S_i|\leq\exp(-\gamma n)\abs{\cS_k(G)}$ 
for all $i<N$. Furthermore, as $G\in\cG$ we have 
$\abs{S_N}\leq \exp(-\gamma n)\abs{\cS_k(G)}$.
Thus, the partition $S_1,\ldots,S_N$ satisfies the first condition
in Definition~\ref{Def_shattering}.

With respect to the second condition, let $\tau\in S_i$ and $\tau'\in S_j$
with $1\leq i<j\leq N$. Assume for contradiction that $\dist(\tau,\tau')<\zeta n$.
Then, for some $\sigma_i\in S_i$ we have that
	$$\dist(\sigma_i,\tau')\leq\dist(\sigma_i,\tau)+\dist(\tau,\tau')
		=2(k-|\sigma_i\cap\tau|)+\zeta n\leq 2b_2k+\zeta n,$$
and thus $|\sigma_i\cap\tau'|=k-\dist(\sigma_i,\tau')/2\leq (1-b_2)k-\zeta n/2\in\brk{(1-b_1)k,(1-b_2)k}$.
This contradicts the fact that $\sigma_i$ is good (which implies that there is no
	independent set $\sigma'$ such that $|\sigma_i\cap\sigma'|\in\brk{(1-b_1)k,(1-b_2)k}$.
Thus, we have established the second condition in Definition~\ref{Def_shattering}.
\qed

\section{Proof of Theorem \ref{theorem:non-maximal}}\label{section:theorem:non-maximal}

In this section we assume that $d\geq d_0$ for some large enough constant $d_0>0$.
Moreover, let $\eps_d\ra0$ be a function of $d$ that tends to $0$ sufficiently slowly,
and assume that $k=(1-\eps)n\ln d/d$ for some $\eps\in\brk{\eps_d,1-\eps_d}$.

Our goal is to show that for a random pair $(G,\sigma)$ chosen from $\cU_k(n,m)$
\whp\ there is a larger independent set $\tau$ in $G$ that contains $\sigma$ as a subset.
More precisely, $\tau$ is supposed to have size $k(1+\frac{2\eps}{1-\eps})$.
In order to construct such a set $\tau$ we need the following concept.

\begin{definition}
A vertex $v\in V\backslash \sigma$ is called {\bf $\sigma$-pure} in $G$ if it is not adjacent to any vertex in $\sigma$.
\end{definition}
Basically, in order to expand $\sigma$ we are going to show that $G$ has an independent set $I\subset V\setminus\sigma$
of size $|I|=2\eps k/(1-\eps)$  consisting of $\sigma$-pure vertices.
Then $\tau=\sigma\cup I$ is the desired larger independent set.
We begin by estimating the number of $\sigma$-pure vertices and the density of the graph
that they span.


\begin{lemma}\label{lemma:pure-graph}
Let $(G, \sigma)$ be chosen from ${\cal P}_k(n,m)$,
	where
		 $k=(1-\eps)\frac{\ln d}{d}n$ with $\eps \in [10\ln\ln d/\ln d, 1]$. 
Let $Q$ be the set of $\sigma$-pure vertices.
Then with probability $\geq1-\exp \left(-\frac{n}{d}\right)$ the following two statements hold.
\begin{enumerate}
\item Let $N=|Q|$. Then $N\geq(1-o_d(1))d^{\eps-1}n$.
\item Let $M$ be the number of edges in the induced subgraph $G\brk Q$.
	Then $M\leq (\frac12+\delta)d^{2\eps-1}n$, with $0<\delta<2d^{-\epsilon/3}$.
\end{enumerate}
\end{lemma}

\begin{proof}
Instead of working directly with the distribution $\cP_k(n,m)$,
let us consider the following variant $\cP_k'(n,m)$. First,
choose a set $\sigma'\subset V$ of size $k$ uniformly at random.
Then, constrict a graph $G'$ by inserting each of the $\bink n2-\bink k2$ possible
edges that do not join two vertices in $\sigma'$ with probability
$p=m/(\bink n2-\bink k2)$ independently.

Thus, the number of edges in $G'$ is binomially distribution with
mean $m$. Furthermore, given that $G'$ has precisely $m$ edges, it
is a uniformly random graph with this property in which $\sigma'$
is an independent set. Therefore, for any event $\cA$ we have
	\begin{eqnarray}\nonumber
	\pr_{\cP_k(n,m)}\brk{\cA}&=&
		\pr_{\cP_k'(n,m)}\brk{\cA\,|\,|E(G')|=m}\\
		&\leq&\frac{\pr_{\cP_k'(n,m)}\brk{\cA}}{\pr\brk{\Bin\bc{\bink n2-\bink k2,p}=m}}=\Theta(\sqrt m)\cdot\pr_{\cP_k'(n,m)}\brk{\cA},
			\label{eqP'knm}
	\end{eqnarray}
where the last step follows from Stirling's formula.

Now, let $N'$ be the number of $\sigma'$-pure vertices in $G'$.
For each vertex $v\not\in\sigma$ the number of neighbours in $\sigma$
is binomially distributed with mean $kp$. In effect, $v$ is pure with
probability $(1-p)^k$. Since these events are mutually independent for
all $v\not\in\sigma$, $N'$ has a binomial distribution $\Bin(n-k,(1-p)^k)$.
Hence, letting $s=k/n=(1-\eps)\ln d/d$, we have
	\begin{eqnarray*}
	\Erw\brk{N'}&=&(n-k)(1-p)^k\sim(1-s)n\exp(-kp)\sim(1-s)n\exp\brk{-\frac{ds}{1-s^2}}\\
		&\geq&(1-s)n\exp\brk{-ds\bc{1+2s^2}}\geq0.99n d^{\eps-1},
	\end{eqnarray*}
provided that $d$ is sufficiently big.
Letting $\gamma=d^{-\eps/3}=o_d(1)$, we obtain from \Thm~\ref{chernoffbounds} (the Chernoff bound)
	\begin{eqnarray*}
	\pr\brk{N'<(1-\gamma) nd^{\eps-1}}&\leq&\exp\brk{-nd^{\eps/3-1}/4}\leq\exp\brk{-2n/d}
	\end{eqnarray*}
for $d$ large enough.
Together with~(\ref{eqP'knm}) this implies the first assertion.

To prove the second assertion, we need an upper bound on $N'$.
Once more by the Chernoff bound,
	\begin{eqnarray}\label{eqlemma:pure-graph1}
	\pr\brk{N'>(1+\gamma)nd^{\eps-1}}&\leq&\exp\brk{-nd^{\eps/3-1}/8}\leq\exp\brk{-2n/d}
	\end{eqnarray}
for $d$ large enough.
Let $Q$ be the set of $\sigma'$-pure vertices in $G'$.
Since each potential edge that does not link two vertices in $\sigma'$
is present in $G'$ with probability $p$ independently, given the value
of $N'$ the number $M'$ of edges spanned by $Q$ is binomially distributed
with mean $\bink{N'}2p$.
Therefore,
	$$\Erw\brk{M'|N'\leq(1+\gamma)nd^{\eps-1}}\leq
		\frac{(1+\gamma)^2n^2d^{2\eps-2}}{2}\cdot\frac{dn/2}{\bink{n}2-\bink{k}2}\leq\frac{1+3\gamma}2 nd^{2\eps-1},$$
provided that $d$ is large.
Hence, by the Chernoff bound and~(\ref{eqlemma:pure-graph1}),
	\begin{eqnarray}\nonumber
	\pr\brk{M'>\bc{\frac12+2\gamma} nd^{2\eps-1}}&\leq&
		\pr\brk{M'>\bc{\frac12+2\gamma} nd^{2\eps-1}|N'\leq(1+\gamma) nd^{\eps-1}}\\
			&&\qquad+\pr\brk{N'>(1+\gamma) nd^{\eps-1}}\nonumber\\
		&\leq&\exp\brk{-nd^{2\eps-1}/8}+\exp\brk{-2n/d}\leq2\exp\brk{-2n/d}
			\label{eqlemma:pure-graph2}
	\end{eqnarray}
for $d$ big.
Finally, the second assertion follows from~(\ref{eqP'knm}) and~(\ref{eqlemma:pure-graph2}).
\end{proof}

\noindent{\bf Proof of \Thm~\ref{theorem:non-maximal}.}
Suppose that $k=(1-\eps)n\ln d/d$. Let $(G,\sigma)$ be a pair
chosen from the distribution $\cP_k(n,m)$. Let $Q$ be the set
of $\sigma$-pure vertices and let $N,M$ be as in \Lem~\ref{lemma:pure-graph}.
Crucially, given $Q$, $N$, $M$, the induced subgraph $G\brk Q$
is just a uniformly random graph on $N$ vertices with $M$ edges,
because the conditioning only imposes the absence of $Q$-$\sigma$-edges.
In other words, $G\brk Q$ is nothing but a random graph $G(N,M)$.
We are going to use this observation to show that $G\brk Q$
contains a large independent set \whp\

Let $\cA$ be the event that $N\geq(1-o_d(1))d^{\eps-1}n$ and
$M\leq (\frac12+o_d(1))d^{2\eps-1}n$. Then by \Lem~\ref{lemma:pure-graph}
\begin{equation}\label{eqtheorem:non-maximal1}
	\pr_{\cP_k(n,m)}\brk{\cA}\geq1-\exp(-n/d).
\end{equation}
Given $\cA$, the average degree of $G\brk Q$ is
$$D=\frac{2M}N\leq(1+o_d(1))\frac{d^{2\eps-1}}{d^{\eps-1}}=(1+o_d(1))d^{\eps}.$$
Let $\cB$ be the event that $\alpha(G\brk Q)\geq(2-o_d(1))\frac{N\ln D}{D}$.
Since $G\brk Q$ is distributed as $G(N,M)$, \Cor~\ref{cor:ExistenceGnm} 
implies that
\begin{equation}\label{eqtheorem:non-maximal2}
	\pr_{\cP_k(n,m)}\brk{\cB|\cA}\geq1-\exp\left(-\frac{8n}{\epsilon^3 d \ln^3 d}\right).
\end{equation}
Combining~(\ref{eqtheorem:non-maximal1}) and~(\ref{eqtheorem:non-maximal2}) with \Thm~\ref{thrm:transfer-theorem}, we thus get
	\begin{equation}\label{eqtheorem:non-maximal3}
	\pr_{\cU_k(n,m)}\brk{\cA\cap \cB}=1-o(1).
	\end{equation}

\noindent
Now assume that $(G,\sigma)\in\cA\cap\cB$.
Let $I$ be the largest independent set of $G\brk Q$.
Then
	\begin{eqnarray}\label{eqtheorem:non-maximal4}
	\abs I&=&(1-o_d(1))\frac{2d^{\eps-1}n\cdot\ln(d^{\eps})}{d^{\eps}}=(1-o_d(1))\frac{2\eps\ln d}d=(1-o_d(1))\frac{2\eps k}{1-\eps}.
	\end{eqnarray}
Since $\sigma\cup I$ is independent, (\ref{eqtheorem:non-maximal4}) shows that $\sigma$ is
	$((2-o_d(1))\eps/(1-\eps),0)$-expandable.
Thus, the assertion follows from~(\ref{eqtheorem:non-maximal3}).
\qed

\section{Proof of Theorem \ref{theorem:maximal}}\label{section:theorem:maximal}

Let $\eps_d=3\ln\ln d/\ln d\ra0$.
In this section we assume that $k=(1+\eps)n\ln d/d$ with $\eps_d\leq \eps\leq1-\eps_d$,
	and that $d\geq d_0$ for some large enough constant $d_0>0$.
Assuming that $\gamma,\delta>0$ are reals such that
	\begin{equation}\label{eqsection:theorem:maximal}
	\gamma>\eps_d\qquad\mbox{and}\qquad\delta<\gamma+\frac{2(\eps-\eps_d)}{1+\eps},
	\end{equation}
we are going to show that in a pair $(G,\sigma)$ chosen from the distribution $\cU_k(n,m)$,
$\sigma$ is not $(\gamma,\delta)$-expandable.

To see why this is plausible, consider a pair $(G,\sigma)$ chosen from the distribution $\cP_k(n,m)$.
(The following argument is not actually needed for our proof of \Thm~\ref{theorem:maximal}; it is only included
	to facilitate understanding.)
Then for each vertex $v\not\in\sigma$ the \emph{expected} number of neighbours of $v$
inside of $\sigma$ is greater than $kd/n=(1+\eps)\ln d$.
Indeed, one could easily show that for each vertex $v$ the number of neighbours in $\sigma$
dominates a Poisson variable $\Po((1+\eps)\ln d)$.
Hence, the probability that $v$ is $\sigma$-pure is bounded by
	$\exp(-(1+\eps)\ln d)=d^{-\eps-1}$,
and thus the expected number of $\sigma$-pure vertices is $\leq n d^{-\eps-1}= o_d(1)\cdot k$.
In effect, in order to expand $\sigma$ significantly we would have to include some
vertices that are \emph{not} $\sigma$-pure.
But each such vertex would `displace' some other vertex from $\sigma$
	(by the very definition of $\sigma$-pure).
In fact, most vertices that are not $\sigma$-pure have several neighbours in $\sigma$,
and thus it seems impossible to expand $\sigma$ substantially without first removing
a significant share of its vertices.

To actually prove \Thm~\ref{theorem:maximal} we use a first moment argument.
We begin by analysing the planted model.

\begin{lemma}\label{lemma:maximality}
With $d\geq d_0$ sufficiently large and $k,\gamma,\delta$ as above, we have
\begin{displaymath}
P_{{\cal P}_k(n,m)}[\textrm{$\sigma$ is not $(\gamma,\delta)$-expandable }]\geq 1-\exp \left( -\frac{n}{d}\right).
\end{displaymath}
\end{lemma}
\begin{proof}
Let $s=k/n$.
For $(G,\sigma)$ chosen from the distribution $\cP_k(n,m)$,
let $X$ be the number of independent sets $\tau$ such that
	\begin{equation}\label{eqlemma:maximality1}
	|\tau|=(1+\gamma)k\mbox{ and }
	|\tau\cap\sigma|\geq(1-\delta)k.
	\end{equation}
The total number of ways to choose a set $\tau\subset V$ satisfying~(\ref{eqlemma:maximality1}) is
	\begin{eqnarray}\label{eqlemma:maximality2}
	\cH&=&\bink k{(1-\delta)k}\bink{n-k}{(\gamma-\delta)k}
	\end{eqnarray}
(first choose $(1-\delta)k$ vertices from $\sigma$, then choose the remaining $(1+\gamma)k-(1-\delta)k=(\gamma-\delta)k$
vertices from $V\setminus\sigma$).
Furthermore, for any $\tau\subset V$ satisfying~(\ref{eqlemma:maximality1}) the probability of being independent is
	\begin{eqnarray}\label{eqlemma:maximality3}
	\cP&=&\bink{\bink n2-\bink k2-\bink{(1+\gamma)k}2+\bink{(1-\delta)k}2}{m}/\bink{\bink{n}2-\bink k2}m
	\end{eqnarray}
Indeed, in order for both $\sigma$ and $\tau$ to be independent we have to forbid all edges
that connects two vertices in either set, and the number of potential such edges is
	$\bink{|\sigma|}2+\bink{|\tau|}2-\bink{|\sigma\cap\tau|}2$ by inclusion/exclusion.
This explains the numerator in~(\ref{eqlemma:maximality3}), and the denominator simply reflects that $G$ is chosen randomly
from all graphs in which $\sigma$ is independent.

Combining~(\ref{eqlemma:maximality2}) and~(\ref{eqlemma:maximality3}) and using the linearity of the expectation, we see that
	\begin{eqnarray}\label{eqlemma:maximality4}
	\Erw\brk X&=&\cH\cdot\cP.
	\end{eqnarray}
We are going to show that $\Erw\brk X$ and then apply Markov's inequality to obtain the lemma.

We begin by estimating $\cH$ and $\cP$ separately.
For $\cH$ we get
	\begin{eqnarray*}
	\cH&=&\bink k{\delta k}\bink{(1-s)n}{(\gamma+\delta)sn}\leq
			\bcfr{\eul}{\delta}^{\delta k}\bcfr{\eul(1-s)}{(\gamma+\delta)s}^{(\gamma+\delta)k}\\
		&=&\exp\brk{s\brk{\delta(1-\ln\delta)+(\gamma+\delta)\bc{1+\ln\bcfr{1-s}{(\gamma+\delta)s}}}n}\\
		&\leq&\exp\brk{s\brk{\delta(1-\ln\delta)+(\gamma+\delta)\bc{1-\ln(\gamma+\delta)-\ln s}}n}.
	\end{eqnarray*}
As we assume that $s\geq\ln d/d$
 and $\gamma\geq\eps_d\geq1/\ln d$ and $\delta\geq0$, we have
	$-\ln s\leq\ln d$ and $-\ln(\gamma+\delta)\leq\ln\ln d$.
Furthermore, the function $x\mapsto x(1-\ln x)$ is monotonically increasing for $x\leq1$.
Hence, if $\gamma+\delta\leq1$, then $\delta(1-\ln\delta)\leq(\gamma+\delta)\bc{1-\ln(\gamma+\delta)}$.
If, on the other hand, $\gamma+\delta>1$, then
	$\delta(1-\ln\delta)\leq1<\gamma+\delta$.
In either case we obtain
	\begin{eqnarray}\label{eqlemma:maximality5}
	\frac1n\ln\cH&\leq&	s(\gamma+\delta)(1+\ln\ln d-\ln d).
	\end{eqnarray}
With respect to $\cP$, we have 
	\begin{eqnarray*}
	\cE&=&\bink{\bink n2-\bink k2-\bink{(1+\gamma)k}2+\bink{(1-\delta)k}2}{m}/\bink{\bink{n}2-\bink k2}m\\
		&=&\prod_{j=0}^{m-1}\frac{\bink n2-\bink k2-\bink{(1+\gamma)k}2+\bink{(1-\delta)k}2-j}{\bink{n}2-\bink k2-j}
			\leq\bcfr{\bink n2-\bink k2-\bink{(1+\gamma)k}2+\bink{(1-\delta)k}2}{\bink{n}2-\bink k2}^m\\
		&=&O(1)\cdot\bcfr{1-s^2-(1+\gamma)^2s^2+(1-\delta)^2s^2}{1-s^2}^m
		=O(1)\cdot\bc{1-\frac{s^2(\gamma+\delta)(2+\gamma-\delta)}{1-s^2}}^m.
	\end{eqnarray*}
Since $m=dn/2$ and $d=(1+\eps)\ln d/d$, the elementary inequality $\ln(1-x)\leq-x$ yields
	\begin{eqnarray}\label{eqlemma:maximality6}
	\frac1n\ln\cE\leq\frac d2\ln\bc{1-s^2(\gamma+\delta)(2+\gamma-\delta)}
		\leq-s(\gamma+\delta)\bc{1+\frac{\gamma-\delta}{2}}(1+\eps)\ln d.
	\end{eqnarray}
Finally, plugging~(\ref{eqlemma:maximality5}) and~(\ref{eqlemma:maximality6}) into~(\ref{eqlemma:maximality4}), we get
for $d\geq d_0$ large enough
	\begin{eqnarray*}
	\frac1n\ln\Erw\brk X&=&\frac1n\ln\cH+\frac1n\ln\cE\leq
		s(\gamma+\delta)\brk{1+\ln\ln d-\ln d-\bc{1+\frac{\gamma-\delta}{2}}(1+\eps)\ln d}\\
		&\leq&s(\gamma+\delta)\brk{1+\ln\ln d-\bc{\eps+\frac{\gamma-\delta}{2}}\ln d}\\
		&\leq&s(\gamma+\delta)\brk{1+\ln\ln d-\frac{\eps_d}2\ln d}\qquad\qquad\quad\mbox{[by our assumption~(\ref{eqsection:theorem:maximal}) and $\gamma,\delta$]}\\
		&\leq&-s(\gamma+\delta)\qquad\qquad\qquad\qquad\qquad\qquad\qquad\quad\mbox{[as $\eps_d=3\ln\ln d/\ln d$]}\\
		&\leq&-s\eps_d\leq-1/d\qquad\qquad\qquad\qquad\qquad\qquad\quad\mbox{[as $\gamma\geq\eps$ and $s\geq \ln d/d$]}.
	\end{eqnarray*}
Thus, the assertion follows from Markov's inequality.
\end{proof}

\noindent
Theorem \ref{theorem:maximal} follows directly from \Lem~\ref{lemma:maximality} and \Thm~\ref{thrm:transfer-theorem}.

\section{Proof of Corollary \ref{cor:MixingTimeBound}}\label{sec:cor:MixingTimeBound}

Let $\eps_d\ra0$ slowly. Throughout this section we assume that
\begin{equation}\label{eq:themuglambda'}
	(1+\eps_d)\frac{\ln d}d\cdot n\leq\Erw\brk{\mu(G(n,m),\lambda)}\leq(2-\eps_d)\frac{\ln d}d\cdot n.
\end{equation}

\noindent
The proof of Corollary \ref{cor:MixingTimeBound} amounts to showing
that the Metropolis process can be ``trapped'' in a relatively
small group of independent sets and it escapes only after an
exponentially large number of steps. To be more specific, let
\begin{equation}\label{eq:TheK}
K=\left\{k:|\Erw\brk{\mu(G_{n,m},\lambda)}-k|\leq {4n}/{d} \right\}.
	\end{equation}
We show that $\bigcup_{k\in K}{\cal S}_k$ can be partitioned into
disconnected parts, i.e.  it is not possible for the process to
move from one part to another without using independent sets of
size much smaller than the minimum $k\in K$. However, we show that
once the process gets to a ``typical'' independent set in $\bigcup_{k\in K}{\cal S}_k$
it will need to wait for exponential time so as to escape by
visiting a small independent set.

Before showing Corollary \ref{cor:MixingTimeBound} we provide some
auxiliary results. The following proposition shows that for a given
parameter $\lambda$ the stationary distribution of the Metropolis
process concentrates on a small range of sizes of independent sets.

\begin{proposition}\label{prop:MostLikelyIS}
With probability at least $1-2\exp\left[-n/(2d^2\ln^4d)\right]$ the
random graph $G=G(n,m)$ has the following property.
\begin{quote}
For an independent set ${\cal I}$ chosen from the stationary distribution
of the Metropolis process on $G$ we have
	\begin{equation}\label{eq:TheK2}
	Pr[|{\cal I}|\notin K]\leq \exp\left( -n/d\right)
	\end{equation}
(where in~(\ref{eq:TheK2}) probability is taken over the choice of $\cI$ only).
\end{quote}
\end{proposition}
The proof of Proposition \ref{prop:MostLikelyIS} appears in Section 
\ref{sec:prop:MostLikelyIS}.

\begin{lemma}\label{lemma:ShateringTube}
\Whp\ the random graph $G=G(n,m)$ has the following property.
The set $\bigcup_{k\in K}{\cal S}_k(G)$
admits a partition into classes ${\cal C}_1, \ldots, {\cal C}_N$ such that the following three statements hold.
\begin{description}
\item[C1.] The distance between any two independent sets in different classes is at least $2$.
\item[C2.] For a random set $\cI$ chosen from the stationary distribution of the Metropolis process we have
	$$Pr[{\cal I}\in {\cal C}_i]\leq 5\exp\left(-n/(2d^2\ln^4 d)\right)
	\qquad \mbox{ for each }i\leq i\leq N.$$
\item[C3.] Furthermore, $Pr[{\cal I} \in \bigcup_{1\leq i\leq N}{\cal C}_i]\geq 
	1-5\exp\left(-n/(2d^2\ln^4 d)\right)$. 
\end{description}
\end{lemma}
The proof of Lemma \ref{lemma:ShateringTube} appears in
Section \ref{sec:lemma:ShateringTube}.

\medskip\noindent{\bf Proof of \Cor~\ref{cor:MixingTimeBound}:}
Let $K$ be as in (\ref{eq:TheK}) and assume that $G=G_{n,m}$ is such
that $\bigcup_{k\in K}{\cal S}_k(G)$ has a partition ${\cal C}_1, \ldots, 
{\cal C}_N$ satisfying {\bf C1--C3} in
Lemma~\ref{lemma:ShateringTube}.
We are going to show that the mixing time of the Metropolis process exceeds $\exp\bc{n/d^{3}}$.
The proof is by contradiction. Thus, assume that the mixing time
of the Metropolis process is $T\leq \exp\bc{n/d^3}$.
Let ${\cal I}_t$ be the state of the Metropolis process at time ($t\geq0$).

Let $t_1=n^{2}T$ and  $t_2=2n^{2}T$.
Since $T$ is the mixing time, for any $t_1\leq t\leq t_2$ the
 distribution of ${\cal I}_t$ is extremely close to the stationary distribution.
More precisely, if $\cI_\infty$ chosen from the stationary distribution, then for any $t\in [t_1, t_2]$ we have
	$$\norm{\cI_t-\cI_\infty}_{tv}\leq\exp\bc{-n^{2}}.$$
Therefore, {\bf C3} implies that for any $t\in [t_1, t_2]$,
\begin{displaymath}
Pr\left[{\cal I}_t\notin \bigcup_{1\leq i\leq N}{\cal C}_i \right]
\leq Pr\left[{\cal I}_\infty\notin \bigcup_{1\leq i\leq N}{\cal C}_i \right]+\norm{\cI_t-\cI_\infty}_{tv}\leq 6\exp \left[-n/(2d^2\ln^4d) \right].
\end{displaymath}
Applying the union bound, we get for $d\geq d_0$ large enough
\begin{equation}\label{eq:NotInRare}
Pr\left[\exists t_1\leq t \leq t_2: {\cal I}_t \notin \bigcup_{1\leq i\leq N}{\cal C}_i 
 \right]\leq 6\exp\left( -\frac{n}{2d^2\ln^4d}+n/d^3\right)\leq 
 \exp\left( -\frac{n}{3d^2\ln^4d}\right).
\end{equation}
In other words, we have shown that to get from ${\cal I}_{t_1}$
to ${\cal I}_{t_2}$, the Metropolis process very likely only
passes through independent  sets from $\bigcup_{1\leq i\leq N}{\cal C}_i$.

Most likely, the two independent sets ${\cal I}_{t_1}$, ${\cal I}_{t_2}$
belong to different classes of the partition $\cC_1,\ldots,\cC_N$,
because the time difference $t_2-t_1=n^2T$ is much bigger than the
mixing time $T$. Formally, if $\cI_\infty$ is chosen from the
stationary distribution and $i_{1}$ such that $\cI_{t_1}\in\cC_{i_{1}}$,
then by {\bf C2}
\begin{eqnarray}\label{eq:InDifferentClasses}
\pr\brk{\cI_{t_2}\in\cC_{i_{1}}}
   &\leq& \pr\brk{\cI_{\infty}\in\cC_{i_{1}}}+\norm{\cI_{t_2-t_1}-\cI_\infty}_{tv}
	\leq 2\exp(-n/(3d^2\ln^4 d)).
\end{eqnarray}
Combining~(\ref{eq:NotInRare}) and~(\ref{eq:InDifferentClasses}),
we thus get
\begin{equation}\label{eq:RandomOverlap}
	Pr[\exists i, j \in [N], i\neq j:{\cal I}_{t_1}\in {\cal C}_i \wedge {\cal I}_{t_2}\in {\cal C}_j]\geq 
			1-\exp(-n/(3d^2\ln^4 d)).
\end{equation}

\noindent
Thus, assume that there are two distinct $i,j\in [N]$ such that
${\cal I}_{t_1}\in {\cal C}_i$ and ${\cal I}_{t_2}\in {\cal C}_j$.
Let $t>t_1$ be the first time that ${\cal I}_t\notin {\cal C}_i$.
Then by definition of the Metropolis process, $dist({\cal I}_{t},
{\cal I}_{t-1})\leq 1$.  Consequently, ${\cal I}_t \notin \bigcup_{l\in N}{\cal C}_l$
because otherwise there would be two independent sets in different
classes at distance one. Thus,
\begin{displaymath}
Pr[\exists i, j \in [N], i\neq j:{\cal I}_{t_1}\in {\cal C}_i \wedge {\cal I}_{t_2}\in {\cal C}_j]\leq
Pr\left[\exists t_1\leq t \leq t_2: {\cal I}_t \notin \bigcup_{1\leq i\leq N}{\cal C}_i  \right],
\end{displaymath}
in contradiction to (\ref{eq:NotInRare}) and (\ref{eq:RandomOverlap}).
\qed

\subsection{Proof of Proposition \ref{prop:MostLikelyIS}}
\label{sec:prop:MostLikelyIS}
For a graph $G$, let 
$$R_G(k,\lambda)=|{\cal S}_k(G)| \lambda^k.$$ 
It is easy to deduce from the definition of  Metropolis process
(see e.g. \cite{jerrum-planted}) that for any set of  integers 
${\cal A}$ it holds that
$$Pr[|{\cal I}|\in {\cal A}]\propto \sum_{k\in {\cal A}} R_G(k,\lambda).$$
Therefore, we have
\begin{eqnarray}
Pr[|{\cal I}|\notin {\cal A}]&=&\frac{\sum_{k\notin {\cal A}}R_G(k,\lambda)}{\sum_{k}R_G(k,\lambda)}
\leq \frac{\sum_{k\notin {\cal A}}R_G(k,\lambda)}{\sum_{k\in {\cal A}}R_G(k,\lambda)}. 
\label{eq:ratio2boundforMetr}
\end{eqnarray}
Consider some $\lambda$ that satisfies (\ref{eq:themuglambda'}).
Then, Proposition \ref{prop:MostLikelyIS} will follow by bounding
appropriately the rightmost ratio above, for ${\cal A}=K$ (as
defined in (\ref{eq:TheK})) and $G$ being a typical instance of
$G(n,m)$.
\\ \vspace{-.3cm}

\noindent
{\bf Remark.}
Observe that when the graph $G$ is distributed as in $G(n,m)$ the
quantity $R_G$ is a random variable which depends {\em only on the
underlying graph.} 
\\ \vspace{-.3cm}

\noindent
Before proving the proposition we need some preliminary results.
With the parameter $\lambda>0$ and the expected degree $d$ in mind, 
for any $x\in (0,1)$ we define the following function:
\begin{displaymath}
f_{\lambda}(x)=-(x\ln x+(1-x)\ln(1-x))+\frac{d}{2}\ln(1-x^2)+x\ln \lambda.
\end{displaymath}
It is straightforward to verify that $\frac{1}{n}\ln E[R_G(k,\lambda)]\sim f_{\lambda}(k/n)$.
$f_{\lambda}(x)$ is twice differentiable, as a matter of fact it holds that
\begin{eqnarray}
f'_{\lambda}(x)&=&\ln(1-x)-\ln x-d\frac{x}{1-x^2}+\ln \lambda \label{eq:f-lambdaPrime} \\
f''_{\lambda}(x)&=&-\frac{1}{x(1-x)}-d\frac{1+x^2}{(1-x^2)^2}. \label{eq:f-lambdaDPrime}
\end{eqnarray}
For any $\lambda$ and $x\in (0,1)$ it holds that $f''_{\lambda}(x)<0$. 
That is, $f'_{\lambda}(x)$ is strictly decreasing. Furthermore, if 
for given $\lambda, d$ there exists $x_0\in (0,1)$ such that 
\begin{eqnarray}\label{eq:LambdaMaxCond}
\lambda=\frac{x_0}{1-x_0}\exp\left(d\frac{x_0}{1-x^2_0}\right),
\end{eqnarray}
then $f_{\lambda}(x_0)$ is a global maximum for $f_{\lambda}$.
Since $f'_{\lambda}(x)$ is strictly decreasing, for any given
$x'\in (0,1)$ and $d$, we can find unique $\lambda_0>0$ such 
that $f_{\lambda_0}(x)$ is maximized when $x=x'$.

\begin{claim}\label{claim:UnlikeIS}
Take $x_0\in (0,1)$ and let $\lambda$ be such that $f_{\lambda}(x)$
is maximized for $x=x_0$. Then for any $x$ such that $|x-x_0|=t$
it holds that
\begin{displaymath}
f_{\lambda}(x)\leq f_{\lambda}(x_0)-\frac{1}{2}t^2d.
\end{displaymath}
\end{claim}
\begin{proof}
From (\ref{eq:f-lambdaDPrime}) it is easy to show that for any $x\in (0,1)$,
it holds that $f''_{\lambda}(x)<-d$. Also, for any $x\in (0,1)$ we can find 
appropriate $\xi\in [(0,1)$  such that 
\begin{eqnarray}
f_{\lambda}(x)&=&f_{\lambda}(x_0)+(x-x_0)f'_{\lambda}(x_0)+\frac{(x-x_0)^2}{2}f''_{\lambda}(\xi)
\nonumber \\
&\leq& f_{\lambda}(x_0)-\frac{(x-x_0)^2}{2}d, \hspace*{3cm}\mbox{[as $f'_{\lambda}(x_0)=0$ and $f''_{\lambda}(x)<-d$]}
\nonumber
\end{eqnarray}
as promised.
\end{proof}

\noindent
Let $\lambda_c$ be such that $f_{\lambda_c}(x)$ is maximized for 
$x=(1+c)\ln d/d$.

\begin{lemma}\label{lemma:MassOfMaximum}
For $c\in [\epsilon_d,1-\epsilon_d]$ and $k=(1+c)\frac{\ln d}{d}n$, it holds that
\begin{displaymath}
Pr\left[R_{G(n,m)}(k,\lambda_{c})\leq \exp\left(-14n\sqrt{{\ln^5d}/{d^{3}}}\right)
\cdot E[R_{G(n,m)}(k,\lambda_{c})]  \right] \leq 
\exp\left[-n/(2d^2\ln^4d) \right].
\end{displaymath}
\end{lemma}
\begin{proof}
The lemma follows directly from Proposition  \ref{Lemma_gap}.
\end{proof}

\begin{lemma}\label{lemma:AlmostMostLikelyIS}
For $c\in [\epsilon_d,1-\epsilon_d]$,  let $k=(1+c)\frac{\ln d}{d}n$
and
$${\cal R}_{c}=\exp\left(-14n\sqrt{{\ln^5 d}/{d^{3}}}\right)E[R_{G(n,m)}(k,\lambda_{c})].$$ 
It holds that
\begin{displaymath}
Pr\left[ \sum_{k': |k-k'|>\frac{1.9 n}{d}}R(k', \lambda_{c})\geq
\exp\left(-n/d \right) {\cal R}_{c} \right]\leq \exp\left(-n/(2d) \right).
\end{displaymath}
\end{lemma}
\begin{proof}
Observe that for any integer $0\leq k'\leq 2n\ln d/d$ it holds that
$E[R_{G(n,m)}(k',\lambda_c)]=\exp\left[f(k'/n)n+o(n)\right]$. 
Since the function $f_{\lambda_c}(x)$ is increasing for every 
$0\leq x < (1+c)\ln d/d $ and decreasing for $(1+c)\ln d/d<x<1$, for 
$k_0=k-1.9n/d$ and sufficiently large $n$ it holds  that
\begin{eqnarray}\label{eq:k_0Bound}
E[R_{G(n,m)}(k_0,\lambda_c)]\geq \max_{k':|k'- k|>1.9n/d}\left\{E[R_{G(n,m)}(k',\lambda_c)]\right\}.
\end{eqnarray}
Furthermore, using Claim \ref{claim:UnlikeIS} we get that
\begin{eqnarray}\label{eq:ERk_0Mass}
E[R_{G(n,m)}(k_0,\lambda_c)]\leq E[R_{G(n,m)}(k,\lambda_c)]\exp\left(-\frac{1.8n}{d}+o(n)\right).
\end{eqnarray}
Let $Q=\sum_{k': |k-k'|>\frac{1.9n}{d}}R(k', \lambda_{c})$. 
It holds that
\begin{eqnarray}
E[Q]&=&\sum_{k': |k-k'|>\frac{1.9n}{d}}E[R(k', \lambda_{c})] 
\nonumber \\
&\leq&  n E[R_{G(n,m)}(k_0,\lambda_c)] \hspace*{5.8cm} \mbox{[from (\ref{eq:k_0Bound})]} 
\nonumber \\
&\leq & E[R_{G(n,m)}(k,\lambda_c)]\exp\left(-\frac{1.8n}{d}+o(n)\right).
\hspace*{3cm} \mbox{[from (\ref{eq:ERk_0Mass})]} 
\label{eq:MassOfEQ}
\end{eqnarray}
The lemma follows by applying Markov's inequality. That is,
for sufficiently large $d$ it holds that
\begin{eqnarray}
Pr\left[Q\geq \exp\left(-n/d \right) {\cal R}_{c} \right]&\leq &
Pr\left[Q\geq E[Q]\exp\left(\frac{n}{2d}\right)\right] \hspace*{2.4cm} \mbox{[from (\ref{eq:MassOfEQ})]}
\nonumber\\
&\leq& \exp\left(-\frac{n}{2d}\right), \hspace*{2.4cm} \mbox{[from Markov's inequality]}\nonumber
\end{eqnarray}
as promised.
\end{proof}

\begin{propositionproof}{\ref{prop:MostLikelyIS}}
Let $c\in (\epsilon_d,1-\epsilon_d)$, for $\epsilon_d\to 0$.

Observe that quantity $\mu(G,\lambda)$ for fixed $\lambda$ and
$G$ distributed as in $G(n,m)$ is a random variable which depends
only on the graph $G$.
We are going to show that for $\lambda_c$ it holds that
\begin{equation}\label{eq:MuLambdaCBound}
Pr\left[\left|\mu(G(n,m),\lambda_c)-(1+c)\frac{\ln d}{d}n\right|
> \frac{1.95n}{d}
\right] \leq \exp\left[-n/(2d)\right].
\end{equation}
Observe that once we have the above tail bound,  the proposition
follows easily from Lemma \ref{lemma:AlmostMostLikelyIS}. In particular
(\ref{eq:MuLambdaCBound}) implies that 
\begin{equation}\label{eq:PropArg1}
\left|E[\mu(G(n,m),\lambda_c)]-(1+c)\frac{\ln d}{d}n\right|\leq \frac{1.95n}{d} +n\exp\left[-n/(2d)\right].
\end{equation}
Also, from Lemma \ref{lemma:AlmostMostLikelyIS} and (\ref{eq:ratio2boundforMetr})
we have the following: 
Consider the Metropolis process with underlying graph $G(n,m)$
and parameter $\lambda_c$. Then, with probability at least $1-\exp(-n/(2d))$
over the graph instances $G(n,m)$, if we choose ${\cal I}$ according
to the stationary distribution of the Metropolis process, then 
\begin{equation}\label{eq:PropArg2}
Pr[{\cal I}\notin \hat{K}]\leq \exp\left(-n/d\right),
\end{equation}
where $\hat{K}=\{k\in \mathbb{N}:|k-(1+c)\frac{\ln d}{d}n|\leq \frac{1.9n}{d}\}$.
The proposition follows from (\ref{eq:PropArg1}) and (\ref{eq:PropArg2}).

It remains to show (\ref{eq:MuLambdaCBound}). By definition we have
that for any fixed graph $G$ it holds that 
$\mu(G,\lambda)=\frac{1}{Z(G,\lambda)}\sum_{k=1}^nkR_G(k,\lambda)$,
where $Z(G,\lambda)=\sum_{k=1}^nR_G(k,\lambda)$. From Lemma
\ref{lemma:AlmostMostLikelyIS} we have that with probability at least
$1-\exp\left[-n/(2d)\right]$ over the graph instances $G(n,m)$ 
it holds that
\begin{equation}\label{eq:ZGLBoundWHP}
0\leq Z(G(n,m),\lambda_c)-\sum_{k\in \hat{K}}R_{G(n,m)}(k,\lambda_c)\leq \exp\left(-n/d\right)\left(\sum_{k\in \hat{K}}R_{G(n,m)}(k,\lambda_c)\right)
\end{equation}
and 
\begin{equation}\label{eq:ZGLKBoundWHP}
0\leq \sum_{k=0}^nk R_{G(n,m)}(k,\lambda_c) -\sum_{k\in \hat{K} }kR_{G(n,m)}(k,\lambda_c)
\leq n\exp\left(-n/(2d)\right)\left(\sum_{k\in \hat{K}}kR_{G(n,m)}(k,\lambda_c)\right).
\end{equation}
Combining (\ref{eq:ZGLBoundWHP}) and (\ref{eq:ZGLKBoundWHP})
we get that with probability at least $1-\exp\left[-n/(2d)\right]$
over $G(n,m)$ it holds that
\begin{eqnarray}
\mu(G(n,m),\lambda_c)=(1+r)\sum_{k\in \hat{K}}k\frac{R_{G(n,m)}(k,\lambda_c)}{\sum_{k\in \hat{K}}R_{G(n,m)}(k,\lambda_c)},
\nonumber
\end{eqnarray}
for some $|r|\leq 2n\exp\left(-n/(2d)\right)$. Then, it is elementary to verify that
the summation on the r.h.s. is a convex combination of values of $k$ in
$K$. That is, the summation is at most $\max\{k\in \hat{K}\}$ and at 
least $\min\{k\in \hat{K}\}$. Then (\ref{eq:MuLambdaCBound}) follows.
\end{propositionproof}

\subsection{Proof of Lemma \ref{lemma:ShateringTube}}\label{sec:lemma:ShateringTube}

As in (\ref{eqeventZdk}) let 
\begin{displaymath}	
\cZ_{d,k}=\cbc{(G,\sigma)\in\Lambda_k(n,m):|\cS_k(G)|\geq\Erw|\cS_k(G(n,m))|\cdot\exp\bc{-14n\sqrt{\ln^5d/d^3}}}.
\end{displaymath}

\begin{lemma}\label{lemma:ShateringTubeTest}
Let $(G, \sigma)\in \Lambda_k(n,m)$ be distributed as in ${\cal U}_k(n,m)$,
for $k\in K $, where $K$ and $\mu(G,\lambda)$ are as in (\ref{eq:TheK}) and 
(\ref{eq:themuglambda}), respectively.
The set $\bigcup_{k\in K}{\cal S}_k(G)$ admits a partition into classes 
${\cal C}_1, \ldots, {\cal C}_N$ such that
\begin{enumerate}
\item $Pr[\sigma \in {\cal C}_i|\cZ_{d,k}]\leq\exp[-n/(2d^{1.2})]$, for any $i \in [N]$
\item $Pr[\sigma \notin \bigcup_{i\in [N]} {\cal C}_i|\cZ_{d,k}]\leq\exp(-n/d)$
\item The distance between two independent sets in different classes is at least $2$.
\end{enumerate}
\end{lemma}

\begin{lemmaproof}{\ref{lemma:ShateringTube}}
{\bf (Given Lemma \ref{lemma:ShateringTubeTest}):}
Consider $G(n,m)$ and the Metropolis process with parameter 
$\lambda$, for $\lambda$ as in (\ref{eq:themuglambda'}).
Let the independent set ${\cal I}$ be chosen according to 
the stationary distribution of the process. 

Conditional that $|{\cal I}|=k$, ${\cal I}$ is distributed
uniformly at random in ${\cal S}_k(G(n,m))$, for any $k$.
For any $A\subset 2^{[n]}$ it holds that
\begin{eqnarray}
Pr[{\cal I}\in A |\cZ_{d,k} ]&\leq& Pr[{\cal I}\in A|\cZ_{d,k},|{\cal I}|\in K]+Pr[|{\cal I}|\notin K|\cZ_{d,k}]
\nonumber \\ 
&\leq&\displaystyle \max_{k\in K}Pr[{\cal I}\in A |\cZ_{d,k},|{\cal I}|=k]+Pr[|{\cal I}|\notin K|\cZ_{d,k}].
\nonumber
\end{eqnarray}
the last inequality follows from the fact that $Pr[{\cal I}\in A|\cZ_{d,k},|{\cal I}|\in K]$
is a convex combination of $Pr[{\cal I}\in A |\cZ_{d,k},|{\cal I}|=j]$ for $j\in K$.
Also, it holds  that
\begin{eqnarray}
Pr[|{\cal I}|\notin K|\cZ_{d,k}]&\leq &
\frac{Pr[|{\cal I}|\notin K]} {Pr[\cZ_{d,k}]}\leq 2Pr[|{\cal I}|\notin K] \qquad  \mbox{[from Proposition \ref{Lemma_gap}]} \nonumber\\
&\leq& 4\exp\left(-n/(2d^2\ln^4 d)\right) \hspace{1.6cm} \mbox{[from Proposition \ref{prop:MostLikelyIS}].}\nonumber
\end{eqnarray}
Hence,
\begin{equation}\label{eq:1896A}
Pr[{\cal I}\in A |\cZ_{d,k} ]\leq \max_{k\in K}Pr[{\cal I}\in A |\cZ_{d,k},|{\cal I}|=k]+
4\exp\left(-n/(2d^2\ln^4 d)\right).
\end{equation}
Also, from the law of total probability we get that
\begin{eqnarray}
Pr[{\cal I}\in A]&\leq& Pr[{\cal I}\in A |\cZ_{d,k}]+Pr[\cZ^c_{d,k}]
\hspace*{2.3cm}  \mbox{[$\cZ^c_{d,k}$ is the complement of $\cZ_{d,k}$]}
\nonumber \\
&\leq &  Pr[{\cal I}\in A |\cZ_{d,k}]+\exp\left(-n/(2d^2\ln^4 d)\right) 
\hspace*{2cm}\mbox{[from Proposition \ref{Lemma_gap}]} \nonumber \\
&\leq&\max_{k\in K}Pr[{\cal I}\in A |\cZ_{d,k},|{\cal I}|=k]+
5\exp\left(-n/(2d^2\ln^4 d)\right) \hspace*{1.3cm}\mbox{[from (\ref{eq:1896A})]}. \label{eq:1896Z}
\end{eqnarray}
\noindent
The statement $\mathbf C_1$  holds  from the statement 3 in Lemma
\ref{lemma:ShateringTubeTest}. Setting   $A={\cal C}_i$ in (\ref{eq:1896Z})
and using Statement 1 from Lemma \ref{lemma:ShateringTubeTest}, we
get the statement $\mathbf C_2$. Similarly, statement $\mathbf C_3$
follows by setting $A=\left(\bigcup_{k\in K}{\cal S}_k\right)\backslash
\left(\bigcup_{i\in[N]}{\cal C}_i\right)$ in (\ref{eq:1896Z})
and  using Statement 2  from Lemma \ref{lemma:ShateringTubeTest}.
\end{lemmaproof}

\subsection{Proof of Lemma \ref{lemma:ShateringTubeTest}}

Consider a uniform pair $(G,\sigma)\in \Lambda_{k}(n,m)$, for some
$k\in K$. For fixed $0<\beta<1$, and $|\gamma|<1$, let $Z_{k,\beta,
\gamma}$ be the number of independent sets $\tau\in {\cal S}_{(1+\gamma)k}
(G)$ such that $|\sigma\cap \tau|=(1-\beta)k$. Also, for $0<\beta_1<\beta_2<1$
consider  $\vec{\beta}=[\beta_1,\beta_2]$ and let the independent
set $\sigma$ be called $(\vec{\beta},\gamma,\delta)$-{\em good} if
$G$ has no independent set $\tau$ such 
\begin{itemize}
	\item $\tau\in S_{k,{\gamma}}=\bigcup_{t=(1-\gamma)\cdot k}^{(1+\gamma)k}{\cal S}_{t}(G)$
	\item $(1-\beta_2)k<|\sigma\cap \tau|<(1-\beta_1)k$
\end{itemize}
while $|\{\tau'\in S_{k,{\gamma}}:(\sigma\cap \tau')>(1-\beta_1)k\}|<	\exp\left(-\delta n\right)
|{\cal S}_k(G)|$.

\begin{lemma}\label{lemma:PsiXiExpctation}
For $\psi(x)$ is as defined in statement of Proposition \ref{Prop_calcShattering}
and $s=k/n$, it holds that 
$$\frac{1}{n}\ln E_{{\cal P}_k(n,m)}[Z_{k,\beta,\gamma}]\leq \psi(\beta)+\xi(\beta,\gamma)+o(1),$$ 
where 
\begin{displaymath}
\xi (x,y)=
s[-x \ln(1+y/x)
+y (1-\ln s-\ln(x+y))]+
\frac{d}{2}\ln \left(1-s^2\frac{2y+y^2}{1-(1+2x-x^2)s^2}\right).
\end{displaymath}
 
\end{lemma}
\begin{proof}
Let $\tau\subset V$ be such that $|\tau|=(1+\gamma)k$ and $|\sigma\cap \tau|=(1-\beta)k$.
With application of inclusion/exclusion principle we get that the total number of graphs
with $m$ edges in which $\sigma$ and $\tau$ are independent sets equals
\begin{displaymath}
{{n \choose 2}-{k\choose 2}-{(1+\gamma)k \choose 2}+ {(1-\beta)k \choose 2} \choose m}.
\end{displaymath}
Since $G$ is chosen uniformly at random among all ${{n \choose 2}-{k\choose 2} \choose m}$
graphs on $n$ vertices and $m$ edges such that $\sigma$ is an independent set, we get that
\begin{eqnarray*}
P[\tau \textrm{ is independent}]&=& 
{{n \choose 2}-{k\choose 2}-{(1+\gamma)k \choose 2}+ {(1-\beta)k \choose 2} \choose m}/
{{n \choose 2}-{k\choose 2} \choose m}\\ 
&=&
\prod_{i=0}^{m-1}\frac{{n \choose 2}-{k\choose 2}-{(1+\gamma)k \choose 2}+ {(1-\beta)k \choose 2}-i}{{n \choose 2}-{k\choose 2}-i} \\ 
&\leq &  
\left( \frac{{n \choose 2}-{k\choose 2}-{(1+\gamma)k \choose 2}+ {(1-\beta)k \choose 2}}
{{n \choose 2}-{k\choose 2}} \right)^m
\\ 
&\leq & 
\left( 1-\frac{(1+\gamma)^2k^2- (1-\beta)^2k^2 }{n^2-k^2} +O(1/n)\right)^m \\ 
&\leq & 
O(1)\cdot \left( 1-s^2\frac{(1+\gamma)^2- (1-\beta)^2 }
{1-s^2}\right)^m \qquad \mbox{[as $k=sn$].}
\end{eqnarray*}
The total number of ways to choose a set of vertices $\tau$ of 
size $(1+\gamma)k$ such that $|\sigma\cap \tau|=(1-\beta)k$ is 
equal to ${k \choose (1-\beta)k}{n-k \choose (\gamma+\beta)k}$. 
By the linearity of expectation, we get that
\begin{eqnarray}
E[Z_{k,\beta,\gamma}]&=& \displaystyle O(1)\cdot {k \choose (1-\beta)k}\cdot 
{n-k \choose (\gamma+\beta)k} \cdot \left( 1-s^2\frac{(1+\gamma)^2- (1-\beta)^2 }
{1-s^2}\right)^m \nonumber\\ 
&\leq & \displaystyle O(1)\cdot {k \choose \beta k}\cdot 
{n-k \choose (\gamma+\beta)k} \cdot \left( 1-s^2\frac{(1+\gamma)^2- (1-\beta)^2 }
{1-s^2}\right)^m \nonumber \\ 
&\leq & \displaystyle O(1)\cdot \left(\frac{e}{\beta }\right)^{\beta k} \cdot 
\left( \frac{(1-s)e}{(\gamma+\beta)s} \right)^{(\gamma+\beta)k}\cdot \left( 1-s^2\frac{(1+\gamma)^2- (1-\beta)^2 }
{1-s^2}\right)^{dn/2}\nonumber \\
&\leq & \displaystyle O(1)\cdot \left(\frac{e}{\beta }\right)^{\beta k} \cdot 
\left( \frac{e}{(\gamma+\beta)s} \right)^{(\gamma+\beta)k}\cdot \left( 1-s^2\frac{(1+\gamma)^2- (1-\beta)^2 }{1-s^2}\right)^{dn/2}.\label{eq:EXbgBound}
\end{eqnarray}
By definition (see  Proposition \ref{Prop_calcShattering}), 
it holds that 
\begin{equation}\label{eq:psibunfold}
\exp\left(\psi(\beta)n\right)=\left(\frac{e}{\beta }\right)^{\beta k}
\left( \frac{e}{\beta s} \right)^{\beta k}
\left( 1-s^2\frac{1- (1-\beta)^2 }{1-s^2}\right)^{dn/2}.
\end{equation}
Combining (\ref{eq:EXbgBound}) and (\ref{eq:psibunfold}) we get that
\begin{equation}\label{eq:ratioEZbgVsexppsib}
\frac{E[Z_{k,\beta,\gamma}]}{\exp\left(\psi(\beta) n\right)}\leq
O(1)\left( \frac{\beta}{\beta+\gamma} \right)^{\beta k}
\left( \frac{e}{(\gamma+\beta)s} \right)^{\gamma k}
\left( 1-s^2\frac{2\gamma+\gamma^2}{1-(2-(1-\beta)^2)s^2}\right)^{dn/2},
\end{equation}
since 
\begin{displaymath}
\begin{array}{c}
\left( \frac{(1-s)e}{(\gamma+\beta)s} \right)^{(\gamma+\beta)k}/
\left( \frac{(1-s)e}{(\gamma+\beta)s} \right)^{\beta k}=
\left( \frac{\beta}{\beta+\gamma} \right)^{\beta k}
\left( \frac{(1-s)e}{(\gamma+\beta)s} \right)^{\gamma k} \qquad \textrm{and}
\\ \vspace{-.3cm}\\
\left( 1-s^2\frac{(1+\gamma)^2- (1-\beta)^2 }{1-s^2}\right)^{dn/2}/
\left( 1-s^2\frac{1- (1-\beta)^2 }{1-s^2}\right)^{dn/2}=
\left( 1-s^2\frac{2\gamma+\gamma^2}{1-(2-(1-\beta)^2)s^2}\right)^{dn/2}.
\end{array}
\end{displaymath}

\noindent
Taking the logarithm and dividing by $n$ the quantities in 
(\ref{eq:ratioEZbgVsexppsib}) we get the lemma.
\end{proof}

\begin{lemma}\label{lemma:SameAsShaterringThrm}
There exist a constant $d_0>0$ and $\epsilon_d\to 0$ such that for
all $d>d_0$ the following is true:
Suppose that $s=(1+q)\ln d/d$, where $\epsilon_d\leq q\leq 1-\epsilon_d$,
then for $b=20/\ln d$  we have that
\begin{eqnarray}
\psi(b)&\leq& -18 qs \label{eq:PsiBound} \\
\sup_{x<b}\psi(x)&\leq &
-s\ln s-(1-s)\ln(1-s)+\frac{d}{2}\ln(1-s^2)-20s.
\label{eq:PsiSupBound} 
\end{eqnarray}
\end{lemma}
The lemma above states explicitly what is implied by the proof of 
Theorem \ref{theorem:shattering}. Thus, the proof of Lemma  \ref{lemma:SameAsShaterringThrm} 
is exactly the same as the one of Theorem \ref{theorem:shattering}.

%
\remove{
\begin{proof}
Let $\epsilon_d=100\ln\ln d/\ln d$. Assume that $k=(1+q)\ln d/d$ for some
$q\in [\epsilon_d,1-\epsilon_d]$. Consider the functions $\psi(x)$ 
and $\xi(x,y)$ as defined in the statement of Lemma
\ref{lemma:PsiXiExpctation}. 
Working as in the proof of Theorem \ref{theorem:shattering} (i.e. (\ref{eqpftheorem:shattering1}))
we get that
\begin{equation}\label{eq:psi(c)bound}
\psi(x)\leq sx\left(2-2\ln x-q\ln d+dsx/2\right)
\end{equation}
  Furthermore, for $x=b$ 
we get that 
\begin{eqnarray}
\psi(b)&\leq & b s\left(22+2\ln\ln d-\ln20-q\ln d\right)\leq -\frac{9}{10}bsq\ln d\leq -18qs.  \end{eqnarray}
From (\ref{eq:psi(c)bound}), for any $\beta<b$ we get that
\begin{eqnarray}
\psi(\beta)&\leq& \beta s\left(22-2\ln \beta-100\ln \ln d\right) \qquad \mbox{[as $d\beta s<40$]}\nonumber \\
&\leq& -2\beta s \ln \beta \leq s, \label{eq:psi(beta)UBound}
\end{eqnarray}
since $-x\ln x<1/2$ for every $x>0$.
By comparison, for $s=(1+q)\ln d/d$ we have
	\begin{eqnarray}
	-s\ln(s)-(1-s)\ln(1-s)+\frac{d}2\ln(1-s^2)&\geq&
		-s\ln s+s-\frac{ds^2}2-\frac{ds^4}2
			\quad\mbox{[using $\ln(1-x)\geq-x-x^2$]}\nonumber\\
		&\geq&s\bc{-\ln s-ds/2+1}\nonumber\\
		&\geq&s\bc{\frac{1-q}{2}\ln d-\ln\ln d+1}\geq 40s \ln\ln d.
		\label{eq:InverseExpctLbound}
	\end{eqnarray}
From  (\ref{eq:InverseExpctLbound}) and (\ref{eq:psi(beta)UBound}),
it is direct that (\ref{eq:PsiSupBound}) holds.
\end{proof}
}

\begin{lemma}
There  is $\epsilon_d\to 0$ such that for $(1+\epsilon_d)n\ln d/d \leq k
\leq (2-\epsilon_d)n\ln d/d$ the following is true: For $\gamma=4/\ln d$,
and $\delta=1/d^{1.2}$ 
there is $\vec{\beta}\in [0,1]^2$ such that
\begin{displaymath}
P_{{\cal U}_k(n,m)}[(G,\sigma) \textrm{ is $(\vec{\beta},\gamma,\delta)$-good}|
\cZ_{k,d}]\geq 1-\exp(-n/d).
\end{displaymath}
\end{lemma}
\begin{proof}
Let $\epsilon_d=100\ln\ln d/\ln d$. Assume that $k=(1+q)\ln d/d$ for some
$q\in [\epsilon_d,1-\epsilon_d]$. Consider the functions $\psi(x)$ 
and $\xi(x,y)$ as defined in the statement of Lemma
\ref{lemma:PsiXiExpctation}. 
In what follows take $b=\frac{20}{\ln d}$. 
Let 
$${\cal H}_k(x)=\psi(x)+\max_{(\beta,\rho)\in \mathbb{A}} \xi(\beta,\rho),$$
where $\mathbb{A}=\{(\beta,\rho)\in [0,b]\times [-\gamma,\gamma]|\beta+\rho\geq 0\}$.
Our choices for $b$ and $\gamma$ ensure that for any
$(\beta,\rho)\in \mathbb{A}$ it holds that
\begin{eqnarray}
\xi (\beta,\rho)&=& 
s[-\beta \ln(1+\rho/\beta) +\rho(1-\ln s-\ln(\beta+\rho))]+
\frac{d}{2}\ln \left(1-s^2\frac{2\rho+\rho^2}{1-(1+2\beta-\beta^2)s^2}\right)
\nonumber  
\\
&\leq& 
s[-(\beta+\rho) \ln(\beta+\rho)+\beta\ln(\beta) +\rho(1-\ln s)]-
ds^2\rho-ds^2\rho^2/2.
\nonumber  
\\
&\leq& 
s\left[25\frac{\ln\ln d}{\ln d}+\rho\left(1-\ln s -ds\right)\right]
\qquad\mbox{[$-x\ln x$ is increasing for $0<x<1/e$ and $\beta\ln \beta <0$]}
\nonumber  
\\
&\leq& 
s\left[25\frac{\ln\ln d}{\ln d} +\gamma q\ln d \right]  
\hspace*{4.8cm} \mbox{[as $s=(1+q)\ln d/d$ and $\rho\geq -\gamma $]} \nonumber \\
&<& {5}qs \hspace*{2cm}\mbox{[as $q\geq 100\ln\ln d/\ln d$].}   \label{eq:XiBound}
\end{eqnarray}

\noindent
Using (\ref{eq:XiBound}) and (\ref{eq:PsiBound}), from Lemma \ref{lemma:SameAsShaterringThrm},
we get that 
\begin{equation}\label{eq:HbUBound}
{\cal H}_k(b)\leq -13qs\leq -1300\ln\ln d/d.
\end{equation}
The function ${\cal H}_k(x)$ is continuous, therefore there exist $b_2>b_1>0$ 
and $\zeta$ such that
\begin{eqnarray}
\sup_{b_1<\beta<\beta_2}{\cal H}_k(\beta)&<&-1300\ln\ln d/d-\zeta \nonumber
\\
\sup_{b>\beta}{\cal H}_k(\beta)&<&-s\ln(s)-(1-s)\ln(1-s)+\frac{d}2\ln(1-s^2)-15s-\zeta.
\nonumber
\end{eqnarray}
The last relation follows from (\ref{eq:PsiSupBound}), of Lemma \ref{lemma:SameAsShaterringThrm}
and (\ref{eq:XiBound}).

Let $\Psi_{k,b_1,b_2}(G,\sigma)$, be the number of $\tau\in
\bigcup_{t=(1-\gamma)k}^{(1+\gamma)k}{\cal S}_t(G)$ such that
$(1-b_2)k\leq |\sigma\cap \tau|\leq (1-b_1)k$. Then,  Markov's
inequality yields
\begin{eqnarray}
P_{{\cal P}_{k}(n,m)}[\Psi_{k,b_1,b_2}>0]&\leq& E_{{\cal P}_{k}(n,m)}[\Psi_{k,b_1,b_2}]=
\sum_{i\in A }\sum_{j\in B }E_{{\cal P}_{k}(n,m)}[Z_{k, j/k ,i/k}] \nonumber 
\end{eqnarray}
where $A=[-4k/\ln d, 4k/\ln d ]$ and $B=[b_1k, b_2k]$. Using Lemma
\ref{lemma:PsiXiExpctation} we get 
\begin{eqnarray}
P_{{\cal P}_{k}(n,m)}[\Psi_{k,b_1,b_2}>0]\leq 
\exp\left[n \cdot \left(\sup_{b_2\leq \beta\leq b_1}{\cal H}(\beta)+o(1)\right)\right]\leq \exp(-10n/d).
\end{eqnarray}

\noindent
Let $\Psi_{k,b_1}(G,\sigma)$ be the number of $\tau\in
\bigcup_{t=(1-\gamma)k}^{(1+\gamma)k}S_{t}(G)$ such that $|\sigma
\cap \tau|>(1-b_1)k$. Moreover, let 
\begin{eqnarray}
\mu &=& E[|{\cal S}_k(G)|]\exp\left(-n/d^{1.2}\right)
\nonumber \\
&=& \exp\left[n\left(-s\ln s-(1-s)\ln(1-s)-\frac{d}{2}\ln(1-s^2)-n/d^{1.2}
+o(1)\right)\right]. \nonumber
\end{eqnarray}
For the derivation in the second line, see in the proof of Corollary
\ref{Cor_condexp}.
For $A'=[-4k/\ln d, 4k/\ln d]$ and $B'=[0,b_1k)$, it holds that
\begin{eqnarray}
P_{{\cal P}_k(n,m)}[\Psi_{k,b_1}>\mu] &\leq & 
\frac{E_{{\cal P}_k(n,m)}[\Psi_{k,b_1}]}{\mu} \leq 
\sum_{i\in A'}\sum_{j\in B'}\frac{E_{{\cal P}_k(n,m)}[Z_{k,j/k,i/k}]}{\mu}
\nonumber 
\\
&\leq &\frac{1}{\mu}\exp\left[ n \left(\sup_{\beta<b_1}{\cal H}(\beta)+o(1)\right) \right]
\leq \exp\left(-15n/d\right). \nonumber 
\end{eqnarray}

\noindent
The lemma follows by noting the following for $\delta=14\sqrt{\ln^5 d/d^3}$, 
\begin{eqnarray}
P_{\cU_k(n,m)}\left[(G,\sigma)\mbox{ is not $(\vec{\beta },\gamma,\delta)$-good}|\cZ_{d,k} \right] 
&\leq&  P_{\cU_k(n,m)}\left[\Psi_{k,b_1}>\mu \;\textrm{or}\;\Psi_{k,b_1,b_2}>0|\cZ_{d,k}\right] 
\nonumber \\
&\hspace*{-6cm}\leq &\hspace*{-3cm}(1-o(1))P_{\cP_k(n,m)}\left[\Psi_{k,b_1}>\mu \mbox{ or } \Psi_{k,b_1,b_2}>0|\cZ_{d,k}\right] \cdot \exp\left[ 14n\sqrt{\ln^5d/d^3}\right] 
\nonumber \\
&\hspace*{-6cm}\leq &\hspace*{-3cm}\exp(-n/d), \nonumber
\end{eqnarray}
as claimed.
\end{proof}

\noindent
Now,  Lemma \ref{lemma:ShateringTubeTest} follows from the above lemma and
by using arguments very similar to those in the proof of Proposition \ref{Prop_calcShattering}.

\end{document}